\documentclass[11pt]{article}
\usepackage{graphics}
\usepackage{amssymb}
\usepackage{amsmath}
\usepackage{enumitem}

\usepackage{fullpage}

\usepackage[ruled]{algorithm2e}

\SetAlFnt{\small}
\SetAlCapFnt{\small}
\SetAlCapNameFnt{\small}
\SetAlCapHSkip{0pt}
\IncMargin{-\parindent}

\usepackage{mathrsfs}

\newcommand{\ignore}[1]{}

\newenvironment{proof}{\noindent {\bf Proof:}}{\hfill$\Box$}

\newtheorem{theorem}{Theorem}[section]
\newtheorem{lemma}[theorem]{Lemma}

\newtheorem{corollary}[theorem]{Corollary}

\newtheorem{definition}[theorem]{Definition}
\newtheorem{remark}[theorem]{Remark}

\newtheorem{invariant}[theorem]{Invariant}

\newcommand{\hcm}[1][1]{\hspace*{#1 cm}}
\newcommand{\istrut}[2][0]{\rule[- #1 mm]{0mm}{#1 mm}\rule{0mm}{#2 mm}}
\newcommand{\rb}[2]{\raisebox{#1 mm}[0mm][0mm]{#2}}
\newcommand{\zero}[1]{\makebox[0mm][l]{$#1$}}

\newcommand{\E}{\operatorname{E}}

\newcommand{\bydef}{\stackrel{\rm def}{=}}
\newcommand{\paren}[1]{\mathopen{}\left( #1 \right)\mathclose{}}

\newcommand{\ceil}[1]{\left\lceil #1 \right\rceil}
\newcommand{\floor}[1]{\lfloor #1 \rfloor}
\newcommand{\f}[2]{\frac{#1}{#2}}
\newcommand{\fr}[2]{\mbox{$\frac{#1}{#2}$}}

\newcommand{\poly}{{\mathrm{poly}}}
\newcommand{\indeg}{\operatorname{indeg}}
\newcommand{\outdeg}{\operatorname{outdeg}}
\newcommand{\dist}{\operatorname{dist}}
\DeclareMathOperator*{\argmax}{arg\,max}

\newcommand{\Lovasz}{Lov\'{a}sz}
\newcommand{\Hanckowiak}{Ha\'{n}\'{c}kowiak}
\newcommand{\Karonski}{Karo\'{n}ski}

\newcommand{\High}{\mathcal{H}}
\newcommand{\I}{\mathcal{J}}
\newcommand{\Shell}{\mathcal{S}}

\newcommand{\hi}{\operatorname{hi}}
\newcommand{\lo}{\operatorname{lo}}
\newcommand{\degtwo}{\deg^{(2)}}

\newcommand{\bottom}{\perp}
\newcommand{\Prop}{\operatorname{prop}}

\newcommand{\LOCAL}{\mathsf{LOCAL}}
\newcommand{\CONGEST}{\mathsf{CONGEST}}
\newcommand{\Match}{\mathsf{Match}}
\newcommand{\MaximalMatch}{\mathsf{MaximalMatching}}
\newcommand{\IndependentSet}{\mathsf{IndependentSet}}
\newcommand{\MIS}{\mathsf{MIS}}
\newcommand{\TreeIndependentSet}{\mathsf{TreeIndependentSet}}
\newcommand{\TreeMIS}{\mathsf{TreeMIS}}
\newcommand{\Cluster}{\operatorname{Cluster}}
\newcommand{\ID}{\operatorname{ID}}
\newcommand{\Color}{\operatorname{Color}}
\newcommand{\OneShotColoring}{\mathsf{OneShotColoring}}
\newcommand{\degColoring}{\mbox{$(\deg+1)$-$\mathsf{Coloring}$}}
\newcommand{\Sparsify}{\mathsf{Sparsify}}
\newcommand{\RulingSet}{\mathsf{RulingSet}}

\newcommand{\Happy}{\mathscr{H}}

\newcommand{\VIB}{V_{IB}}
\newcommand{\GIB}{G_{IB}}
\newcommand{\hGammaIB}{\hat{\Gamma}_{IB}}
\newcommand{\GammaIB}{\Gamma_{IB}}
\newcommand{\degIB}{\deg_{IB}}

\newcommand{\DET}{{\sc Det.}}

\newcommand{\MaximalIndependentSet}{{\sc Maximal Independent Set}}
\newcommand{\MaximalMatching}{{\sc Maximal Matching}}
\newcommand{\RulingSetProb}{{\sc Ruling Set}}
\newcommand{\Coloring}{{\sc Coloring}}

\begin{document}


\title{The Locality of Distributed Symmetry Breaking\thanks{
A preliminary version of this paper appeared in the 
Proceedings of the 53rd IEEE Symposium on Foundations of Computer Science (FOCS), 2012.
This work is supported by NSF grants CCF-0746673, CCF-1217338, and CNS-1318294,
US-Israel Binational Science Foundation grant 2008390, and Israeli Academy of Science grant 593/11.
This research was partly performed while S. Pettie was on sabbatical at the 
Center for Massive Data Algorithmics (MADALGO), Aarhus University,
which is supported by Danish National Research Foundation grant DNRF84.
Author's addresses: Leonid Barenboim, 
Department of Mathematics and Computer Science, The Open University of Israel;
Michael Elkin, Department of Computer Science, Ben-Gurion University; 
Seth Pettie (corresponding author), Department of Electrical Engineering and Computer Science, 
University of Michigan;
Johannes Schneider, ABB Research, Z\"{u}rich.
}}

\author{Leonid Barenboim \and Michael Elkin \and Seth Pettie \and Johannes Schneider}

\maketitle

\begin{abstract}
Symmetry breaking problems are among the most well studied in the field of distributed computing
and yet the most fundamental questions about their complexity remain open.  In this paper we work in the 
$\LOCAL$ model (where the input graph and underlying distributed network are identical) and study
the {\em randomized} complexity of four fundamental symmetry breaking problems on graphs: computing MISs (maximal independent sets), 
maximal matchings,
vertex colorings, and ruling sets.  A small sample of our results includes
\begin{itemize}
\item An MIS algorithm running in $O(\log^2\Delta + 2^{O(\sqrt{\log\log n})})$ time, where $\Delta$ is the maximum degree.
This is the first MIS algorithm to improve on the 1986 algorithms of Luby and Alon, Babai, and Itai, when $\log n \ll \Delta \ll 2^{\sqrt{\log n}}$,
and comes close to the $\Omega(\log \Delta)$ lower bound of Kuhn, Moscibroda, and Wattenhofer.

\item A maximal matching algorithm running in $O(\log\Delta + \log^4\log n)$ time.  This is the first significant 
improvement to the 1986 algorithm of Israeli and Itai.
Moreover, its dependence on $\Delta$ is {\em provably optimal}.

\item A $(\Delta+1)$-coloring algorithm requiring $O(\log\Delta + 2^{O(\sqrt{\log\log n})})$ time, improving on an
$O(\log\Delta + \sqrt{\log n})$-time algorithm of Schneider and Wattenhofer.

\item A method for reducing symmetry breaking problems in low arboricity/degeneracy graphs to low degree graphs.
(Roughly speaking, the arboricity or degeneracy of a graph bounds the density of any subgraph.)
Corollaries of this reduction include an $O(\sqrt{\log n})$-time maximal matching algorithm for graphs with arboricity up to $2^{\sqrt{\log n}}$
and an $O(\log^{2/3} n)$-time MIS algorithm for graphs with arboricity up to $2^{(\log n)^{1/3}}$.
\end{itemize}

Each of our algorithms is based on a simple, but powerful technique for reducing 
a {\em randomized} symmetry breaking task to a corresponding {\em deterministic} one on a $\poly(\log n)$-size graph.
\end{abstract}








\section{Introduction}

Breaking symmetry is one of the central themes in the theory of distributed computing.  
At initialization the nodes of a distributed system
are assumed to be in the same state, possibly with distinct node IDs, yet to perform any computation the nodes 
frequently must take different roles.  That is, they must somehow break their initial symmetry.
In this paper we study several of the most fundamental symmetry breaking tasks in the 
$\LOCAL$ model~\cite{Linial92}: 
computing {\em maximal independent sets} (MIS), {\em maximal matchings}, {\em ruling sets},
and {\em vertex colorings}.  
These problems are defined below.
In the $\LOCAL$ model each node of the input graph $G$ hosts a processor, 
which is only aware of its neighbors and upper bounds on various graph parameters such as $n$ and $\Delta$, 
which are
the number of nodes and maximum degree, respectively.\footnote{This assumption can sometimes be removed. Korman, Sereni, and Viennot~\cite{KormanSV13} presented a method to convert non-uniform distributed algorithms (which know $n,\Delta,$ and possibly other parameters) into uniform distributed algorithms.}
The computation proceeds in synchronized rounds in which each processor sends one unbounded message along each edge.
{\em Time} is measured by the number of rounds; local computation is free.  
At the end of the computation each node must report its portion of the output, that is,
whether it is in the MIS or ruling set, which incident edge is part of the matching, or its assigned color.
This model should be contrasted with $\CONGEST$, which is identical to $\LOCAL$ 
except messages consist of $O(1)$ words, that is, $O(\log n)$ bits.  
Refer to Peleg~\cite[Ch.~1-2]{Peleg00} for a discussion of distributed models.
None of our algorithms seriously abuse the power of the $\LOCAL$ model.
Our message size and local computation are always 
$O(\poly(\Delta)\log n)$, usually $O(\poly(\log n))$, and in several cases $O(1)$.

Let us define the four problems formally.

\begin{description}
\item[\MaximalIndependentSet] Given $G=(V,E)$, find any set $I\subseteq V$ such that no two nodes in $I$ are adjacent and $I$ is maximal with respect to inclusion.  (That is, every $v\not\in I$ is adjacent to some member of $I$.)

\item[$(\alpha,\beta)$-\RulingSetProb] Given $G(V,E)$, find any $R\subset V$ such that for 
every $u\in V$, $\dist(u,R) \le \beta$ and for every $u\in R$, $\dist(u,R\backslash \{u\}) \ge \alpha$.
Note that $(2,1)$-ruling sets are maximal independent sets.  (Here $\dist(u,X)$ is the length of a shortest path from $u$ to any member of $X$.)

\item[\MaximalMatching] Given $G=(V,E)$, find any matching $M \subseteq E$ (consisting of node-disjoint edges) that is maximal with respect to inclusion.

\item[$K$-\Coloring] Given $G=(V,E)$, find a proper coloring $\Color \;:\; V\rightarrow \{1,\ldots,K\}$, that is, one for which $(u,v)\in E$ implies $\Color(u) \neq \Color(v)$.  We are mainly interested in $(\Delta+1)$-colorings, whose existence is trivially guaranteed.
\end{description}

We study the complexities of these problems on general graphs, as well as graphs with a specified {\em arboricity} $\lambda$.
By definition $\lambda(G)$ is the minimum number of edge-disjoint forests that cover $E$, 
which is roughly the maximum density of any subgraph.
We believe arboricity is an important graph parameter as it robustly captures the notion of 
{\em sparsity} without imposing any strict structural constraints, such as planarity or the like.
We always have $\lambda\le \Delta$, but in general $\lambda$ could be significantly smaller than $\Delta$.
Most sparse graph classes, for example, have $\lambda=O(1)$ though their maximum degree is unbounded.
These include planar graphs $(\lambda=3)$, graphs avoiding a fixed minor, bounded genus graphs, 
and graphs of bounded treewidth or pathwidth.  However, none of our algorithms actually 
depend on having $\lambda=O(1)$.

\subsection{The State of the Art in Distributed Symmetry Breaking}\label{sect:state-of-the-art}

The reader will soon notice two striking features of prior research on distributed symmetry breaking:
the wide gulf between the efficiency of deterministic and randomized algorithms
and the paltry number of algorithms that are {\em provably} optimal.
It is typical to see randomized algorithms that are {\em exponentially} faster (in terms of $n$ or $\Delta$) 
than their deterministic counterparts, and they are usually simpler to analyze and simpler to implement.
Very few problems can be solved in $O(1)$ time, independent of $\Delta$ and $n$.
The $\omega(1)$ lower bounds of Linial~\cite{Linial92}
and Kuhn, Moscibroda, and Wattenhofer~\cite{KuhnMW10}
are known to be tight in only a few cases, typically on very special classes of graphs.

We survey lower bounds and algorithms for each of the symmetry breaking problems below.
Tables~\ref{table:MIS}--\ref{table:rulingset} provide an at-a-glance history of the problems.
In the tables, deterministic algorithms are indicated by \DET{}  
All other algorithms are randomized, 
which return a correct answer with high probability.\footnote{An event occurs {\em with high probability} 
if its probability is at least $1 - n^{-c}$ for an arbitrarily large $c$,
where $c$ may influence other constants, for example, those hidden in asymptotic running times.}

\paragraph{Lower Bounds}
Linial~\cite{Linial92} proved that $\log^{(k)} n$-coloring the $n$-cycle takes $\Omega(k)$ time,
and therefore that $O(1)$-coloring the $n$-cycle takes $\Omega(\log^* n)$ time.
On the $n$-cycle, MIS, maximal matching, and ruling sets are equivalent to $O(1)$-coloring, 
so Linial's lower bound applies to these problems as well.
Kuhn, Moscibroda, and Wattenhofer~\cite{KuhnMW10} (henceforth, {\em KMW}) proved that $O(1)$-approximate
minimum vertex cover (MVC) takes $\Omega(\min\{\sqrt{\log n},\log\Delta\})$ time.
Since $2$-approximate MVC is reducible to maximal matching and maximal matching is reducible
to MIS (on the line graph of the original graph), the KMW lower bound implies $\Omega(\min\{\sqrt{\log n},\log\Delta\})$ lower bounds on these problems as well.  It does {\em not} apply to coloring problems, nor the $(\alpha,\beta)$-ruling set problem except when $(\alpha,\beta)=(2,1)$.

\paragraph{Deterministic MIS}

The fastest deterministic MIS algorithms for general graphs run in $2^{O(\sqrt{\log n})}$ time~\cite{PanconesiS96} 
and $O(\Delta + \log^* n)$ time~\cite{BarenboimEK14}.
The Panconesi-Srinivasan~\cite{PanconesiS96} result is actually a {\em network decomposition} algorithm, which can be used to solve many symmetry breaking problems in $2^{O(\sqrt{\log n})}$ time.  It improved on an earlier algorithm of Awerbuch et al.~\cite{AwerbuchGLP89} running in $2^{O(\sqrt{\log n\log\log n})}$ time.  
Recent work on deterministic MIS algorithms has focussed on restricted graph classes.  
Schneider and Wattenhofer~\cite{SchneiderW10-J} gave an optimal 
$O(\log^* n)$-time MIS algorithm for growth-bounded graphs.\footnote{A graph class has bounded growth if for each $v\in V$
and radius $r$, the maximum size of an independent set in $v$'s $r$-neighborhood is a constant depending on $r$.  
For example, unit-disc graphs have bounded growth.}
Barenboim and Elkin~\cite{BarenboimE10,BarenboimE13} gave an $O(\lambda\sqrt{\log n}+\log n)$-time MIS algorithm, 
and another that runs in $O(\frac{\log n}{\delta\log\log n})$ when the arboricity is $\lambda = (\log n)^{1/2 - \delta}$.
The subsequent vertex coloring algorithms of Barenboim and Elkin~\cite{BarenboimE11} give, as corollaries, MIS algorithms running
in $O(\lambda + \min\{\lambda^\epsilon\log n, \log^{1+\epsilon} n\})$ time and $O(\lambda^{1+\epsilon} + \log\lambda\log n)$ time, where $\epsilon>0$ influences the leading constants.  

\paragraph{Randomized MIS}
Nearly 30 years ago Luby~\cite{Luby86} and Alon, Babai, and Itai~\cite{ABI86} presented very simple randomized MIS algorithms running in $O(\log n)$ time.  These algorithms are faster than the best deterministic algorithms when $\Delta=\omega(\log n)$ and remain the fastest MIS algorithms for general graphs when running time is expressed solely as a function of $n$.  Lenzen and Wattenhofer~\cite{LenzenW11} showed that in the special case of trees ($\lambda=1$), an MIS can be computed in $O(\sqrt{\log n}\log\log n)$ time with high probability.\footnote{See footnote~\ref{fn:LW}.}

\paragraph{Deterministic Maximal Matching}
Panconesi and Srinivasan's~\cite{PanconesiS96} network decomposition algorithm implies a deterministic 
$2^{O(\sqrt{\log n})}$-time maximal matching algorithm.  This bound was dramatically improved by 
\Hanckowiak, \Karonski, and Panconesi~\cite{HanckowiakKP01} to $O(\log^4 n)$.  When $\Delta = o(\log^4 n)$, maximal matchings can be computed faster, in $O(\Delta + \log^* n)$ time, using the algorithm of Panconesi and Rizzi~\cite{PanconesiR01}.  Barenboim and Elkin~\cite{BarenboimE10,BarenboimE13} gave improved algorithms for low arboricity graphs.  Their algorithms run in $O(\lambda + \log n)$ time, for any $\lambda$, and 
in $O(\frac{\log n}{\delta\log\log n})$ time when $\lambda = \log^{1-\delta} n$.

\paragraph{Randomized Maximal Matching}
Since a maximal matching in $G$ is simply an MIS in the line graph of $G$, the randomized 
MIS algorithms of \cite{Luby86,ABI86} can be used to solve maximal matching in $O(\log n)$ time as well.\footnote{These simulations increase the local computation at each node.}
Israeli and Itai~\cite{II86} presented a direct randomized algorithm for computing maximal matchings in $O(\log n)$ time.
This algorithm is faster than the deterministic algorithms when $\Delta=\omega(\log n)$,
and remains the fastest maximal matching algorithm whose running time is expressed solely as a function of $n$.

\paragraph{Deterministic Vertex Coloring}
The vertex coloring problem allows for a tradeoff between the palette size (number of colors) and running time.
Linial~\cite{Linial92} proved that $O(\Delta^2)$-coloring can be computed deterministically 
in $O(\log^* n)$ time, independent of $\Delta$.
Szegedy and Vishwanathan~\cite{SzegedyV93} later improved the running time of this algorithm to $\frac{1}{2}\log^* n + O(1)$.
The best deterministic $(\Delta+1)$-coloring algorithms run in $2^{O(\sqrt{\log n})}$ time~\cite{PanconesiS96}
or $O(\Delta + \log^* n)$ time~\cite{BarenboimEK14}.  Even if the palette size is enlarged to $O(\Delta)$, the Panconesi-Srinivasan~\cite{PanconesiS96} algorithm remains the fastest, when time is expressed as a function of $n$.
However, Barenboim and Elkin~\cite{BarenboimE11} gave an $O(\min\{\lambda^\epsilon\log n, \lambda^\epsilon + \log^{1+\epsilon} n\})$-time algorithm for $O(\lambda)$-coloring, and an $O(\log\lambda\log n)$-time algorithm for $\lambda^{1+\epsilon}$-coloring.  (The hidden constants are exponential in $1/\epsilon$.)  
Since the arboricity $\lambda$ is at most $\Delta$, one can substitute $\Delta$ for $\lambda$ in the bounds cited above.

\paragraph{Randomized Vertex Coloring}
As usual, significantly faster coloring algorithms can be obtained using randomization.  Luby~\cite{Luby86} gave a reduction from $(\Delta+1)$-coloring to MIS, which implies an $O(\log n)$ time randomized algorithm.  A direct $O(\log n)$-time $(\Delta+1)$-coloring algorithm was analyzed by Johansson~\cite{Johansson99}.  By enlarging the palette, vertex coloring can be solved dramatically faster.  
Kothapalli et al.~\cite{KothapalliSOS06} showed that $O(\sqrt{\log n})$ time suffices for computing an $O(\Delta)$-coloring, 
for any $\Delta$.  Schneider and Wattenhofer~\cite{SchneiderW10} gave an $O(\log\Delta + \sqrt{\log n})$-time $(\Delta+1)$-coloring algorithm, for any $\Delta$, and several faster $O(\Delta)$-coloring algorithms when $\Delta$ is sufficiently large.  
For example, when $\Delta=\Omega(\log n)$, $O(\Delta)$-coloring can be computed in $O(\log\log n)$ time and
when $\Delta=\Omega(\log^{1+1/\log^* n} n)$, $O(\Delta)$-coloring can be computed in $O(\log^* n)$ time.
Kuhn and Wattenhofer~\cite{KuhnW06} showed that $O(\Delta\log n\log^{(k)} n)$-coloring is computable in $O(k)$ time and in particular, an $O(\Delta\log^2 n)$-coloring could be computed in a single round.

\paragraph{Ruling Sets}
As noted earlier, an MIS is a $(2,1)$-ruling set.  More generally, an $(\alpha,(\alpha-1)\beta)$-ruling set can be found 
by computing a $(2,\beta)$-ruling set in the graph $G^{[1,\alpha-1]}$,  whose edge set consists of pairs $(u,v)$ 
for which $\dist_G(u,v) \in [1,\alpha-1]$.  (See Section~\ref{sec:prel} for details of graph notation.)   A distributed algorithm in $G^{[1,\alpha-1]}$ can be simulated in $G$ with an $(\alpha-1)$-factor slowdown.
This reduction changes various graph parameters so it is not always applicable.
For example, $\Delta(G^{[1,\alpha-1]})$ is roughly $(\Delta(G))^{\alpha-1}$  and $\lambda(G^{[1,\alpha-1]})$ cannot be bounded
as a function of $\lambda(G)$.

Awerbuch, Goldberg, Luby, and Plotkin~\cite{AwerbuchGLP89} gave a deterministic $(2,\log n)$-ruling set algorithm running in $O(\log n)$ time.  Schneider, Elkin, and Wattenhofer~\cite{SchneiderEW13} recently discovered a $(2,\beta)$-ruling set algorithm running in $O(\beta\Delta^{2/\beta} + \log^* n)$ time, for any integer parameter $\beta$, 
and another $(2,\beta\Delta^{1/\beta})$ ruling set algorithm running in $O(\beta + \log^* n)$ time.

These are the only deterministic ruling set algorithms.  Using randomization, Gfeller and Vicari~\cite{GfellerV07} showed that a $(1,O(\log\log\Delta))$-ruling set could be computed such that the maximum degree in the graph induced by the ruling set is $O(\log^5 n)$.
Schneider and Wattenhofer~\cite{SchneiderW10} gave a randomized algorithm for computing a $(2,\beta)$-ruling set in 
$O(2^{\beta/2}\log^{2/(\beta-1)} n)$ time.  This bound was improved by Bisht, Kothapalli, and Pemmaraju~\cite{BishtKP14}
to $O(\beta\log^{1/(\beta-1)} \Delta + 2^{O(\sqrt{\log\log n})})$ time.  In earlier work, Kothapalli and Pemmaraju~\cite{KothapalliP12} gave a randomized $(2,2)$-ruling set algorithm running in $O(\log^{1/2}\Delta\cdot \log^{1/4} n)$ time and a randomized $(2,3)$-ruling set algorithm running in $\poly(\log\log n)$ time for graphs with arboricity $\lambda=O(1)$.
\begin{table*}
\centering
\begin{tabular}{|l|l|l@{\istrut[2.5]{4.5}}|}
\multicolumn{3}{c}{\sc Maximal Independent Set}\\ 
\multicolumn{3}{c}{\ }\\ 
\multicolumn{1}{l}{\sc Citation} & \multicolumn{1}{l}{\sc Running Time} & \multicolumn{1}{l}{\sc Graphs}\\\cline{1-3}
\small Linial \cite{Linial92}  & $\Omega(\log^* n)$  &  $n$-cycle  \\\cline{1-3}
\small Kuhn, Moscibroda   & \rb{-3}{$\Omega\left(\min\left\{\sqrt {\log n},\; \log\Delta\right\}\right)$} & \rb{-3}{general} \\
\small \ \& Wattenhofer \cite{KuhnMW10}  &  & \\\cline{1-3}
\small Luby~\cite{Luby86}  & \rb{-3}{$\log n$} & \rb{-3}{general} \\
\small  Alon, Babai \& Itai \cite{ABI86} & &\\\cline{1-3}
\small Panconesi \& Srinivasan \cite{PanconesiS96}  & $2^{O(\sqrt{\log n})}$  \hfill \DET & general \\\cline{1-3}
\small Barenboim, Elkin  &  & \\
\small \ \& Kuhn \cite{BarenboimEK14} & \rb{3}{$\Delta + \log^* n$}  \hfill \rb{3}{\DET} & \rb{3}{general} \\\cline{1-3}
											& $\frac{\log n}{\delta\log\log n}$ \hfill \DET & $\lambda = \log^{1/2-\delta} n$   \\\cline{2-3}
\small Barenboim \& Elkin 		 & $\lambda\sqrt{\log n}+\log n$ \hfill \DET &    \\\cline{2-2}
\ \cite{BarenboimE10,BarenboimE11} & $\lambda + \min\{\lambda^\epsilon\log n, \log^{1+\epsilon} n\}$ \hfill\DET & \rb{3}{all $\lambda$,} \\\cline{2-2}
								& $\lambda^{1+\epsilon} + \log\lambda\log n$ \hfill\DET & \rb{3}{fixed $\epsilon>0$}\\\cline{1-3}
\small Schneider  & &   \\
\small \ \& Wattenhofer \cite{SchneiderW10-J}  & \rb{3}{$\log^* n$} \hfill \rb{3}{\DET} &  \rb{3}{bounded growth} \ \\\cline{1-3}
\small Lenzen \& Wattenhofer \cite{LenzenW11}  & $\sqrt{\log n} \log \log n$ & trees ($\lambda=1$)\\\cline{1-3}
									& $\log^2\Delta + 2^{O(\sqrt{\log\log n})}$ & general \\\cline{2-3}
									& $\log^2\Delta + \f{\log\log n}{\delta\log\log\log n}$	& $\lambda = \log^{1/2-\delta} \log n\;$\\\cline{2-3}
									&  $\log^2 \lambda + \log^{2/3} n$ & all $\lambda$\\\cline{2-3}
									& $\log^2 \Delta + \lambda^{1+\epsilon} + \log\lambda\log\log n$	& \rb{-3}{all $\lambda$,}\\\cline{2-2}
\rb{0}{\small \bf{This paper}} 				& $\log^2 \Delta + \lambda + \lambda^\epsilon\log\log n$ & \rb{-3}{fixed $\epsilon>0$}\\\cline{2-2}
									& $\log^2 \Delta + \lambda + (\log\log n)^{1+\epsilon}$ 	& \\\cline{2-3}
									& $\sqrt{\log n\log\log n}$	&\\\cline{2-2}
									& $\log\Delta\log\log\Delta + \frac{\log\log n}{\log\log\log n}$   & \rb{3}{trees ($\lambda=1$)}\\\cline{2-3}
									& $\log\Delta\log\log n + 2^{O(\sqrt{\log\log n})}$ & girth $> 6$\\\cline{1-3}
\end{tabular}
\caption{\label{table:MIS}}
\end{table*}

\begin{table*}
\centering
\begin{tabular}{|l|l|l@{\istrut[2.5]{4.5}}|}
\multicolumn{3}{c}{\sc Maximal Matching}\\ 
\multicolumn{3}{c}{\ }\\ 
\multicolumn{1}{l}{\sc Citation\istrut[2]{0}} & \multicolumn{1}{l}{\sc Running Time} & \multicolumn{1}{l}{\sc Graphs}  \\     
\cline{1-3}
\small Linial \cite{Linial92}  & $\Omega(\log^* n)$ & $n$-cycle \\
\cline{1-3}
\small Kuhn, Moscibroda  & \rb{-3}{$\Omega\left(\min\left\{\sqrt {\log n},\; \log\Delta\right\}\right)$}  & \rb{-3}{general}\\
\small  \ \& Wattenhofer \cite{KuhnMW10}  &  &\\\cline{1-3}
\small Israeli \& Itai \cite{II86}  & $\log n$ & general \\
\cline{1-3}
\small \Hanckowiak, \Karonski   & $\log^4 n$ \hfill \DET & general\\\cline{2-3}
\small \& Panconesi \cite{HanckowiakKP01} & $\log^3 n$ \hfill \DET & bipartite \\\cline{1-3}
\small Panconesi \& Rizzi \cite{PanconesiR01} & $\Delta + \log^* n$ \hfill \DET & general\\
\cline{1-3}
 & $\frac{\log n}{\delta\log\log n}$ \hfill \DET & $\lambda = \log^{1-\delta} n$\\\cline{2-3}
\rb{3}{\small Barenboim \& Elkin \cite{BarenboimE10}}  & $\lambda+\log n$ \hfill \DET & all $\lambda$\\
\cline{1-3}
  & \rb{0}{$\log \Delta+\log^4\log n$}  & \rb{0}{general} \\\cline{2-3}
  & $\log\Delta + \log^3\log n$		& bipartite\\\cline{2-3}
\rb{0}{\small \bf{This paper}}  & $\log \Delta+\f{\log\log n}{\delta \log\log\log n}$ & $\lambda = \log^{1-\delta} \log n\; $ \\\cline{2-3}
 & $\log \lambda + \sqrt{\log n}$  &   \\\cline{2-2}
						& $\log\Delta + \lambda + \log\log n$		& \rb{3}{all $\lambda$}\\\cline{1-3}
\end{tabular}
\caption{\label{table:MM}}
\end{table*}

\begin{table*}
\centering
\begin{tabular}{|l|c|l@{\istrut[2]{4}\hcm[.1]}|}
\multicolumn{3}{c}{\sc Vertex Coloring}\\ 
\multicolumn{3}{c}{\ }\\ 
\multicolumn{1}{l}{\sc Citation}
& \multicolumn{1}{c}{\sc Colors}
& \multicolumn{1}{l}{\sc Running Time}\\\cline{1-3}
\small Linial~\cite{Linial92}							& 	3				& $\Omega(\log^*n)$                       \\\cline{1-1}\cline{3-3}
\small Cole \& Vishkin~\cite{ColeV86}					&	{\small (on the $n$-cycle)}			& $\log^* n$ {\small $+ O(1)$}\hfill \DET	\\\cline{1-3}
\small Luby~\cite{Luby86}								&						& \rb{-3}{$\log n$}	 \\\cline{1-1}
\small Johansson~\cite{Johansson99}					&						& \\\cline{1-1}\cline{3-3}
\small Panconesi \& Srinivasan~\cite{PanconesiS96}			& 						& $2^{O(\sqrt{\log n})}$   \hfill \DET                                       \\\cline{1-1}\cline{3-3}
\small Barenboim, Elkin \& Kuhn~\cite{BarenboimEK14}		&						& $\Delta + \log^* n$ \hfill \DET    \\\cline{1-1}\cline{3-3}
\small Schneider \& Wattenhofer~\cite{SchneiderW10}	&	$\Delta+1$		& $\log\Delta + \sqrt{\log n}$                                 \\\cline{1-1}\cline{3-3}
				& 						& $\log\Delta + 2^{O(\sqrt{\log\log n})}$                  \\\cline{3-3}
				& 						& $\log\Delta + \lambda^{1+\epsilon}+ \log\lambda\log\log n$   \\\cline{3-3}
				& 						& $\log\Delta + \lambda+ \lambda^\epsilon\log\log n$    \\\cline{3-3}
								&						& $\log\Delta + \lambda+ (\log\log n)^{1+\epsilon}$\\\cline{2-3}
\rb{3}{{\bf\small This paper}}				& \rb{-3}{$\Delta + O(\lambda)$} & $\log\Delta + \lambda^\epsilon \log\log n$             \\\cline{3-3}
				&					& $\log\Delta + \lambda^\epsilon + (\log\log n)^{1+\epsilon}$\\\cline{2-3}
				& $\Delta + \lambda^{1+\epsilon}$ & $\log\Delta + \log\lambda\log\log n$             \\\cline{2-2}\cline{3-3}
				& 						& $2^{O(\sqrt{\log\log n})}$					\\\cline{1-1}\cline{3-3}
\small Kothapalli, Scheideler, Onus 	& 	\rb{-3}{$O(\Delta)$}			& \rb{-3}{$\sqrt{\log n}$}                                     \\
\small \ \& Schindelhauer~\cite{KothapalliSOS06}						&						&				 \\\cline{1-1}\cline{3-3}
			&						& $\min\{\Delta^{\epsilon}\log n, \Delta^\epsilon + \log^{1+\epsilon} n\}$   \hfill \DET \\\cline{2-3}
\rb{-3}{\small Barenboim \& Elkin~\cite{BarenboimE11}}		& $O(\lambda)$	& $\min\{\lambda^\epsilon\log n,\lambda^\epsilon + \log^{1+\epsilon} n\}$	\hfill\DET	\\\cline{2-3}
											& $\Delta^{1+\epsilon}$		& $\log\Delta\log n$  \hfill \DET      \\\cline{2-3}
											& $\lambda^{1+\epsilon}$		& $\log\lambda\log n$ \hfill \DET						\\\cline{1-3}
\cline{1-2}\cline{3-3}
\rb{-3}{\small Schneider}							& $O(\Delta + \log n)$					& $\log\log n$	\\\cline{2-3}
\rb{-3}{\small \ \& Wattenhofer~\cite{SchneiderW10}}	& \ \ $\Delta\log^{(k)} n$\hfill { }		& \rb{-3}{$k$}\hfill\rb{-3}{(for $k<\log^* n$)}                  \\
										& \hfill $+ \log^{1+1/k} n$ {\ \ }					&\\\cline{1-3}
\small Kuhn \& Wattenhofer~\cite{KuhnW06}	& $\Delta\log n\log^{(k)} n$\hfill { } 	& $k$\hfill (for $k<\log^* n$) \\\cline{1-3}
\small Linial~\cite{Linial92}					& \rb{-3}{$O(\Delta^2)$}			& $\log^* n$ {\small $+ O(1)$}   \hfill \DET\\\cline{1-1}\cline{3-3}
\small Szegedy \& Vishwanathan \cite{SzegedyV93} &					& $\frac{1}{2}\log^* n$ {\small $+ O(1)$} \hfill \DET\\\cline{1-3}	
\small Barenboim \& Elkin~\cite{BarenboimE10}			& \rb{-3}{$\lambda\cdot n^{1/k}$}	& $\Omega(k)$\\\cline{1-1}\cline{3-3}
\small Kothapalli \& Pemmaraju~\cite{KothapalliP11}		& 							& $k$	\hfill (for $\log\log n < k < \sqrt{\log n}$)\\\cline{1-3}
\end{tabular}
\caption{
\label{table:coloring}
}
\end{table*}

\begin{table*}
\centering
\begin{tabular}{|l|c|l@{\istrut[2]{4}\hcm[.1]}|}
\multicolumn{3}{c}{\sc Ruling Sets}\\ 
\multicolumn{3}{c}{\ }\\ 
\multicolumn{1}{l}{\sc Citation}
& \multicolumn{1}{c}{$(\alpha,\beta)$}
& \multicolumn{1}{l}{\sc Running Time}\\\cline{1-3}
 									& $(2,1)$					&	\mbox{\sc mis} time\\\cline{2-3}
\small trivial						& \rb{-3}{$(\alpha,(\alpha-1)\beta)$}	&	$\alpha\cdot(2,\beta)$-{\sc ruling set} time\\
									&						&	\hfill\small (see text, \S~\ref{sect:state-of-the-art})\\\cline{1-3}
\small Awerbuch, Goldberg,						&								&					\\
\small Luby \& Plotkin~\cite{AwerbuchGLP89}		& \rb{3}{$(2,\log n)$}			&	\rb{3}{$\log n$}\hfill \rb{3}{\DET}		\\\cline{1-3}
\small Gfeller \& Vicari~\cite{GfellerV07}				& $(1,O(\log\log \Delta))$			&	$\log\log \Delta$ \hfill {\small (see text, \S~\ref{sect:state-of-the-art})}\\\cline{1-3}
\small Schneider \& Wattenhofer~\cite{SchneiderW10}	& $(2,\beta)$						&	$2^{\beta/2} \log^{\frac{2}{\beta-1}} n$\\\cline{1-3}
											& $(2,2)$								&	$(\log^{1/2}\Delta)(\log^{1/4} n)$\\\cline{2-3}
\small Kothapalli \& Pemmaraju~\cite{KothapalliP12}	& $(2,3)$ \ \ {\small ($\lambda=1$)} 					&	$(\log\log n)^2\log\log\log n$\\\cline{2-3}
										& $(2,3)$ \ \ {\small ($\lambda=O(1)$)}			&	$(\log\log n)^3$\\\cline{1-3}
\rb{-3}{\small Schneider, Elkin \& Wattenhofer~\cite{SchneiderEW13}}		& $(2,\beta\Delta^{1/\beta})$		& $\beta+\log^* n$\hfill\DET\\\cline{2-3}
									& $(2,\beta)$				& $\beta\Delta^{2/\beta} + \log^*n$\hfill\DET\\\cline{1-3}
\small Schneider, Elkin \& Wattenhofer~\cite{SchneiderEW13}	&						&\\
\small \ \ + Gfeller \& Vicari~\cite{GfellerV07}						& \rb{3}{$(2,O(\log\log n))$}	&	\rb{3}{$\log\log n$}\\\cline{1-3}							
\small Barenboim \& Elkin~\cite{BarenboimE10}  & \rb{-3}{$(2,\log\lambda+\sqrt{\log n})$}	& \rb{-3}{$\log\lambda + \sqrt{\log n}$}\hfill \rb{-3}{\DET}\\
\small \ \ + Awerbuch et al.~\cite{AwerbuchGLP89} &&\\\cline{1-3}
\small Bisht, Kothap. \& Pemmaraju~\cite{BishtKP14}		& $(2,\beta)$				&	$\beta \log^{\frac{1}{\beta-1}} \Delta + 2^{O(\sqrt{\log\log n})}$\\\cline{1-3}
{\bf\small This paper}						& $(2,\beta)$						&	$\beta \log^{\frac{1}{\beta-1/2}} \Delta + 2^{O(\sqrt{\log\log n})}$\\\cline{1-3}
\end{tabular}
\caption{\label{table:rulingset}}
\end{table*}

\subsection{The Union Bound Barrier}\label{sect:union-bound}

Our algorithms confront a fundamental barrier in randomized distributed algorithms 
we call the {\em union bound barrier}, which, 
to our knowledge, has never been explicitly discussed. 

Consider a generic symmetry breaking algorithm that works as follows.  The nodes execute
some number of iterations of an 
$O(1)$-time randomized experiment, 
the purpose of which is to commit to some fragment of the output.
That is, some nodes are committed to the MIS or ruling set, 
some edges are committed to the matching, some nodes commit to a color, etc. 

The experiment {\em fails} at each node $v$ with probability $1-\Omega(1)$.
For example, failure may be defined as the event that no edge incident to $v$ joins the matching.
The failure events are not independent in general,
but are independent for sufficiently distant nodes.  If the random experiment takes $t$ time steps, nodes at distance at least $2t+1$ are influenced by disjoint sets of nodes.
Although each node succeeds after $\Theta(1)$ time in expectation,
the union bound only lets us claim that a {\em global} solution is reached with probability $1-n^{-\Omega(1)}$ after $\Theta(\log n)$ time.  Symmetry breaking algorithms based on a random experiment with failure probability $p$ seem intrinsically incapable of running in $o(\log_{1/p} n)$ time.\footnote{Moreover, existing randomized algorithms~\cite{Luby86,ABI86,II86} do not even fit in this framework.  They do {\em not} guarantee each node succeeds with probability $\Omega(1)$, only that an $\Omega(1)$-fraction of the edges are incident to nodes that
succeed with probability $\Omega(1)$.}
However, there are several conceivable strategies one could use to escape this conclusion.  Among them,

\begin{description}
\item[Use no randomness]  Deterministic algorithms have no probability of failure.  
\item[Redefine failure]  If the experiment is kept the same but the notion of failure is relaxed 
such that it only occurs with probability $n^{-\Omega(1)}$, the union bound {\em can} be applied.
\end{description}

We borrow an idea used in early constructive algorithms for the \Lovasz{} Local Lemma~\cite{Beck91,Alon91} and more recently by~\cite{RubinfeldTVX11}, which combines elements from both of the strategies above.

All of our algorithms consists of two discrete phases.
In Phase I we execute $O(\log\Delta)$ or $\poly(\log\Delta)$ iterations (rather than $\Theta(\log n)$) of an experiment whose {\em local} probability of failure is $1-\Omega(1)$.  Using the fact that failure events are independent for sufficiently distant nodes, we show that {\em every} connected component in the remaining graph\footnote{That is, the portion not dominated by the independent set (in the case of MIS), or not adjacent to a matched edge (in the case of maximal matching), etc.} has size $s = \poly(\log n)$, or in one case $s = \poly(\Delta)\log n$, 
with probability $1 - n^{-\Omega(1)}$.  

In Phase II we revert to the best available {\em deterministic} algorithm and apply it to each connected component, letting it run for time sufficient to solve any instance on $s$ nodes.  (If there is a component with more than $s$ nodes, this is a {\em global} failure, which occurs with probability $n^{-\Omega(1)}$.)

This two-phase structure explains some conspicuous features of our results listed in Tables~\ref{table:MIS}-\ref{table:rulingset}.
The runnings times are always expressed as two (or more) terms, one that usually depends on 
$\log\Delta$ and another that exactly matches the time bound of one of the deterministic algorithms, 
except that it is {\em scaled down exponentially}.  In other words, $2^{\sqrt{\log n}}$ becomes $2^{\sqrt{\log\log n}}$, 
$\frac{\log n}{\log\log n}$ becomes $\frac{\log\log n}{\log\log\log n}$, and so on.  

The two-phase strategy is {\em one} way around the union bound barrier, but is it the only one?  More to the point, is it true that the randomized complexities of certain problems (MIS, maximal matching, etc.) are at least their deterministic complexities on $\poly(\log n)$-size instances?  We have no theorem to this effect, but it is easy to see that it is true for algorithms using a limited number of random bits, as we show below.  
We are not aware of any randomized symmetry breaking algorithms that do {\em not} use limited random bits.

Consider a happy situation where Phase I is completely free, that is, the input graph happens to be the union of $n/\log^\epsilon n$ identical subgraphs of size $\log^\epsilon n$.
These subgraphs are worst-case instances for whatever algorithm is used.
We can assume the algorithm runs in at most $O(\log^\epsilon n)$ time since the diameter of each component 
is at most $\log^\epsilon n$.
If, in each time step, node $v$ generates at most $(\deg(v))^\delta$ random bits, for some $\delta=O(1)$,
a component will generate at most $\log^{\epsilon(\delta+2)} n$ random bits in total.  For $\epsilon < (\delta +2)^{-1}$, {\em every} string of random bits will be generated with high probability, so if 
the algorithm errs on a component with {\em any} non-zero probability it must err on {\em some} component with probability close to 1.  On the other hand, if the algorithm errs with zero probability, 
we might as well commit to the all-zero string of `random' bits and make it deterministic.

\subsection{New Results}

We introduce numerous symmetry breaking algorithms using the two-phase strategy outlined in Section~\ref{sect:union-bound}.  For Phase I we design new iterated randomized experiments and analyze their 
local probability of failure.  After Phase I the connected components in the surviving subgraph have size 
$\poly(\log n)$ or $\poly(\Delta)\log n$ with high probability.
For Phase II we invoke the best available deterministic algorithm, usually applied in a black-box fashion.
For general graphs there always happens to be one best deterministic algorithm.  
However, for low arboricity graphs we have access to several algorithms, 
each of which is asymptotically superior for different values of $\lambda,\Delta,$ and $n$.

For graphs with a large disparity between $\lambda$ and $\Delta$ the method 
described above does not get optimal results.  We give a general randomized reduction showing that MIS and maximal matching are reducible in $O(\log^{1-\gamma} n)$ time to instances with maximum degree $\lambda\cdot 2^{\log^\gamma n}$, for any $\gamma\in (0,1)$.  
This reduction allows us to obtain algorithms whose running time is {\em sub}logarithmic in $n$, 
given algorithms that run in time polylogarithmic in $\Delta$.

We shall now discuss the results claimed in Tables~\ref{table:MIS}-\ref{table:rulingset}.

\paragraph{MIS and Ruling Sets}
Our primary result is a new MIS algorithm running in $O(\log^2 \Delta + 2^{O(\sqrt{\log\log n})})$ time,
which is within a $\log\Delta$ factor of the KMW lower bound.  Moreover, this is the {\em first} improvement to the 1986 algorithms of Luby~\cite{Luby86} and Alon, Babai, and Itai~\cite{ABI86} for such a broad range of degrees: from 
$\Delta = \Omega(\log n)$ to $2^{O(\sqrt{\log n})}$.
The Phase II portion of this algorithm is rather complicated since we cannot afford to apply an existing MIS algorithm in a black box fashion.  After Phase I the surviving components are shown to have size $\poly(\Delta)\log n$.  By invoking the Panconesi-Srinivasan~\cite{PanconesiS96} algorithm on each component, Phase II would run in $2^{O(\sqrt{\log(\poly(\Delta)\log n)})}$ time, which is fine if $\Delta = \poly(\log n)$ but not if $\Delta$ is just slightly super-logarithmic.  We prove that by a certain deterministic clustering procedure, each component can be decomposed into $\log n$ clusters with diameter $O(\log \Delta)$.  A version of the Panconesi-Srinivasan~\cite{PanconesiS96} algorithm can then be simulated on the cluster graph formed by virtually contracting each cluster to a single node.

Using our degree-reduction routine, we can solve MIS on graphs with arboricity $\lambda$ in \\
$O\paren{\log^{1-\gamma} n + \log^2(\lambda\cdot 2^{\log^\gamma n}) + 2^{O(\sqrt{\log\log n})}}$ time, 
which simplifies to $O(\log^2 \lambda + \log^{2/3} n)$ when $\gamma = 1/3$.  
Other MIS algorithms that depend at least linearly on $\lambda$ can be generated 
by invoking one of the MIS algorithms of Barenboim and Elkin~\cite{BarenboimE11}.

Finally, we give an $O\paren{\log\Delta\log\log\Delta + \frac{\log\log n}{\log\log\log n}}$-time algorithm for 
MIS on trees ($\lambda=1$), which, using the degree-reduction routine with $\gamma = 1/2-o(1)$, implies a time bound of 
$O\paren{\sqrt{\log n\log\log n}}$, independent of $\Delta$.\footnote{\label{fn:LW}Lenzen and Wattenhofer~\cite{LenzenW11} claimed 
an MIS algorithm running in 
$O(\sqrt{\log n\log\log n})$ time on trees, but there is a flaw in their analysis. 
We repair this flaw in Section~\ref{sect:TreeMIS}.  
By incorporating Lemma~\ref{lem:TreeMIS-badnode} into the proof of 
\cite[Lemma 4.8]{LenzenW11}, the resulting algorithm would only run in $O(\sqrt{\log n}\log\log n)$ time.}
With minor modifications, this algorithm can be made to work on general 
graphs with girth greater than 6, not just trees.  The {\em girth} of a graph is the length of its shortest cycle.

Bisht, Kothapalli, and Pemmaraju~\cite{BishtKP14} showed how to reduce the 
$(2,\beta)$-ruling set problem on degree-$\Delta$ graphs to an MIS problem on 
graphs with degree much smaller than $\Delta$.  
Using their reduction and our new MIS algorithm, 
we get a  
$(2,\beta)$-ruling set algorithm running in 
$O\paren{\beta \log^{\frac{1}{\beta-1/2}} \Delta + 2^{O(\sqrt{\log\log n})}}$ time.
This result is notable because it establishes a provable gap between the complexity of computing an 
MIS (a $(2,1)$-ruling set) and a $(2,2)$-ruling set.  
By the KMW bound, an MIS cannot be computed in $o(\log\Delta)$ time
whereas $(2,2)$-ruling sets can be computed in $O(\log^{2/3} \Delta + 2^{O(\sqrt{\log\log n})})$ time.\footnote{When time bounds are expressed in terms of $n$ (rather than $\Delta$), our result only 
demonstrates that $(2,3)$-ruling sets are easier to compute than MISs.  They can be computed in
$O(\log^{2/5} \Delta + 2^{O(\sqrt{\log\log n})}) = O(\log^{2/5} n)$ time whereas MISs need $\Omega(\sqrt{\log n})$ time~\cite{KuhnMW10}.}

\paragraph{Maximal Matching}
We give a new maximal matching algorithm running in 
$O(\log\Delta + \log^4\log n)$ time using $O(1)$-size messages, that is, it works in the $\CONGEST$ model.  
In some ways this is our strongest result.  By the KMW bound
its dependence on $\Delta$ is optimal and for 
$\log\Delta \in [\log^4\log n,\, \sqrt{\log n}]$ it cannot be improved asymptotically.
The result is one of only a handful of provably optimal symmetry breaking algorithms for general graphs.\footnote{Other sharp bounds include 
(i) $\Theta(\log^* n)$ time for MIS/maximal matching/$O(\Delta)$-coloring, but only when $\Delta = O(\log^* n)$, 
(ii) computing an MIS in growth-bounded graphs,
in $\Theta(\log^* n)$ time~\cite{SchneiderW10}, 
(iii) $O(\lambda\cdot n^{1/k})$-coloring graphs in $\Theta(k)$ time~\cite{BarenboimE10,KothapalliP12}, for a certain range of $k$, and
(iv) $O(1)$-approximate minimum vertex cover in $\Theta(\log\Delta)$ time~\cite{KuhnMW10}.
With the exception of (iv), these algorithms only apply to narrow classes of graphs.}
Using the degree-reduction routine with $\gamma=1/2$, 
we obtain a maximal matching algorithm running in $O(\log\lambda + \sqrt{\log n})$ time.  
Since the KMW graphs have arboricity $\lambda = 2^{\Theta(\sqrt{\log n})}$, this algorithm is provably optimal for that particular 
arboricity.  Generalizing the KMW lower bound, we prove that {\em even on trees}, maximal matching
requires $\Omega(\sqrt{\log n})$ time.  Thus, our algorithm is optimal for all $\lambda$ from 1 to $2^{O(\sqrt{\log n})}$.
Using the Barenboim-Elkin~\cite{BarenboimE10,BarenboimE13} maximal matching algorithm we obtain more results 
that are superior when $\lambda$ is small and $\log\Delta = o(\sqrt{\log n})$.
For example, when $\lambda=O(1)$, a maximal matching can also be computed
in $O(\log\Delta + \frac{\log\log n}{\log\log\log n})$ time.

\paragraph{Vertex Coloring}
The vertex coloring problem is, in one respect, qualitatively different than maximal matching and MIS.
In Phase II of the MIS and matching algorithms, each connected component forms a (small)
instance of MIS or maximal matching.  However, in our vertex coloring algorithms, at the beginning
of Phase II some nodes have been permanently colored, which affects the palettes of their as-yet
uncolored neighbors. Thus, the connected components of uncolored nodes form instances of
the {\em list-coloring} problem---each vertex may hold a palette of an arbitrary set of allowable colors.  
This distinction sometimes makes no difference.

Our main coloring result is a $(\Delta+1)$-coloring algorithm running in $O(\log\Delta + 2^{O(\sqrt{\log\log n})})$ time,\footnote{The algorithm actually solves the list-coloring problem, 
where a vertex $v$'s  palette contains $\deg(v)+1$ colors.} 
which improves the $O(\log\Delta + \sqrt{\log n})$ bound of 
Schneider and Wattenhofer~\cite{SchneiderW10} 
and implies that $O(\Delta)$-coloring can be computed in $2^{O(\sqrt{\log\log n})}$ time, independent of 
$\Delta$.
The KMW lower bound does not apply to vertex coloring, so we do not know if the dependence on $\Delta$ is optimal.
So long as the Panconesi-Srinivasan algorithm goes unimproved, it will be difficult or impossible
to improve the dependence on $n$.

By invoking the Barenboim-Elkin~\cite{BarenboimE10,BarenboimE11,BarenboimE13} coloring algorithms we obtain numerous results for graphs with small arboricity. 
Since the Barenboim-Elkin algorithms do {\em not} solve the general list-coloring problem, we have to start 
Phase II with a ``fresh'' palette of unused colors.  
This fact leads to $(\Delta + \Omega(\lambda))$-coloring algorithms whose running time is sublinear in $\lambda$,
and $(\Delta+1)$-coloring algorithms whose running time is at least linear in $\lambda$.
Elkin, Pettie, and Su~\cite{ElkinPS15} recently considered randomized distributed algorithms for coloring locally sparse graphs.
One consequence of their results is that $(\Delta+1)$-coloring can be computed in 
$O(\log\lambda) + 2^{O(\sqrt{\log\log n})}$ time for all $\lambda,\Delta,n$, and in $O(\log^* n)$ time for certain ranges of the parameters.

\subsection{Organization}

In Section~\ref{sec:prel} we review some notation for graphs and their parameters,
as well as some useful symmetry breaking primitives due to 
Awerbuch et al.~\cite{AwerbuchGLP89} and Panconesi and Srinivasan~\cite{PanconesiS96}.
Sections~\ref{sect:MIS}--\ref{sect:ruling-sets} are devoted to algorithms for the four symmetry breaking problems on {\em general} graphs.
In Section~\ref{sect:arb} we present a new degree-reduction method (parameterized by the arboricity) 
and derive numerous results for small arboricity graphs.
Section~\ref{sect:TreeMIS} presents a faster algorithm for MIS on trees and graphs of girth greater than 6.
We conclude and discuss some open problems in Section~\ref{sect:conclusion}.

In our analyses we use several standard concentration inequalities due to 
Chernoff, Janson, and Azuma-Hoeffding.  The statements of these theorems can be found in 
Appendix~\ref{sec:concentration-bounds}.  
Refer to Dubhashi and Panconesi~\cite{DubhashiPanconesi09} for derivations of 
these and other concentration bounds.

\section{Preliminaries}\label{sec:prel}

\subsection{Graph Notation}

Let $G=(V,E)$ be the undirected input graph and underlying distributed network.  
Define $\Gamma_H(v), \hat\Gamma_H(v)$, and $\deg_H(v)$ to be the neighborhood, inclusive neighborhood, and degree of $v$ with respect to a graph $H$.  Typically $H$ is an induced subgraph of $G$.  Formally,
\begin{align*}
\Gamma_H(v) &\bydef \{u \:|\: (v,u)\in E(H)\},\\
\hat\Gamma_H(v) &\bydef \{v\}\cup \Gamma_H(v),\\
\mbox{and } \; \deg_H(v) &\bydef |\Gamma_H(v)|.
\end{align*}
For succinctness we sometimes put $U\subseteq V(G)$ or $U\subseteq E(G)$ in the subscript
to refer to the subgraph of $G$ induced by $U$.
The subscript may be omitted altogether if $H=G$.

We assume the nodes know global graph parameters\footnote{This assumption can sometimes be removed.  See~\cite{KormanSV13}.} such as $n \bydef |V(G)|$, 
$\Delta \bydef \max_{v\in V}\deg_G(v)$, and, if applicable, the arboricity $\lambda(G)$.
To simplify calculations we often assume $n,\Delta,$ and $\lambda$ are at least some sufficiently large constant.
The arboricity of a graph $H$ is the minimum number of forests that cover $E(H)$.
By the Nash-Williams~\cite{NW64} theorem, $\lambda(H)$ can also be defined as
\begin{align*}
\lambda(H) &\bydef \max \left\{\left.\ceil{\frac{\left|E(H) \cap {U \choose 2}\right|}{|U| - 1} }  \;\; \right| \;\; U\subseteq V(H) \mbox{ and } |U|\ge 2 \right \},
\intertext{that is, roughly the edge-density of any subgraph of $H$ with at least 2 nodes. 
Other measures of graph sparsity are, for our purposes, 
equivalent to $\lambda$.  For example, the {\em degeneracy} of a graph $H$ is defined to be}
d(H) &\bydef \max_{U\subseteq V(H)} \min_{v\in U} \: \deg_{U}(v).
\end{align*}
It is known that $\lambda(H) \le d(H) \le 2\lambda(H)-1$.

Our matching algorithms internally generate {\em directed} graphs.  In a directed graph $H$, 
the {\em indegree} and {\em outdegree} of $v$ (written $\indeg_H(v)$ and $\outdeg_H(v)$) are the number
of edges oriented towards $v$ and away from $v$, respectively, and $\deg_H(v) \bydef \indeg_H(v) + \outdeg_H(v)$.  
A {\em pseudoforest} is a directed graph in which all nodes have outdegree at most 1.

Let $\dist_H(u,v)$ be the distance (length of the shortest path) between $u$ and $v$ in $H$. 
For any integers $1\le a \le b$, define
\begin{align*}
H^{[a,b]} &\bydef \paren{V(H), \{(u,v) \;|\; \dist_H(u,v) \in [a,b]\}}\\
\mbox{and } \; H^{a} &\bydef H^{[a,a]}.
\end{align*}

In other words, we put edges between pairs whose distance is in the interval $[a,b]$.

\subsection{Decompositions and Ruling Sets}\label{sect:decomp-rulingset}

A network decomposition is a powerful tool
used in symmetry breaking algorithms.  The fastest known deterministic decomposition algorithm
is due to Panconesi and Srinivasan~\cite{PanconesiS96}.  See~\cite{AwerbuchGLP89,LinialS93}
for earlier decomposition algorithms.

\begin{definition} (Network Decompositions)
Let $H$ be an $n$-vertex graph.  A $(d(n),c(n))$-network decomposition is a 
pair $(\mathscr{D},\mathscr{C})$ such that $\mathscr{D}$ is a partition of $V(H)$ into clusters, each with diameter at most $d(n)$,
and $\mathscr{C} \,:\, \mathscr{D} \rightarrow \{1,\ldots,c(n)\}$ is a proper $c(n)$-coloring of the graph derived by contracting the clusters.
More formally, we have
$\mathscr{D} = \{D_i\}$, where $\bigcup_i D_i = V(H)$,  $D_i \cap D_{i'} = \emptyset$ for $i\neq i'$, 
and if $v,v'\in D_i$ then $\dist_{D_i}(v,v')\le d(n)$.  
If there exists $(v,v')\in E(H)$ with $v\in D_i$ and $v'\in D_{i'}$ then $\mathscr{C}(D_i)\neq \mathscr{C}(D_{i'})$.
\end{definition}

\begin{theorem}\label{thm:PS96} (Panconesi and Srinivasan~\cite{PanconesiS96})
A $\paren{2^{O(\sqrt{\log n})}, 2^{O(\sqrt{\log n})}}$-network decomposition can be 
computed deterministically in 
$2^{O(\sqrt{\log n})}$ time.
\end{theorem}

Definition~\ref{def:rulingset} and Theorem~\ref{thm:AGLP89} generalize, 
slightly, Awerbuch et al.'s~\cite{AwerbuchGLP89} original 
definition of a ruling set.

\begin{definition}\label{def:rulingset} (Ruling Sets)
Let $H$ be a graph and $U\subseteq V(H)$.
An $(\alpha,\beta)$-ruling set for $U$ (w.r.t.~$H$) is a node set $R\subseteq U$ such that for each $v\in U$,
$\dist_H(v,R) \le \beta$ and, if $v\in R$, $\dist_H(v,R\backslash\{v\}) \ge \alpha$.  
For example, maximal independent sets are $(2,1)$-ruling sets for $V(H)$ with respect to $H$.
\end{definition}

\begin{theorem}\label{thm:AGLP89} (Awerbuch, Goldberg, Luby, and Plotkin~\cite{AwerbuchGLP89})
Let $H$ be a graph and $U\subseteq V(H)$.
Given a proper $K$-coloring of $H^{[1,\alpha-1]}$, 
an $(\alpha,(\alpha-1)\ceil{\log K})$-ruling set for $U$
can be computed in $(\alpha-1)\ceil{\log K}$ time.
\end{theorem}

\begin{proof}
Let $\chi : V\rightarrow \{1,\ldots,K\}$ be the coloring.  
Recursively, and in parallel, compute two $(\alpha,\alpha(\ceil{\log K} - 1))$-ruling sets $R_0$ and $R_1$
for, respectively,
\begin{align*}
U_0 &= \{v\in U \:|\: \chi(v) \in \{1,\ldots,\floor{K/2}\}\}\\
\mbox{and } \; U_1 &= \{v\in U \:|\: \chi(v) \in \{\floor{K/2}+1,\ldots,K\}\}.
\intertext{After $R_0$ and $R_1$ are computed, return the $(\alpha,(\alpha-1)(\ceil{\log K}))$-ruling set $R$, where}
R &= R_0 \cup \{v \in R_1 \:|\: \dist_H(v,R_0) \ge \alpha\}.
\end{align*}
That is, each $R_0$ node ``knocks out'' all $R_1$ nodes within distance $\alpha-1$.
Once $R_0$ and $R_1$ are computed, in $(\alpha-1)(\ceil{\log K}-1)$ time,
$R$ can be computed in $\alpha-1$ additional time.
\end{proof}

If the nodes of $H$ are endowed with distinct $\beta$-bit IDs, we can use them as a proper $2^{\beta}$-coloring
and compute an $(\alpha,(\alpha-1)\beta)$-ruling set in $O((\alpha-1)\beta)$ time.
(This was Awerbuch et al.'s~\cite{AwerbuchGLP89} original algorithm.)
However, a better bound can be obtained by first computing a good coloring.

\begin{corollary}\label{cor:AGLP}
Let $H$ be a graph with maximum degree $\Delta$ whose nodes are assigned distinct $\beta$-bit IDs.
For any $\alpha\ge 2$ and $U\subseteq V(H)$, an $(\alpha,2(\alpha-1)^2(\log\Delta+O(1)))$-ruling set for $U$ 
with respect to $H$ can be computed in $O(\alpha\log^*\beta + \alpha^2\log\Delta)$ time.
\end{corollary}

\begin{proof}
The graph $H^{[1,\alpha-1]}$ has maximum degree less than $\hat\Delta \bydef \Delta^{\alpha-1}$. 
The first step is to $O(\hat\Delta^2)$-color $H^{[1,\alpha-1]}$ in $O(\alpha\log^*\beta)$ time.
The coloring algorithms of \cite{Linial92,SzegedyV93} take $O(\log^* \beta)$ time steps in $H^{[1,\alpha-1]}$,
each of which can be simulated with $\alpha-1$ time steps in $H$.
By Theorem~\ref{thm:AGLP89}, an $(\alpha,(\alpha-1)\log(O(\hat\Delta^2)))$-ruling set can 
be computed for $U$ in 
$O(\alpha\log(\hat\Delta^2))$ time.  
Note that $(\alpha-1)\log(O(\hat{\Delta^2})) = 2(\alpha-1)^2(\log\Delta+O(1))$.
\end{proof}

\subsection{Miscellany}

In each of our algorithms there is some arbitrary (constant) parameter $c$ that controls the failure probability, 
which is always of the form $n^{-\Omega(c)}$.
All logarithms are base 2 unless specified otherwise.
We make repeated use of the inequality $(1+x) \le e^x$, which holds for all $x$.

\section{A Maximal Independent Set Algorithm}\label{sect:MIS}

In Section~\ref{sect:almost-MIS} we give an $O(\log^2 \Delta)$-time randomized 
algorithm called $\IndependentSet$ that computes a large, but not necessary maximal, independent set.
A new two-phase MIS algorithm is presented in Section~\ref{sect:MIS-algorithm}.
In Phase I it invokes $\IndependentSet$ to find a set $I$ with two properties,
(i) all surviving vertices in $V(G)\backslash \hat\Gamma(I)$ form components with size $\poly(\Delta)\log n$,\footnote{Recall that $\hat\Gamma(I) \bydef I \cup \Gamma(I)$ contains all vertices in or adjacent to $I$.}
and (ii) all $(5,O(\log\Delta))$-ruling sets in each component have size less than $\log n$.
As a consequence of property (i) we can bound the message size by $\poly(\Delta)\log n$.  (In the worst case a message encodes the topology of the entire component.)
Using property (ii) we can extend $I$ to an MIS in $O(\log\Delta\cdot \exp(O(\sqrt{\log\log n})))$ time, 
deterministically.  Phase I succeeds with probability $1-n^{-\Omega(1)}$
and if it does succeed, Phase II succeeds with probability 1.

Refer to Figures~\ref{alg:IndependentSet} and \ref{alg:MIS} for the pseudocode of 
$\IndependentSet$ and $\MIS$.

\subsection{Computing an Almost Maximal Independent Set}\label{sect:almost-MIS}

The $\IndependentSet$ algorithm uses a generalization of Luby's~\cite{Luby86} randomized experiment.
It consists of $\log\Delta$ {\em scales}, each composed of $O(\log\Delta)$ Luby steps.  The purpose
of the $k$th scale is to reduce the maximum degree in the surviving graph to $\Delta/2^k$.
At some nodes this invariant will fail to hold with some non-negligible probability.  We call such nodes
{\em bad} and remove them from consideration.  The components induced by bad nodes are reconsidered
in Phase II of the $\MIS$ algorithm.

\begin{figure}
\centering
\framebox{\hcm[.1]
\begin{minipage}{5.5in}
$\IndependentSet(\mbox{Graph } G)$
\begin{enumerate}
	\item Initialize sets $I,B\subset V(G)$:\\
	$
	\begin{array}{rlll}
		I &\leftarrow \emptyset	&\hcm[.4]& \mbox{\{an independent set\}}\\	
		B &\leftarrow \emptyset	&& \mbox{\{a set of `bad' nodes\}}
	\end{array}
	$
	
	Throughout, let $\VIB \bydef V(G) \backslash (\hat\Gamma(I) \cup B)$ be the nodes still under consideration: those not marked bad and not in or adjacent to the independent set.  
	Let $\GIB$ be the graph induced by $\VIB$ and let $\GammaIB$ and $\degIB$ be the neighborhood and degree functions w.r.t.~$\GIB$.
	\item For each {\em scale} $k$ from 1 to $\log \Delta+1$,
	
	\begin{enumerate}
		\item Execute $c \log\Delta$ iterations of steps i and ii.
		\begin{enumerate}
			\item Each node $v\in \VIB$ chooses a random bit $b(v)$:\\
			\item[]
				$
				b(v) \leftarrow \left\{
				\begin{array}{lll}
					1 & \hcm[.2] & \mbox{with probability $1/(\degIB(v) + 1)$}\istrut[4]{0}\\
					0 && \mbox{with probability $1 - 1/(\degIB(v) + 1)$}
				\end{array}
				\right.
				$
			\item \istrut{5}$I \leftarrow I \cup \{v\in \VIB \;|\; b(v)=1 \,\mbox{ and }\, b(u) = 0 \,\mbox{ for all $u\in \GammaIB(v)$}\}$.\\
				{\em (Add nodes to the independent set.)}
		\end{enumerate}
		\item $B\leftarrow B \cup \{v\in \VIB \;|\; \degIB(v) > \Delta/2^k\}$.\\
			{\em (Mark high-degree nodes as bad.)}
	\end{enumerate}
	\item $B\leftarrow B\backslash \hat\Gamma(I)$.\\
	{\em (Bad nodes adjacent to $I$ no longer need to be considered bad.)}
	
	\item Return $(I, B)$.
\end{enumerate}
\end{minipage}
}
\caption{\label{alg:IndependentSet}}
\end{figure}

\begin{lemma}\label{lem:IndependentSet-prob}
Consider a single iteration of Step 2a (a `Luby step') in $\IndependentSet$.
If $v\in \VIB$ and $\degIB(v) > \Delta/2^k$ before the iteration, 
the probability that $v\in\hat\Gamma(I)$ after the iteration is at least $(1-e^{-1/2})e^{-1}$.
\end{lemma} 

\begin{figure}
\centering
\scalebox{.37}{
\includegraphics{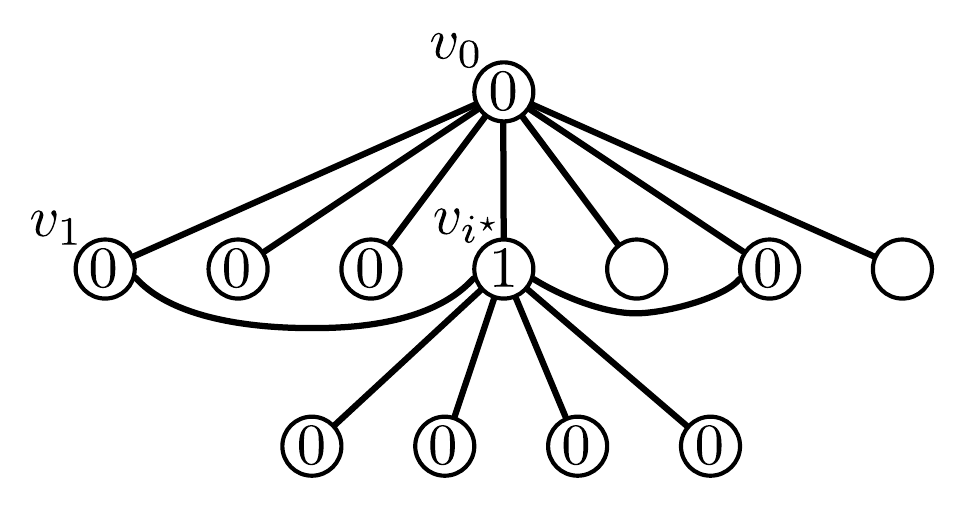}}
\caption{\label{fig:MIS-one-iteration}The node $v_0$ is eliminated if some node in its inclusive neighborhood joins the independent set.  This occurs if some $v_{i^\star}$ chooses $b(v_{i^\star})=1$ and $1 \not \in b(\Gamma(v_{i^\star}))$.}
\end{figure}

\begin{proof}
Let $\hGammaIB(v) = \{ v=v_0,v_1,v_2,\ldots,v_{\degIB(v)}\}$ be the inclusive
neighborhood of $v$.  By assumption $\degIB(v) > \Delta/2^k$ and since $v_1,\ldots,v_{\degIB(v)}$
were not marked bad (placed in $B$) in the last execution of Step 2b, 
$\degIB(v_i) \le \Delta/2^{k-1}$ for each $i\le \degIB(v)$.
Let $i^\star \in \{0,\ldots,\degIB(v)\}$ be the first index for which $b(v_{i^\star}) = 1$.
The probability that $i^\star$ exists is
\[
1 - \prod_{i=0}^{\degIB(v)}  \paren{1-\frac{1}{\degIB(v_i)+1}}
\ge 1 - \paren{1 - \frac{1}{\Delta/2^{k-1}+1}}^{\Delta/2^k+1} > 1-e^{-1/2}.
\]
If $i^\star$ does exist, $v_{i^\star}$ is included in the independent set $I$ if all its neighbors set their $b$-values to zero.
This occurs with probability
\[
\prod_{u \in \GammaIB(v_{i^\star})\backslash\{v_0,\ldots,v_{i^\star-1}\}} \paren{1 - \frac{1}{\deg(u)+1}}
\ge \paren{1-\frac{1}{\Delta/2^{k-1}+1}}^{\Delta/2^{k-1}} > e^{-1}.
\]
Nodes $v_0,\ldots,v_{i^\star-1}$ are excluded from consideration since, by definition of $i^\star$, they have already set their $b$-values to zero.
Thus, after one iteration of Step 2a, $v$ is in $\hat{\Gamma}(I)$ with probability $(1-e^{-1/2})e^{-1} \approx 0.145$.\
See Figure~\ref{fig:MIS-one-iteration} for an illustration.
\end{proof}

\begin{lemma}\label{lem:IndependentSet-t-prob}
Let $U\subset V(G)$ be a node set such that $\dist_G(u,U\backslash\{u\})\ge 5$ for each $u\in U$.
The probability that $U\subseteq B$ after a call to $\IndependentSet(G)$
is less than $\Delta^{-c|U|/5}$.
\end{lemma}

\begin{proof}
The event that a node $v\in \VIB$ 
appears in $\hat\Gamma(I)$ after one iteration of Step 2a
depends only on the random bits chosen by $v$'s neighbors and neighbors' neighbors.
Since all nodes in $U$ are mutually at distance at least five,
in each iteration the events that they appear in $\hat\Gamma(I)$ are independent.
Call a node $v\in \VIB$ {\em vulnerable} in a particular iteration of Step 2a
if $\degIB(v) > \Delta/2^k$.  We cannot say for certain when a node will be vulnerable,
but eventually each must, for some $k$, be vulnerable throughout scale $k$, 
until it appears in $\hat\Gamma(I)$ or is placed in $B$ at the end of the scale.
By Lemma~\ref{lem:IndependentSet-prob} the probability that an individual node
ends up in $B$ is at most $p^{c\log\Delta}$, where 
$p=1 - (1-e^{-1/2})e^{-1} \approx 0.855$.
Since $\log p < -0.22$, $p^{c\log\Delta} = \Delta^{c\log p} < \Delta^{-c/5}$.
Since outcomes for $U$-nodes are independent in any iteration of Step 2a, 
the probability that every node in $U$ ends up in $B$ is at most $\Delta^{-c|U|/5}$.
\end{proof}

\begin{lemma}\label{lem:IndependentSet-properties}
Let $(I,B)$ be the pair returned by $\IndependentSet(G)$. For $t= \log_\Delta n$,
$(I,B)$ satisfies the following properties with probability $1-n^{-c/5 + 11}$.
\begin{enumerate}[leftmargin=1cm]
\item There does not exist any $U\subset \VIB$ with $|U| = t$ such that for any $U' \subset U$, $\dist_G(U',U\backslash U') \in [5,9]$.\label{part:IndependentSet-props:one}
\item All components in the graph induced by $\VIB$ have fewer than $t\Delta^4$ nodes.\label{part:IndependentSet-props:two}
\end{enumerate}
\end{lemma}

\begin{proof}
A set $U\subset V$ satisfying the criteria of Part (1) forms a $t$-node tree in the graph $G^{[5,9]}$.  
(This tree is not necessarily unique.)
The number of rooted unlabeled $t$-node trees is less than $4^t$ since the Euler tour of such a tree can be encoded as a bit-vector with length $2t$.
The number of ways to embed such a tree in $G^{[5,9]}$ is less than $n\cdot \Delta^{9(t-1)}$ : there are $n$ choices for the root
and less than $\Delta^9$ choices for each subsequent node.  By Lemma~\ref{lem:IndependentSet-t-prob} the probability that
$U\subseteq B$ is less than $\Delta^{-ct/5}$.  
By a union bound, the probability that any such $U$ is contained in $B$ is less than
\begin{align*}
4^t \cdot n \cdot \Delta^{9(t-1)} \cdot \Delta^{-ct/5} 
\,<\, n^{\log_\Delta 4 + 10 - c/5} \,<\, n^{-c/5 + 11}.
\end{align*}

Turning to Part (2), suppose there is such a connected component $C$ with $t\Delta^4$ nodes.
We can find a subset $U$ of the nodes satisfying the criteria of Part (1) by the following greedy 
procedure.
Choose an arbitrary initial node $v_1\in C$ and set $U\leftarrow\{v_1\}$.  
Iteratively select a $v_i \in C\backslash U$ 
for which $\dist_G(v_i, U) = 5$, set $U\leftarrow U\cup\{v_i\}$, and 
then remove from consideration all nodes within distance 4 of $v_i$.
The number removed is less than $\Delta^4$, 
hence $U$ has size at least $(t\Delta^4)/\Delta^4 = t$.  
\end{proof}

\subsection{The MIS Algorithm}\label{sect:MIS-algorithm}

The pseudocode for $\MIS$ appears in Figure~\ref{alg:MIS}.  We walk through each step of the algorithm below.
Recall that $\IndependentSet(G)$ returns an independent set $I$ and set of `bad' nodes $B$.

\begin{figure}
\centering
\framebox{\hcm[.1]
\begin{minipage}{5.2in}
$\MIS(\mbox{Graph } G)$
\begin{enumerate}
	\item[] {\bf Phase I:}
	\item $(I,B) \leftarrow \IndependentSet(G)$.\label{Step-one}%
	
	\item[] The following steps focus on a single connected component $C$ in $\GIB$.  They are executed in parallel for each such $C$.

	\item $(I_C,B_C) \leftarrow \IndependentSet(C)$.\label{Step-two}
	\item[] {\bf Phase II:}
	\item $R_C \leftarrow $ a $(5,32\log\Delta+O(1))$-ruling set for $B_C = V(C)\backslash \hat\Gamma(I_C)$ w.r.t.~$C$.\label{Step-three}
	\item Form a cluster around each node $x\in R_C$ and form the cluster graph $C^\star$.\label{Step-four}
	
	\begin{align*}
	\Cluster(x) &\leftarrow \left\{v \in B_C \;\left|\; 
		\begin{array}{ll}
		\mbox{for any other $x'\in R_C$, $\dist_C(v,x) < \dist_C(v,x')$\istrut[3]{5}}\\
		\mbox{or $\dist_C(v,x) = \dist_C(v,x')$ and $\ID(x) < \ID(x')$\istrut[2]{0}}
		\end{array}
		\right.
		\right\}
	\\&\\
	C^\star &\leftarrow \paren{R_C, \left\{(x,x') \;\left|\;  
		\begin{array}{l}
		\mbox{there exists $(v,v')\in E(C)$ such that}\\
		\mbox{$v\in\Cluster(x)$ and $v'\in\Cluster(x')$}
		\end{array}
		\right.
		\right\}
		}
	\end{align*}

	\item $(\mathscr{D},\mathscr{C}) \leftarrow $ a $ \paren{2^{O(\sqrt{\log \log n})},2^{O(\sqrt{\log \log n})}}$-network\ decomposition of $C^\star$.\label{Step-five}
	\item Compute the clustering of $V(C)$ defined by $\mathscr{D}$. For each $D\in \mathscr{D}$,\label{Step-six}
	\begin{align*}
	\Cluster^\star(D) \leftarrow \bigcup_{x\in D} \Cluster(x).
	\end{align*}
	\item For each color $k \in \left\{1,\ldots,2^{O(\sqrt{\log\log n})}\right\}$,\label{Step-seven}
		\begin{enumerate}
		\item For each cluster $D\in \mathscr{D}$ with $\mathscr{C}(D) = k$, in parallel,
		\item[] \hcm[.5]$J_D \leftarrow $ an MIS of the graph induced by $\Cluster^\star(D) \backslash \hat\Gamma(I_C)$.
		\item $\displaystyle I_C \leftarrow I_C \cup \bigcup_{\substack{D\in\mathscr{D} : \\\mathscr{C}(D) = k}} J_D$.
		\end{enumerate}
	\item Return $\displaystyle I \cup \bigcup_{C \;\mathrm{ in }\; \GIB} I_C$.\label{Step-eight}
\end{enumerate}
\end{minipage}
}
\caption{\label{alg:MIS}}
\end{figure}

\paragraph{Step~\ref{Step-one}} After Step~\ref{Step-one} we have an independent set $I$ and a set of bad nodes
$B = \VIB = V(G) \backslash \hat\Gamma(I)$.  By Lemma~\ref{lem:IndependentSet-properties}(\ref{part:IndependentSet-props:two}), with high probability
each connected 
component in $\GIB$ has at most $t \cdot \Delta^4$ nodes and therefore at most $t\cdot \Delta^5/2$ edges, 
where $t = \log_\Delta n$.
Step~\ref{Step-one} (and Step~\ref{Step-two}) require only 1-bit messages since each node only has to notify its neighbors about its status (whether in $I$ or not, whether in $\VIB$ or not) and the $b$-values it selects in each round.
Since the remaining steps operate on each component in $\GIB$ independently, messages of size $O(\Delta^5\log_\Delta n)$ suffice.

\paragraph{Step~\ref{Step-two}} At this point we could simply run Panconesi and Srinivasan's~\cite{PanconesiS96} deterministic MIS algorithm on each component. 
This would take time $2^{O\paren{\sqrt{\log(t \Delta^4)}}}$, which is not the desired bound, unless $\Delta$ happens to be polylogarithmic in $n$.
In order to make this approach work for all $\Delta$ we need to reduce the ``effective'' size of each component $C$ to at most $\log n$, independent of $\Delta$.  After Step~\ref{Step-two} we have partitioned $V(C)\subseteq \VIB$ 
into $\hat\Gamma(I_C)$ and $B_C$.  As we argue in the next paragraph, 
Lemma~\ref{part:IndependentSet-props:one} implies
that $B_C$ (the bad nodes in $C$) can be efficiently partitioned into $\log n$ low-radius clusters.
This is the only property of $(I_C,R_C)$ that we use in subsequent steps.

\paragraph{Steps~\ref{Step-three} and \ref{Step-four}} 
Recall that nodes are assigned distinct $O(\log n)$-bit IDs.
Using Corollary~\ref{cor:AGLP} with $\alpha=5$, we can compute a $(5,32\log\Delta+O(1))$-ruling set $R_C$ for $B_C$ in 
$O(\log\Delta + \log^* n)$ time.  We form a cluster around each ruling set node in the obvious way: each $v\in B_C$ joins the cluster of the nearest $x\in R_C$, called $\Cluster(x)$, breaking ties by node ID.  Note that {\em nearest} is with respect to $\dist_C$, so the shortest path from $v$ to $x$ does not leave $C$ 
but may go through nodes in $\hat\Gamma(I_C)$.  The cluster graph $C^\star$ 
is obtained by contracting each cluster $\Cluster(x)$ to a single node, also called $x$.  

We cannot use Lemma~\ref{lem:IndependentSet-properties}(\ref{part:IndependentSet-props:one}) directly to bound the size of $R_C$
since $\dist_C(v,R_C\backslash\{v\})$ is only guaranteed to be at least $5$, not in the interval $[5,9]$.   
Consider a greedy procedure for obtaining a $(5,4)$-ruling set $R_C' \supseteq R_C$.
Initialize $R_C' \leftarrow R_C$, then evaluate each $u\in B_C$, setting $R_C' \leftarrow R_C'\cup\{u\}$ 
if $u$ is at distance at least 5 from all vertices $R_C'$.   After this process completes, any $u\not\in R_C'$
has $\dist_C(u,R_C') \le 4$ and for any $U' \subset R_C'$, $\dist_C(U',R_C' \backslash U') \in [5,9]$.
Thus, with probability $1-n^{-c/5+11}$, $|R_C| \le |R_C'| \le t$.
Note that the algorithm does not actually compute $R_C'$. It was just 
introduced to obtain an upper bound on $|R_C|$.

\paragraph{Steps~\ref{Step-five} and \ref{Step-six}} We run Panconesi and 
Srinivasan's~\cite{PanconesiS96} decomposition algorithm on $C^\star$.  
(See Remark~\ref{rem:black-box-phase-II}, below, for a discussion of the subtle difficulties in implementing this algorithm.)
Since $|R_C| \le t = \log_\Delta n < \log n$ we can compute a $\paren{2^{O(\sqrt{\log\log n})},2^{O(\sqrt{\log\log n})}}$-network decomposition $(\mathscr{D},\mathscr{C})$ in $2^{O(\sqrt{\log\log n})}$ time.  Since the underlying network is $C$, not $C^\star$, each step of this algorithm requires $64\log\Delta+O(1)$ steps to simulate in $C$.  The total time is therefore 
$\log\Delta \cdot 2^{O(\sqrt{\log\log n})}$.  
Since $\Cluster^\star(D)$ is the union of disjoint clusters in $\{\Cluster(x) \;|\; x\in D\}$,
the diameter of $\Cluster^\star(D)$ with respect to $\dist_C$ is at most 
$(64\log\Delta+O(1))\cdot 2^{O(\sqrt{\log\log n})}$.

\paragraph{Step~\ref{Step-seven}}
We extend $I_C$ to an MIS on $C$ using the network decomposition.  For each color class, for each cluster $D$, supplement $I_C$ with an MIS $J_D$ on $\Cluster^\star(D)/\hat{\Gamma}(I_C)$.  These MISs are computed by the trivial algorithm and in parallel: 
a representative node in $D$ retrieves the status of all nodes in $\Cluster^\star(D)$, in $O(\log\Delta\cdot 2^{O(\sqrt{\log\log n})})$ time, then computes an MIS $J_D$ and announces it to all nodes in 
$\Cluster^\star(D)$.
At the end of this process $I_C$ is a maximal independent set on $C$.

\paragraph{Step \ref{Step-eight} and Correctness}
The set returned in Step~\ref{Step-eight}, $I\cup\bigcup_C I_C$, is usually an MIS of $G$.  However, 
poor random choices in Steps~\ref{Step-one} and \ref{Step-two} can cause the algorithm to fail during Step~\ref{Step-five}.
The ruling set $R_C$ has size at most $t$ with high probability.  If it is larger than $t$ then Steps~\ref{Step-three} and \ref{Step-four} will be executed without error, but Step~\ref{Step-five} may fail to produce a 
$\paren{2^{O(\sqrt{\log\log n})},2^{O(\sqrt{\log\log n})}}$-network decomposition in the time allotted. 
If this occurs, Steps~\ref{Step-six} and \ref{Step-seven} cannot be executed.

\paragraph{Running Time} The time for Steps~\ref{Step-one} and \ref{Step-two} is $O(\log^2 \Delta)$ and the time for Steps~\ref{Step-three} and \ref{Step-four} is $O(\log\Delta + \log^* n)$.  Steps~\ref{Step-five}--\ref{Step-seven} take 
$O(\log\Delta) \cdot \exp(O(\sqrt{\log\log n}))$ time.  In total the time is $O(\log^2\Delta + \log\Delta\cdot \exp(O(\sqrt{\log\log n})))$, which is $O(\log^2\Delta + \exp(O(\sqrt{\log\log n})))$.

\begin{theorem}\label{thm:MIS}
In a graph with maximum degree $\Delta$, 
an MIS can be computed in $O(\log^2\Delta + \exp(O(\sqrt{\log\log n})))$ time, 
with high probability, using messages with size $O(\Delta^5\log_\Delta n)$.
\end{theorem}

\begin{remark}\label{rem:black-box-phase-II}
One must be careful in applying deterministic algorithms in Phase II in a black box fashion.
In the proof of Theorem~\ref{thm:MIS} we reduced the number of clusters per component to $t$ and deduced that
the Panconesi-Srinivasan~\cite{PanconesiS96} algorithm runs in 
$\log\Delta \cdot 2^{O(\sqrt{\log t})}$ time on each component.  This is not a correct inference.
The stated running time of the Panconesi-Srinivasan algorithm depends on nodes being endowed with
$O(\log t)$-bit IDs (if the number of nodes is $t$), whereas in Step~\ref{Step-five} nodes still have their
original $O(\log n)$-bit IDs.
There is a simple generic fix for this problem.
Suppose a deterministic Phase II algorithm $\mathcal{A}$ runs in time $T=T(t)$ on any instance $C$ with size $t$
whose nodes are assigned distinct $O(\log t)$-bit labels.
Let $k$ be minimal such that $t \geq \log^{(k)} n$.
Just before executing $\mathcal{A}$, 
first compute an $O(t^2\log^{(k)} n) = O(t^3)$-coloring in the graph $C^{[1,2T+1]}$ with Linial's~\cite{Linial92} 
algorithm and use these colors as $(3\log t + O(1))$-bit node IDs.  
This takes $O(Tk)$ time, that is, $O(T)$ time whenever $t = \log^{(O(1))} n$.
As far as $\mathcal{A}$ can tell, all nodes have distinct IDs since no node can ``see'' two nodes with the same ID.
\end{remark}

\section{An Algorithm for Maximal Matching}\label{sec:MMalg}

The $\Match$ procedure given in Figure~\ref{fig:Match} is a generalized version of one iteration of the 
Israeli-Itai~\cite{II86} matching algorithm.
It is given not-necessarily-disjoint node sets $U_1,U_2$
and a matching $M$, and returns a matching on $U_1\times U_2$ that is node-disjoint from $M$.
It works as follows.  
Each unmatched node in $U_1$ {\em proposes} to an unmatched neighbor in $U_2$, selected uniformly at random.
Each node in $U_2$ receiving a proposal {\em accepts} one, breaking ties by node ID.
The accepted proposals form a set of {\em directed} paths and cycles.
At this point each node $v$ generates a bit $b(v)$: 0 if $v$ is at the beginning of a path, 1 if at the end of a path,
and uniformly at random otherwise.  A directed edge $(u,v)$ enters the matching if and only if $b(u) = 0$ and $b(v)=1$.
Refer to Figure~\ref{fig:MM-example} for an execution of $\Match$ on a small graph.

\begin{figure}
\centering\framebox{\hcm[.1]
\begin{minipage}{4.8in}
$\Match(U_1,U_2,M)$
\begin{enumerate}
\setlength{\itemsep}{2pt}%
\item Each $u\in U_1\backslash V(M)$ {\em proposes} to $\Prop(u)$:\label{step:MM1}
\item[] $\Prop(u) \leftarrow $ a random neighbor of $u$ in $U_2\backslash V(M)$.
\item Each $v\in U_2\backslash V(M)$ with a proposal {\em accepts} the best one:\label{step:MM2}
\item[] $\displaystyle\Prop^\star(v) \leftarrow \argmax_{u \::\: \Prop(u) = v} \{\ID(u)\}$.
\item $F \leftarrow \{(\Prop^\star(v),v) \;|\; v\in U_2 \mbox{ for which $\Prop^\star(v)$ exists}\}$\\
{\em ($F$ is a set of \underline{directed} edges.  It consists of directed paths and cycles.)}\label{step:MM3}

\item Each $v \in U_1 \cup U_2$ with $\deg_F(v)>0$ chooses a $b(v) \in \{0,1\}$:\istrut[2]{0}\\
$b(v) \leftarrow \left\{\begin{array}{l@{\hcm[.4]}l}
0 			&	\mbox{if $\indeg_F(v)=0$,}\\
1			&	\mbox{if $\outdeg_F(v)=0$,}\\
\mbox{a random value in $\{0,1\}$}	& \mbox{otherwise.}
\end{array}\right.$
\label{step:MM4}

\item Return the matching $\{(u,v)\in F \;|\; b(u) = 0 \mbox{ and } b(v) = 1\}$.\label{step:MM5}
\end{enumerate}
\end{minipage}
}
\caption{\label{fig:Match}}
\end{figure}

\begin{figure}
\centering
\begin{tabular}{ccc}
\scalebox{.4}{\includegraphics{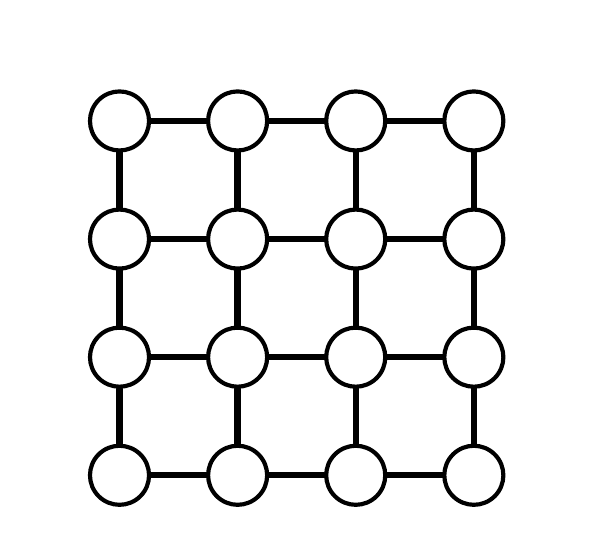}}	& \scalebox{.4}{\includegraphics{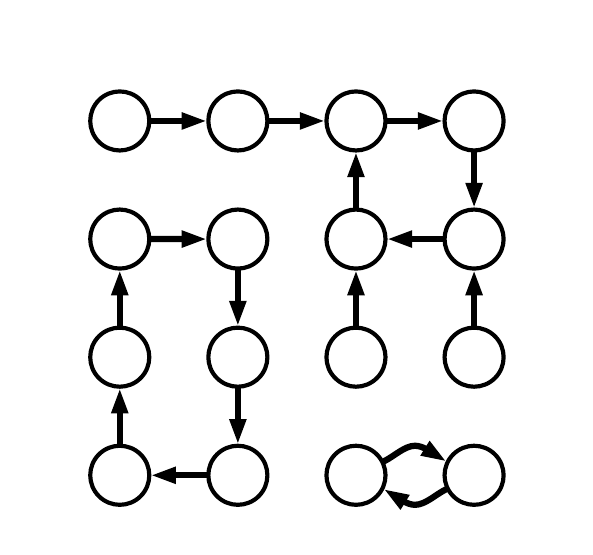}} & \scalebox{.4}{\includegraphics{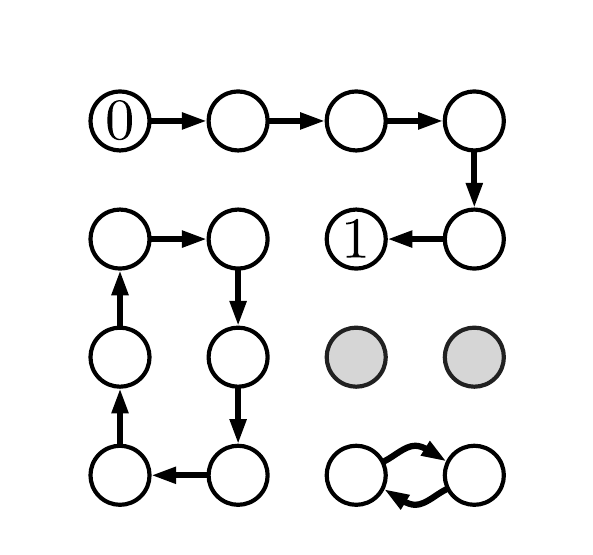}}\\
(a)	& (b)		& (c)
\end{tabular}
\caption{\label{fig:MM-example} 
One possible execution of $\Match(V,V,\emptyset)$.
Left: the undirected input graph $G=(V,E)$.  
Middle: the directed pseudoforest $(V,\, \{(u,\Prop(u))\})$ induced by the proposals.
Right: $F$ consists of directed paths and cycles.  
The beginning and end of each path are labeled 0 and 1, respectively.  Grayed, 
isolated nodes receive no label.  All other nodes are assigned random labels in $\{0,1\}$.}
\end{figure}

The procedure $\MaximalMatch$ has a two-phase structure.  Phase I consists of $O(\log \Delta)$ {\em stages}
in which the matching, $M$, is supplemented using two calls to $\Match$.  After Phase I all components
of unmatched vertices have fewer than $s = (c\ln n)^9$ nodes, with probability $1-n^{-\Omega(c)}$.  
We apply the deterministic
$O(\log^4 s) = O(\log^4 \log n)$ time maximal matching algorithm of \cite{HanckowiakKP01} on each component, in parallel.
In total the running time is  $O(\log\Delta + \log^4\log n)$.

Let $V_i \bydef V(G) \backslash V(M)$ be the set of unmatched nodes just {\em before} stage $i$.  
For brevity we let $\deg_i$ and $\Gamma_i$
be the degree and neighborhood functions for the graph induced by $V_i$.
The parameters for stage $i$ are given below.  Roughly speaking, $\delta_i$ is the maximum degree
at stage $i$, $\tau_i = 2\delta_i/(c\ln n)$ is a certain `low-degree' threshold, 
and $\nu_i  = \delta_i\tau_i/2$ is a bound on the sum of degrees 
of nodes in $\Gamma_i(v)$, for any $v$.  
Define
\begin{align*}
\delta_i &\bydef \f{\Delta\sqrt{c\ln n}}{\rho^i},\istrut[5]{0}\\
\tau_i &\bydef \f{2\Delta}{\rho^i \sqrt{c\ln n}},\istrut[5]{0}\\
\mbox{and } \; \nu_i &\bydef \f{\Delta^2}{\rho^{2i}} \,=\, \f{\delta_i\tau_i}{2}, \; \; \mbox{ where } \rho \bydef \sqrt{16/15} < 1.033.
\intertext{Define the low degree and high degree nodes before stage $i$ to be}
V_{i}^{\lo} &\bydef \{v\in V_{i} \;|\; \deg_{i}(v) \le \tau_{i+1}\}\\
\mbox{and } \; V_{i}^{\hi} &\bydef \{v\in V_{i} \;|\; \deg_{i}(v) > \delta_{i+1}\}.
\end{align*}
Note that nodes with degree between $\tau_{i+1}$ and $\delta_{i+1}$ are in neither set.
In stage $i$ we supplement the current matching, 
first with a matching on 
$V_i^{\lo} \times V_i^{\hi}$, then with a matching on $V_i$.
As we soon show, certain invariants will hold after stage $i$ with probability $1-\exp(-\Omega(\tau_i))$.
Thus, in order to obtain high probability bounds we must switch to a different analysis
when $\tau_i = \Theta(\log n)$, that is, when the maximum degree is $\delta_i = \Theta(\log^2 n)$.

\begin{figure}
\centering
\framebox{\hcm[.1]
\begin{minipage}{4in}
$\MaximalMatch(\mbox{Graph } G)$
\begin{enumerate}
\setlength{\itemsep}{0pt}%
\item[] {\bf Phase I:}
\item Initialize $M_{0} \leftarrow \emptyset$
\item For each {\em stage} $i$ from $0$ to $z \bydef \log_\rho \Delta + \log_{4/3}(c\ln n) - 1$.
\item[] \hcm[.5] $M \leftarrow M \, \cup \, \Match(V_{i}^{\lo},V_{i}^{\hi},M)$
\item[] \hcm[.5] $M \leftarrow M \, \cup\,  \Match(V_{i},V_{i},M)$\\
\item[] {\bf Phase II:}
\item[] Let $\mathscr{C}$ be the connected components in the graph induced by $V_{z}$ 
containing less than $(c\ln n)^{9}$ nodes.
\item For each $C\in \mathscr{C}$,
\item[] \hcm[.5] $M_C \leftarrow $ a maximal matching on $C$
\item Return $\displaystyle M\cup \bigcup_{C\in\mathscr{C}} M_C$.
\end{enumerate}
\end{minipage}
}
\caption{\label{alg:MM}}
\end{figure}

The algorithm always returns a matching.  If, at the beginning of Phase II, $\mathscr{C}$
contains {\em all} connected components on $V_{z}$ then the returned matching is clearly 
maximal.  Thus, our goal is to show that with high probability, after Phase I there is 
no connected component of unmatched nodes with size greater than $(c\ln n)^{9}$.
In the lemma below $\deg(S)$ is short for $\sum_{u\in S}\deg(u)$, where $S\subset V$.

\begin{lemma}\label{lem:degree-bound}
Define $i^\star$ to be the last stage for which 
$\tau_{i^\star} \ge 2c\ln n$.
With probability $1 - 2n^{-c/660 + 1}$,
the following bounds hold for all $v\in V(G)$ 
after each stage $i < i^\star$.
\begin{align*}
\deg_{i+1}(v) &\le \delta_{i+1}\\
\mbox{ and } \; \degtwo_{i+1}(v) &\le \nu_{i+1},
\end{align*}
where $\degtwo_{i+1}(v) \bydef \deg_{i+1}(\Gamma_{i+1}(v))$.
\end{lemma}

\begin{proof}
The inequalities hold trivially when $i=0$.  We analyze the probability that they hold after stage $i$, 
assuming they hold just before stage $i$.
For the sake of minimizing notation we use $\deg_{i}, \Gamma_{i}$, etc.
to refer to the degree and neighborhood functions just before {\em each} call to $\Match$ in stage $i$.  
This should not cause confusion.

Consider a node $v\in V_{i}$ at the beginning of stage $i$.  
By assumption $\deg_{i}(v) \le \delta_{i}$ and $\degtwo_{i}(v) \le \nu_{i}$.
Since, by definition, nodes in $V_i^{\lo}$ have degree at most $\tau_{i+1}$,
$v$ has less than $\nu_{i}/\tau_{i+1} = \delta_{i+1}\cdot (\rho^2/2)$ neighbors that are {\em not} in $V_{i}^{\lo}$.  
We argue that if $v\in V_i^{\hi}$ (that is, $\deg_{i}(v) > \delta_{i+1}$) then $v$ will be matched in the {\em first} call to $\Match$ in stage $i$
with probability $1-\exp((1-\rho^2/2)c\ln n/2)$.
Note that the forest induced by the proposals consists solely of stars 
(all edges being directed from $V_i^{\lo}$ to $V_i^{\hi}$)
which implies that $F$, the graph consisting of accepted proposals, consists solely of single-edge paths.  
Single-edge paths in $F$ are always committed to the matching since
their endpoints' $b$-values are chosen deterministically in 
Step~\ref{step:MM4} of $\Match$ to satisfy the criterion of Step~\ref{step:MM5}.
Thus, $v\in V_i^{\hi}$ will be matched if any neighbor $u \in V_i^{\lo}$ chooses $(u,v)$ in Step~\ref{step:MM2}.  The probability 
that this does {\em not} occur is at most
\begin{align*}
\paren{1-\frac{1}{\tau_i}}^{|\Gamma_{i}(v) \cap V_i^{\lo}|}
	&\le \paren{1-\frac{1}{\tau_i}}^{\paren{1-\frac{\rho^2}{2}}\delta_{i+1}}\\
	&\le \exp\paren{-\paren{1-\frac{\rho^2}{2}}\frac{\delta_{i+1}}{\tau_i}}\\
	&=\, \exp\paren{-\paren{1-\frac{\rho^2}{2}}\frac{c\ln n}{2\rho}} \,<\, n^{-0.22c}			& \{\rho < 1.033\}
\end{align*}
By a union bound, every $v\in V_i^{\hi}$ will be matched with probability more than $1-n^{-c/5+1}$.
Therefore, we proceed under the assumption that after the first call to $\Match$ in stage $i$, 
all unmatched nodes have degree less than $\delta_{i+1}$.  It remains to show that after 
the second call to $\Match$, $\degtwo_{i+1}(v) \le \nu_{i+1}$, for all $v\in V(G)$.

A node $v$ will be guaranteed to have positive degree in $F$ under two circumstances:
(i) some node offers $v$ a proposal, or (ii) among those nodes proposing to $\Prop(v)$, 
$v$ has the highest ID.  Once $v$ is in a path or cycle in $F$ 
it becomes matched with probability at least 1/2.
(It is actually exactly $1/2$, except if $v$ is in a single-edge path, in which case it is 1.)

In the following analysis we first expose the proposals made by 
all nodes in $V_{i} \backslash \hat\Gamma_i(v)$ 
then expose the proposals of $\hat\Gamma_i(v)$ in {\em descending} order of node ID.
Consider the moment just before a neighbor $u\in\Gamma_{i}(v)$ makes a proposal.
If at least $\deg_{i}(u)/2$ neighbors of $u$ have yet to receive a proposal
(by nodes already evaluated) then place $u$ in set 
$A$, otherwise place $u$ in set $B$.
If $u$ is put in set $A$ and $u$ {\em does}
offer $\Prop(u)$ its first proposal thus far---implying that $u$ will have positive degree in $F$---then also place $u$ in set $A'$.
See Figure~\ref{fig:ABC} for an illustration.

\begin{figure}
\centering
\scalebox{.4}{\includegraphics{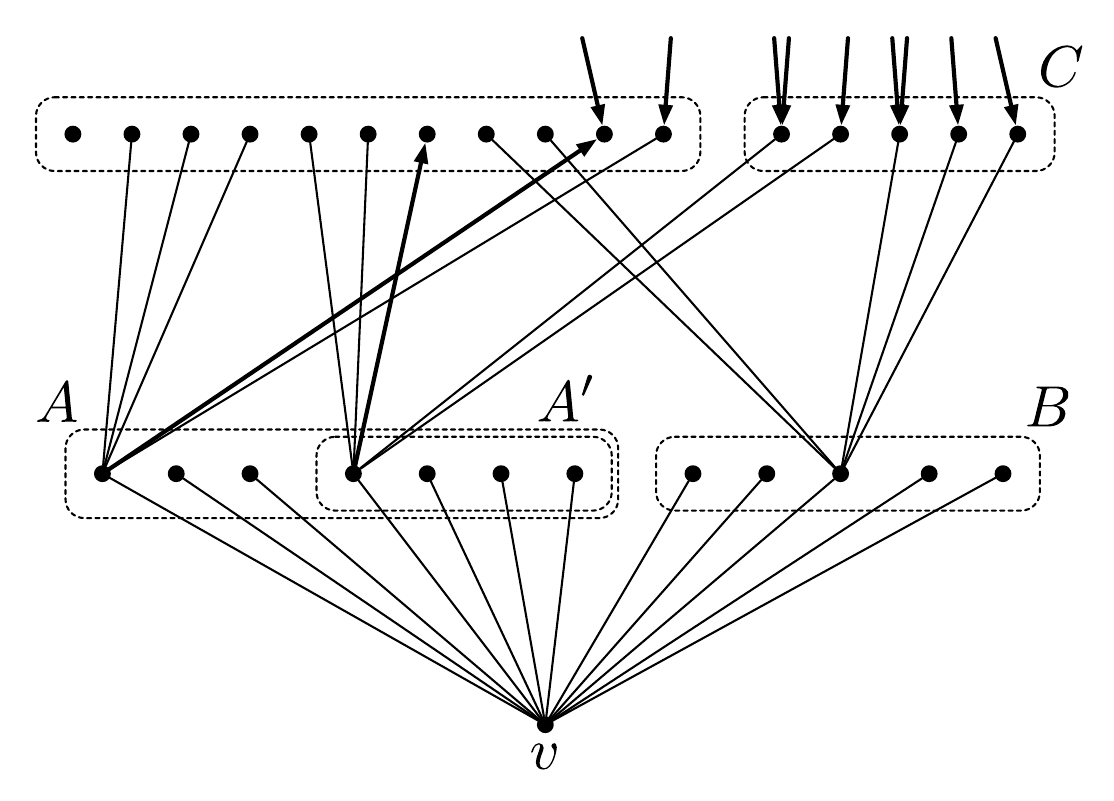}}
\caption{\label{fig:ABC}The neighborhood of $v$ is partitioned into $A$ and $B$, and $A$ is partitioned
into $A'$ and $A\backslash A'$.  Proposals are indicated by directed edges.  
A node is in $A$ if a majority of its neighbors do not already have a proposal
and in $B$ otherwise.  An $A$-node is in $A'$ if it makes the first proposal to a node.  
A node is in $C$ if it is adjacent to $B$ and has a proposal.  Note: nodes with a proposal that are adjacent to $A$
but not $B$ are not in $C$.  Contrary to the depiction, $A$-nodes and $B$-nodes may be adjacent
and $C$ may intersect both $A$ and $B$. }
\end{figure}

We split the rest of the analysis into 
two cases depending on whether $A$-nodes or 
$B$-nodes account for the larger share of edges in $v$'s 2-neighborhood.
In both cases we show that $\degtwo_{i+1}(v) \leq \nu_{i+1}$ with high probability.

\subsection{Case I: The $A$-nodes}

We first analyze the case that $\deg_{i}(A)\ge \degtwo_{i}(v)/2 \ge \nu_{i+1}/2$.  
(If $\degtwo_{i}(v)$ is already less than $\nu_{i+1}$ there is nothing to prove.)
Observe that each node $u$, once in $A$, is moved to $A'$ with probability at least 1/2, and if so,
contributes $\deg_{i}(u) \le \delta_{i+1}$ to $\deg_{i}(A')$.\footnote{Note that this process fits in the martingale framework of Corollary~\ref{cor:Azuma-Hoeffding}.
Here $X_j$ is the state of the system after evaluating the $j$th neighbor $u$ of $v$ and $Z_j$ is $\deg_i(u)$ if $u$ joins $A'$ and $0$ otherwise, which is a function of $X_j$.
Thus, each $Z_j$ has a range of at most $\delta_{i+1}$.}
The probability that after evaluating each 
$u\in\Gamma_{i}(v)$, $\deg_{i}(A')$ is less than a $\frac{1}{\sqrt{2}}$-fraction of its expectation is
\begin{align*}
\lefteqn{\Pr(\deg_{i}(A') < \fr{1}{\sqrt{2}} \cdot \E[\deg_{i}(A')])}\\
										&\le \exp\paren{-\f{\paren{(1 - \fr{1}{\sqrt{2}}) \E[\deg_i(A')]}^2}{2\sum_{u\in A}(\deg_{i}(u))^2}}               & \{\mbox{Corollary~\ref{cor:Azuma-Hoeffding}}\}\\
										&\le \exp\paren{-\f{\paren{(1 - \fr{1}{\sqrt{2}}) \frac{1}{2} \deg_{i}(A)}^2}{2(\deg_{i}(A)/\delta_{i+1}) \delta_{i+1}^2}}              & \{\mbox{linearity of expectation}\} \\
										&\le \exp\paren{-\paren{\f{(1-\fr{1}{\sqrt{2}})^2}{8}}\paren{\f{\deg_{i}(A)}{\delta_{i+1}}}}\\
										&\le \exp\paren{-\paren{\f{(1-\fr{1}{\sqrt{2}})^2}{32}}\tau_{i+1}}                     & \left\{\deg_{i}(A)\ge \f{\nu_{i+1}}{2} = \f{\delta_{i+1}\tau_{i+1}}{4}\right\}\\
										&< n^{-c/187}								& \{\tau_{i+1} \ge \tau_{i^\star} \ge 2c\ln n\}
\intertext{%
We proceed under the assumption that this unlikely event does not hold, so
$
\deg_{i}(A') \ge \f{1}{\sqrt{2}}\cdot\E[\deg_i(A')] \ge \f{1}{2\sqrt{2}}\cdot \deg_i(A) \ge \f{1}{4\sqrt{2}}\cdot \nu_{i+1}.
$
Since each node with positive degree in $F$ is matched with probability at least 1/2, by linearity of expectation
$\E[\deg_i(A') - \deg_{i+1}(A')] \ge \fr{1}{2}\deg_{i}(A')$.  Moreover, whether $v\in A'$ is matched depends only on the $b$-values of neighboring nodes
in $F$.  The dependency graph of these events has chromatic number $\chi = 5$ since the nodes of a cycle can be 5-colored
such that any two nodes within distance 2 receive different colors.  
The probability that $\deg_i(A') - \deg_{i+1}(A')$ is less than a $\f{1}{\sqrt{2}}$-fraction of its expectation is therefore}
\lefteqn{\zero{\Pr\paren{\deg_i(A') - \deg_{i+1}(A') < \fr{1}{\sqrt{2}}\cdot \E[\deg_{i}(A') - \deg_{i+1}(A')]}}}\\
		&\le \zero{\displaystyle\exp\paren{-\frac{2\paren{\paren{1-\f{1}{\sqrt{2}}}\E[\deg_{i}(A') - \deg_{i+1}(A')]}^2}{\chi\cdot \sum_{u\in A'} (\deg_i(u))^2}}}			& \{\mbox{Theorem~\ref{thm:Janson}, $\chi=5$}\}\\
		&\le \exp\paren{-\frac{2\paren{\paren{1-\fr{1}{\sqrt{2}}}\fr{1}{2}\deg_i(A')}^2}{\chi\cdot (\deg_i(A')/\delta_{i+1})\delta_{i+1}^2}}\\
		&\le \exp\paren{-\paren{\frac{(1-\f{1}{\sqrt{2}})^2}{10}}\paren{\frac{\deg_i(A')}{\delta_{i+1}}}}\\
		&\le \exp\paren{-\paren{\frac{(1-\f{1}{\sqrt{2}})^2}{80\sqrt{2}}}\tau_{i+1}}	& \left\{\deg_i(A') \ge \f{\nu_{i+1}}{4\sqrt{2}} = \f{\delta_{i+1}\tau_{i+1}}{8\sqrt{2}}\right\}\\
		&< n^{-c/660}				& \{\tau_{i+1} \ge \tau_{i^\star} \ge 2c \ln n\}
\end{align*}

To sum up, if this unlikely event does not occur,
\begin{align*}
\degtwo_i(v) - \degtwo_{i+1}(v) &\ge \deg_{i}(A')-\deg_{i+1}(A') 		& \mbox{\{because $A' \subseteq \Gamma_i(v)$\}}\\
					&\ge {1 \over {\sqrt{2}}} \cdot \E[\deg_i(A') - \deg_{i+1}(A')]\\
					&\ge \frac{1}{2\sqrt{2}}\cdot \deg_i(A') 	
					\zero{\displaystyle\;\ge\; \paren{\frac{1}{2\sqrt{2}}}^2\cdot \deg_i(A) 
					\;\ge\; \frac{1}{16}\degtwo_i(v).}
\end{align*}

Thus, with high probability, $\degtwo_{i+1}(v) \le \frac{15}{16}\cdot \degtwo_i(v)$.

\subsection{Case II: The $B$-nodes}

We now turn to the case when $\deg_i(B) \ge \f{1}{2}\cdot \degtwo_i(v) \ge \f{1}{2}\cdot \nu_{i+1}$.
By definition, just before any $u\in B$ makes its proposal, at least $\f{1}{2}\cdot \deg_i(u)$ of its neighbors have already received a proposal.
We do not care who $u$ proposes to.
Let $C\subseteq \Gamma_i(B)$ be the set of nodes in $B$'s neighborhood that receive at least one proposal.
For $x\in C$, let $\deg_B(x) \le \delta_{i+1}$ be the number of its neighbors in $B$.
Thus, if $x$ is matched then $\degtwo(v)$ is reduced by at least $\deg_B(x)$.  
It follows that 
\begin{align*}
\deg_B(C) &= \sum_{x\in C} \deg_B(x) = \sum_{u\in B} \deg_C(u) \ge \sum_{u\in B} \f{1}{2}\cdot \deg_i(u) & \{\mbox{by defn. of $u\in B$}\}\\
&= \f{1}{2}\cdot \deg_i(B) \ge \f{1}{4}\cdot \degtwo_i(v) > \f{1}{4}\cdot \nu_{i+1}.
\end{align*}
Since $C$-nodes are matched with probability 1/2, by linearity of expectation,
$\E[\deg_{i+1}(B)] \le \deg_i(B) - \f{1}{2}\cdot \deg_B(C) \le \f{3}{4}\deg_i(B)$.  
We bound the probability that $\deg_{i+1}(B)$ deviates from its expectation 
using Janson's inequality, in exactly the same way as we bounded $\deg_{i+1}(A')$.
It follows that
\begin{align*}
\lefteqn{\Pr\paren{\deg_{i+1}(B) \ge \deg_i(B) - \f{1}{4}\cdot \deg_B(C)}}\\
&\le \exp\paren{-\f{2(\fr{1}{4}\deg_B(C))^2}{\chi\cdot \sum_{x\in C}(\deg_B(x))^2}}       & \{\mbox{Theorem~\ref{thm:Janson}}\}\\
										&\le \exp\paren{-\f{1}{40}\cdot \f{(\deg_B(C))^2}{(\deg_B(C)/\delta_{i+1}) \delta_{i+1}^2}}   &  \{\chi=5, \deg_B(x) \le \delta_{i+1}\}\\
										&\le \exp\paren{-\f{1}{320}\tau_{i+1}}      & \{\deg_B(C)\ge\nu_{i+1}/4 = \delta_{i+1}\tau_{i+1}/8\}\\
										&\le n^{-c/160}			& \{\mbox{$\tau_{i+1} \ge \tau_{i^\star} \ge 2c \ln n$}\}
\end{align*}

Thus, with high probability 
\[
\degtwo_{i+1}(v) \le \degtwo_{i}(v) - \f{1}{4}\cdot \deg_B(C) \le \f{15}{16}\cdot \degtwo_i(v),
\]
since $\deg_B(C) \ge \f{1}{4}\cdot \degtwo_i(v)$.
Whether we are in Case I or Case II, 
$\degtwo_{i+1}(v) \le \f{15}{16}\cdot \degtwo_i(v) \le \f{15}{16}\cdot \nu_i$ with high probability.
Since $\nu_{i+1} = \nu_i/\rho^2$, we set $\rho = \sqrt{16/15}$.

By a union bound, the probability of error at {\em any} node is at most
$2n^{-c/660+1}$.  This covers the probability that the first call to $\Match$ fails to match all $V_i^{\hi}$-nodes
or the second call fails to make $\degtwo_{i+1}(v) \le \nu_{i+1}$, for all $v\in V_i$.
\end{proof}

\subsection{The Emergence of Small Components}

Lemma~\ref{lem:degree-bound} implies that before stage
$i^\star < \log_\rho\Delta$,
the maximum degree is at most 
$\delta_{i^{\star}} = \tau_{i^{\star}} (c/2)\ln n \le (c\ln n)^2$.
In Lemmas~\ref{lem:degree-decay} and \ref{lem:size-bound} we prove that after another 
$O(\log\log n)$ iterations of the $\Match$ procedure, all components
of unmatched vertices have size at most $(c\ln n)^9$, with high probability.
Thus, Phase II of $\MaximalMatch$ correctly extends the matching after 
Phase I to a maximal matching.

\begin{lemma}\label{lem:degree-decay}
For any node $v$ and any stage $i$, 
$\Pr\left(\deg_{i+1}(v) \le \frac{3}{4}\cdot \deg_i(v)\right) \ge \frac{1}{4}$.
\end{lemma}

\begin{proof}
We analyze the expected drop in $v$'s degree during the {\em second} call to $\Match$ (the one in which all nodes participate), 
then apply Markov's inequality.
Expose the proposals in descending order of node ID, and consider the moment just before $v$ makes its proposal.
Let $P\subseteq \Gamma_i(v)$ be those neighbors already holding a proposal
and $Q\subseteq \Gamma_i(v)$ be the neighbors with no proposal.
All nodes in $P$ will be matched with 1/2 probability and $v$ will be matched with 
1/2 probability {\em if} it proposes to a member of $Q$.  The probability $v$ is matched is at least $\frac{\epsilon}{2}$, where $\epsilon = |Q|/\deg_i(v)$.
The probability that $u\in P$ is still a neighbor of $v$ after this call to $\Match$ is therefore
at most $\f{1}{2}(1-\frac{\epsilon}{2})$.  The probability that $u\in Q$ is still a neighbor is at most $1-\frac{\epsilon}{2}$.  
By linearity of expectation, 
\begin{align*}
\E[\deg_{i+1}(v)] &\le \paren{\epsilon\paren{1-\fr{\epsilon}{2}} + \fr{1}{2}(1-\epsilon)\paren{1-\fr{\epsilon}{2}}} \cdot  \deg_i(v) \\
			&= (1-\fr{\epsilon}{2})(\fr{1}{2}+\fr{\epsilon}{2})\cdot \deg_i(v)\\
			&\le (\fr{3}{4})^2\cdot  \deg_i(v)			& \mbox{\{maximized at $\epsilon=1/2$\}}
\end{align*}
That is, we lose at least a $\frac{7}{16}$-fraction of $v$'s neighbors in expectation.
By Markov's inequality, $\Pr\left(\deg_{i+1}(v) \le \frac{3}{4}\cdot \deg_i(v)\right) \ge \frac{1}{4}$.
\end{proof}

\begin{lemma}\label{lem:size-bound}
Let $\hat{G}$ be the subgraph induced by unmatched nodes at some point in Phase I,
whose maximum degree is at most $\hat{\Delta}$.
After $12\log_{4/3} \hat{\Delta}$ more stages in Phase I,
all components of unmatched nodes have size at most 
$t\hat{\Delta}^4$
with probability $1-n^{-c}$, where $t \bydef c\ln n$.
\end{lemma}

\begin{proof}
The proof follows the same lines at that of Lemma~\ref{lem:IndependentSet-t-prob} and \ref{lem:IndependentSet-properties}, 
but has some added complications.
We say $v$ is {\em successful} in stage $i$ if $\deg_{i+1}(v) \le \frac{3}{4}\cdot \deg_i(v)$.
If $v$ experiences $\log_{4/3} \hat{\Delta}$ successes then either $v$ has been matched or all neighbors
of $v$ are matched.

The events that $u$ and $v$ are successful in a particular stage $i$ are independent if 
$\dist_{\hat{G}}(u,v) \ge 5$ since the success of $u$ and $v$ only depend on the random choices
of nodes within distance 2.  Any subgraph of size $t\hat{\Delta}^4$ must contain a subset $T$ of 
$t$ nodes such that (i) each pair of nodes in $T$ is at distance at least 5 and (ii) $T$ forms a $t$-node tree
in $\hat{G}^5$.  Call $T$ a {\em distance-5 set} if $|T|=t$ and it satisfies (i) and (ii).  There are less than $4^t\cdot n\cdot \hat{\Delta}^{5(t-1)}$
distance-5 sets in $\hat{G}$.  (There are less than $4^t$ topologically distinct trees with $t$ nodes
and less than $n \hat{\Delta}^{5(t-1)}$ ways to embed one such tree in $\hat{G}^5$.)

Consider any distance-5 set $T$.
Over $12\log_{4/3} \hat{\Delta}$ consecutive stages, $v\in T$ experiences
some number of successful stages.  Call this random variable $X_v$ and define 
$X \bydef \sum_{v\in T} X_v$.  By Lemma~\ref{lem:degree-decay} and 
linearity of expectation,
\[
\E[X] = \sum_{v\in T} E[X_v] \ge t\cdot\fr{1}{4}(12\log_{4/3} \hat{\Delta}) = 3t\log_{4/3} \hat{\Delta}.
\]
If $X \ge t\log_{4/3}\hat{\Delta}$ then some $X_v \ge \log_{4/3}\hat{\Delta}$, 
implying that $v$ becomes isolated and therefore that no component contains all $T$-nodes.
We will call $T$ {\em successful} if any member of $T$ becomes isolated.
By a Chernoff bound (Theorem~\ref{thm:Chernoff}), the probability that $T$ is unsuccessful is at most
\begin{align*}
\Pr\paren{X < t\log_{4/3}\hat{\Delta}} &\le \Pr\paren{X < \f{1}{3}\cdot \E[X]}\\
						&\le \exp\paren{-\frac{2\paren{\frac{2}{3}\E[X]}^2}{4t\log_{4/3}\hat{\Delta}}}\\
						&\le \exp\paren{-2 t\log_{4/3}\hat{\Delta}}			& \{\E[X] \ge 3t\log_{4/3}\hat{\Delta}\}\\
						&= \hat{\Delta}^{-(2\log_{4/3} e)t}
\end{align*}
After $12\log_{4/3}\hat{\Delta}$ stages, if there exists a component with size 
$t\hat{\Delta}^4$ then it must contain an unsuccessful subset $T$.  By the union bound, this occurs with probability less than
\begin{align*}
\lefteqn{4^t \cdot n \cdot \hat{\Delta}^{5(t-1)} \cdot \hat{\Delta}^{-(2\log_{4/3} e)t}}\\
&< 4^{c\ln n} \cdot n \cdot \hat{\Delta}^{(5 - 2\log_{4/3} e) \cdot c\ln n}	&\\
&< n^{-c}											& \mbox{\{for $\hat{\Delta}$ sufficiently large.  Note: $5-2\log_{4/3} e < 0$.\}}
\end{align*}
\end{proof}

\begin{theorem}\label{thm:MM}
In a graph with maximum degree $\Delta$, a maximal matching can be computed in
$O(\log\Delta + \log^4\log n)$ time with high probability using $O(1)$-size messages.
When the graph is bipartite and 2-colored, the time bound becomes $O(\log \Delta + \log^3\log n)$.
\end{theorem}

\begin{proof}
After $i^\star = \log_\rho (\Delta/(c\ln n)^{3/2})$ stages in Phase I 
the maximum degree is $\hat\Delta = (c\ln n)^2$, with high probability.
After another $4\log_{4/3} \hat\Delta$ stages in Phase I all connected components have 
at most $s \bydef \hat{\Delta}^4\cdot c\ln n = (c\ln n)^9$ nodes, with high probability.
We execute the deterministic maximal
matching algorithm of~\cite{HanckowiakKP01} for time sufficient to solve any instance
on $s$ nodes: $O(\log^4 s)$ time for general graphs and $O(\log^3 s)$ time for bipartite, 2-colored graphs.
Both Phase I and Phase II can be implemented with $O(1)$-size messages, that is,
this algorithm works in the $\CONGEST$ model.
\end{proof}

\section{Vertex Coloring}\label{sect:Deltaplusone}

We consider a slightly more stringent version of $(\Delta+1)$-coloring called $(\deg+1)$-coloring, 
where each node $v$ must adopt a color from the palette $\{1,\ldots,\deg(v)+1\}$, or more generally, an arbitrary 
set with size $\deg(v)+1$.\footnote{Some applications~\cite{AmirKKNP14} demand $(\deg+1)$-colorings, not $(\Delta+1)$-colorings.}
Although the palette of a node does not depend on $\Delta$, our algorithm still requires that nodes
know $\Delta$ and $n$.  

In Section~\ref{sect:OneShotColoring} we define and analyze a natural $O(1)$-time algorithm called $\OneShotColoring$
that colors a subset of the nodes.  Johannson~\cite{Johansson99} showed
that $O(\log n)$ applications of a variant of $\OneShotColoring$ suffice to $(\Delta+1)$-color a graph, with high probability.
Our goal is to show something stronger.  We show that 
after $O(\log \Delta)$ applications of $\OneShotColoring$, 
all nodes have at most $O(\log n)$ uncolored neighbors that each have $\Omega(\log n)$ uncolored neighbors.
This property allows us to reduce the resulting $(\deg+1)$-coloring problem to {\em two} $(\deg+1)$-coloring problems on subgraphs with maximum degree $O(\log n)$.  It is shown that on these instances, $O(\log\log n)$ further applications
of $\OneShotColoring$ suffice to reduce the size of all uncolored components to $\poly(\log n)$.
In Phase II we apply the deterministic 
$(\deg+1)$-coloring algorithm of Panconesi and Srinivasan~\cite{PanconesiS96} to 
the $\poly(\log n)$-size uncolored components.
The remainder of this section constitutes a proof of Theorem~\ref{thm:degplusone-coloring}.

\begin{theorem}\label{thm:degplusone-coloring}
In a graph with maximum degree $\Delta$, 
a $(\deg+1)$-coloring can be computed in $O(\log\Delta + \exp(O(\sqrt{\log\log n})))$ time
using $\poly(\log n)$-length messages.
\end{theorem}

\subsection{Analysis of $\OneShotColoring$}\label{sect:OneShotColoring}

The algorithm maintains a proper partial coloring $\Color \::\: V(G) \rightarrow \{1,\ldots,\Delta+1,\bottom\}$, 
where $\bottom$ denotes no color and $\Color(v) \in \{1,\ldots,\deg(v)+1\} \cup \{\bottom\}$.  
Initially $\Color(v)\leftarrow\: \bottom$ for all $v\in V(G)$.
Before a call to $\OneShotColoring$ some nodes have already committed to their final colors.
Each remaining uncolored node $v$ chooses $\Color^\star(v)$, a color selected uniformly at random from its remaining palette.
It may be that neighbors of $v$ also choose $\Color^\star(v)$.  If $v$ holds the highest ID among all such nodes contending 
for $\Color^\star(v)$, it permanently commits to that color.
The pseudocode for $\OneShotColoring$ appears in Figure~\ref{fig:OneShotColoring}.

\begin{figure}
\centering
\framebox{\hcm[.1]
\begin{minipage}{4.7in}
$\OneShotColoring(G,\Color)$
\begin{enumerate}
\item[] Define $U\subseteq V(G)$ and $\Psi \;:\; V(G) \rightarrow 2^{\{1,\ldots,\Delta+1\}}$ as follows.
\begin{align*}
	U &\bydef \{u\in V(G) \;|\; \Color(u) =\: \bottom\}	,							& \mbox{the uncolored vertices,}\\
\mbox{ and } \; \;\Psi(v) &\bydef \{1,\ldots,\deg(v)+1\} \backslash \Color(\Gamma(v)),		& \mbox{$v$'s available palette.}
\end{align*}
\item[] The following steps are executed for all $v\in U$, in parallel.
\item \istrut[3]{0}Select a $\Color^\star(v) \in \Psi(v)$ uniformly at random.

\item If $\ID(v) > \max\left\{ \ID(u) \;|\; u\in \Gamma_U(v) \mbox{ and } \Color^\star(u) = \Color^\star(v)\right\}$,\\
	Permanently assign $\Color(v) \leftarrow \Color^\star(v)$.
\end{enumerate}
\end{minipage}
}
\caption{\label{fig:OneShotColoring}}
\end{figure}

We analyze the properties of $\OneShotColoring$ from the point of view of some arbitrary uncolored node $v\in U$.  
Note that whether $v$ is colored depends only on its behavior and the behavior of neighbors with larger IDs,
denoted $\Gamma_U^>(v) \bydef \{u\in \Gamma_U(v) \;|\; \ID(u) > \ID(v)\}$.
Define $\Psi^{-1}(q) \bydef \{u \in \Gamma_U^>(v) \;|\; q\in \Psi(u)\}$ 
to be the set of $v$'s uncolored neighbors that are contending for color $q$ and have higher IDs.
Define $w(q) = \sum_{u\in\Psi^{-1}(q)} 1/|\Psi(u)|$ to be the {\em weight} of color $q$.  In other words, each neighbor $u$ distributes $1/|\Psi(u)|$ units of weight to each color in its palette.  Note that $1/|\Psi(u)| \le 1/(\deg_U(u)+1) \le 1/2$.
The probability that $q\in \Psi(v)$ is {\em available} to $v$ 
after exposing $\Color^\star(\Gamma_U^>(v))$ is 
\begin{align}
\Pr(q \not\in \Color^\star(\Gamma_U^>(v)))
	&= \prod_{u\in \Psi^{-1}(q)} \paren{1-\fr{1}{|\Psi(u)|}} \nonumber\\
	&\ge \prod_{u\in \Psi^{-1}(q)} \paren{\fr{1}{4}}^{1/|\Psi(u)|} 	\label{eqn:coloring1}\\
	&= \paren{\fr{1}{4}}^{w(q)}.	\nonumber
\intertext{%
Inequality (\ref{eqn:coloring1}) follows from the fact that $(1-x) \ge (1/4)^{x}$ when $x\in [0,1/2]$.
Let $X_q\in\{0,1\}$ be the indicator variable for the event that $q$ is available and $X = \sum_q X_q$.
By linearity of expectation, $\E[X] \ge \sum_{q\in \Psi(v)} \paren{\fr{1}{4}}^{w(q)}$.
By the convexity of the exponential function, this quantity is minimized when all color weights are equal.  Hence,}
\E[X] \ge \sum_{q\in \Psi(v)} \paren{\fr{1}{4}}^{w(q)} &\ge |\Psi(v)|\cdot \paren{\fr{1}{4}}^{\sum_q w(q)/|\Psi(v)|}	\nonumber\\
								&\ge |\Psi(v)|\cdot \paren{\fr{1}{4}}^{\deg_U(v)/|\Psi(v)|}	\label{eqn:coloring2}\\
								&> |\Psi(v)|/4.									\label{eqn:coloring3}
\end{align}
Inequalities (\ref{eqn:coloring2},\ref{eqn:coloring3}) follow 
from the fact that each neighbor in $\Gamma_U^>(v)$ can contribute at most one unit of weight
and that $|\Psi(v)| \ge \deg_U(v)+1 \ge \deg_U^>(v)+1$.
We will call $v$ {\em happy} if $X \ge |\Psi(v)|/8$, that is, if the number of available colors is at least half its expectation.  Let $\Happy_v$ be the event that $v$ is happy.
The variables $\{X_q\}$ are not independent.  However, Dubhashi and Ranjan~\cite{DubhashiR98} showed that 
$\{X_q\}$ are negatively correlated, and more generally that all balls and bins experiments of this form give rise to negatively correlated variables.\footnote{In this situation the colors are bins and the neighbors' choices are balls.}
By Theorem~\ref{thm:Chernoff},
\[
\Pr(\overline{\Happy_v}) \bydef \Pr\paren{X < \f{|\Psi(v)|}{8}} < \exp\paren{-\f{2\cdot (|\Psi(v)|/8)^2}{|\Psi(v)|}} = \exp\paren{-\f{|\Psi(v)|}{32}}.
\]

Lemma~\ref{lem:OneShotColoring} summarizes the relevant properties of $\OneShotColoring$ used in the next section.

\begin{lemma}\label{lem:OneShotColoring}
Let $U$ be the uncolored nodes before a call to $\OneShotColoring$ and $v\in U$ be arbitrary.
\begin{enumerate}[leftmargin=1cm]
\item $\Pr(\mbox{$v$ is colored}) > 1/4$.\label{part:OneShotColoring:a}
\item $\displaystyle\Pr(\Happy_v) > 1 - \exp\paren{-\fr{\deg_U(v)+1}{32}}$.\label{part:OneShotColoring:b}
\end{enumerate}
\end{lemma}

\subsection{A $(\deg+1)$-Coloring Algorithm}\label{sect:VertexColoring}

It goes without saying that our $\degColoring$ algorithm (Figure~\ref{alg:degColoring}) has a two-phase structure.
The ultimate goal of Phase I is to reduce the global problem to some number
of independent $(\deg+1)$-coloring subproblems, each on $\poly(\log n)$-size components,
which can be colored deterministically in Phase II.
We first prove that this is possible with $O(\log\log n)$ applications of $\OneShotColoring$,
{\em if} the uncolored subgraph already has maximum degree $\poly(\log n)$.

\begin{lemma}\label{lem:OneShotColoring-component-size}
Apply an arbitrary proper {\em partial} coloring to $G$, and let
$\hat{\Delta}$ be the maximum degree in the subgraph induced by uncolored nodes.
After $5\log_{4/3}\hat{\Delta}$ iterations of $\OneShotColoring$, 
all uncolored components have less than $t\hat{\Delta}^2$ nodes
with probability $1-n^{-c}$, where $t \bydef c\log_{\hat{\Delta}} n$.
\end{lemma}

\begin{proof}
The proof is similar to that of Lemma~\ref{lem:IndependentSet-properties} and Lemma~\ref{lem:size-bound}.  Whether a node 
is colored depends only on the color choices of nodes in its inclusive neighborhood.
Thus, if two nodes are at distance at least 3, their coloring events are independent.
Let $T\subset U$ be a distance-3 set, that is, one for which (i) $|T|=t=c\log_{\hat{\Delta}} n$,
(ii) the distance between each pair of nodes is at least 3, and (iii) $T$ forms a tree in the uncolored part of $G^3$.
There are less than $4^t\cdot n \cdot \hat{\Delta}^{3(t-1)} < n^{4c}$ distance-3 sets and the probability that one
is entirely uncolored after $5\log_{4/3}\hat{\Delta}$ iterations of $\OneShotColoring$ 
is, by Lemma~\ref{lem:OneShotColoring}, less than 
\[
\paren{\fr{3}{4}}^{5t\log_{4/3}\hat{\Delta}} = \paren{\fr{3}{4}}^{5(c\log_{\hat{\Delta}} n)\log_{4/3}\hat{\Delta}} = n^{-5c}.
\]
By a union bound, no distance-3 set exists with probability $n^{4c - 5c}=n^{-c}$.
Moreover, if there were an uncolored component with size $t\hat{\Delta}^2$ after $5\log_{4/3}\hat{\Delta}$ iterations of $\OneShotColoring$, 
it would have to contain such a distance-3 set.
\end{proof}

Lemma~\ref{lem:OneShotColoring-component-size} implies a $(\deg+1)$-coloring algorithm
running in $O(\log\Delta + \exp(O(\sqrt{\log(\Delta^2\log n)})))$ time.  Once the component size
is less than $\Delta^2\log n$ we can apply the deterministic 
$(\deg+1)$-coloring algorithm of Panconesi and Srinivasan~\cite{PanconesiS96}
to each uncolored component.
The exponential dependence on $\sqrt{\log \Delta}$ is undesirable.
Using Lemma~\ref{lem:OneShotColoring}
we show that, roughly speaking, degrees decay geometrically with each call to $\OneShotColoring$, 
with high probability.  This will allow us to reduce the dependence on $n$ to $\exp(O(\sqrt{\log\log n}))$.

\begin{lemma}\label{lem:highdegreeneighbors}
Define $U^{\hi} = \{u\in U \;|\; \deg_U(u) > \hat{\Delta}\}$ to be those 
{\em high degree} uncolored nodes, 
where $\hat{\Delta} \bydef c\ln n$.
Let $U_0$ and $U_1$ be the uncolored nodes before and after a particular call to $\OneShotColoring$.
Let $\Happy \bydef \bigcap_{v\in U_0^{\hi}} \Happy_v$ be the event that all $U_0^{\hi}$ nodes are happy.
\begin{enumerate}[leftmargin=1cm]
\item $\Pr(\overline{\Happy}) < n^{-c/32+1}$.\label{item:highdegreeneighbors-part1}
\item $\Pr\paren{\deg_{U_1^{\hi}}(v) \le \f{15}{16}\cdot \deg_{U_0^{\hi}}(v)} > 1 - n^{-c/512} - n^{-c/32+1}$.\label{item:highdegreeneighbors-part2}
\end{enumerate}
\end{lemma}

\begin{proof}
By Lemma~\ref{lem:OneShotColoring}(\ref{part:OneShotColoring:b}), the definition of $\hat{\Delta}=c\ln n$, 
and the union bound,
\[
\Pr(\overline{\Happy}) \;<\; |U_0^{\hi}| \cdot \exp\paren{-\frac{\hat{\Delta}+1}{32}} \;<\; n^{-c/32+1}.
\]
In other words, with high probability, every vertex in $U_0^{\hi}$ has a 1/8 fraction of its palette available to it.

Turning to Part~\ref{item:highdegreeneighbors-part2}, fix any vertex $v\in U_0^{\hi}$.  There are two ways 
a neighbor of $v$ in $U_0^{\hi}$
can fail to be a neighbor in $U_1^{\hi}$ after this call to $\OneShotColoring$.  It can either be colored
(in which case it is not in $U_1$) or a sufficient number of its neighbors can be colored so that it is no 
longer in $U_1^{\hi}$.  We ignore the second possibility and analyze the number of neighbors of 
$v$ in $U_0^{\hi}$ that are colored.   List the nodes of 
$\Gamma_{U_0^{\hi}}(v)$ in decreasing order of ID as $u_1,\ldots,u_{\deg_{U_0^{\hi}}(v)}$.
At step 0 we expose $\Color^\star(u)$ for all $u\not\in \Gamma_{U_0^{\hi}}(v)$
and at step $i$ we expose $\Color^\star(u_i)$.  Let $\mathbf{Y}_i$ be the information exposed after step $i$.
Whether $u_i$ is successfully colored is a function of $\mathbf{Y}_i$.  Moreover, the probability that $u_i$ is colored,
given $\mathbf{Y}_{i-1}$, is precisely the fraction of its palette that is still available, according to $\mathbf{Y}_{i-1}$.
Let $X_i \in \{0,1\}$ be the indicator variable for the event that $u_i$ is colored and $X = \sum_i X_i$.
Unless the unlikely event $\overline{\Happy}$ occurs, 
\[
\Pr(X_i = 1 \;|\; \mathbf{Y}_{i-1}) = \Pr(\mbox{$u_i$ is colored} \;|\; \mathbf{Y}_{i-1})  \ge 1/8,
\]
and by Corollary~\ref{cor:Azuma-Hoeffding}, 
\[
\Pr(X < \fr{1}{16}\deg_{U_0^{\hi}}(v) \;|\;\Happy) 
< \exp\paren{-\frac{(\fr{1}{16}\deg_{U_0^{\hi}}(v))^2}{2\deg_{U_0^{\hi}}(v)}}
= \exp\paren{-\frac{1}{512}\deg_{U_0^{\hi}}(v)}
\le n^{-c/512}.
\]
Thus, by a union bound, $\deg_{U_1^{\hi}}(v) \le \fr{15}{16} \deg_{U_0^{\hi}}(v)$ holds for {\em all} $v\in U_0^{\hi}$,
with probability $1 - n^{-c/512 + 1} - n^{-c/32+1}$.
\end{proof}

Lemma~\ref{lem:highdegreeneighbors} implies that after $\log_{16/15} \Delta$ iterations
of $\OneShotColoring$, with high probability 
no node has $\hat{\Delta} = c\ln n$ uncolored neighbors, each having $\hat{\Delta}$ uncolored neighbors.
At this point we break the remaining $(\deg+1)$-coloring problem into two 
subproblems with maximum degree $\hat{\Delta}$.
The first subproblem is on the graph induced by
$U^{\hi}$, the second is on $U \backslash U^{\hi}$.  
The maximum degree in $U^{\hi}$ is $\hat{\Delta}$, by the observation above,
and the maximum degree in $U\backslash U^{\hi}$ is $\hat{\Delta}$ by definition.
According to Lemma~\ref{lem:OneShotColoring-component-size}, after 
$O(\log\hat{\Delta})=O(\log\log n)$ more iterations of $\OneShotColoring$, the size of all uncolored components
is less than $s = \hat{\Delta}^2\cdot c\log_{\hat{\Delta}} n < \hat{\Delta}^3$.  
Each can be $(\deg+1)$-colored deterministically in $\exp(O(\sqrt{\log s})) = \exp(O(\sqrt{\log \log n}))$ time
using the algorithm of Panconesi and Srinivasan~\cite{PanconesiS96}.
The failure probability of the $\degColoring$ algorithm (see Figure~\ref{alg:degColoring} for pseudocode) 
is therefore $O(n^{-c/515+2})$.

\begin{figure}
\centering\framebox{\hcm[.1]
\begin{minipage}{5in}
$\degColoring(\mbox{Graph } G)$
\begin{enumerate}
\item[] {\bf Phase I:}
\item Initialize $\Color(v) \leftarrow \: \bottom$, for all $v\in V(G)$.

\item Repeat $\log_{16/15} \Delta$ times:\\
	$\OneShotColoring(G,\Color)$ 

\item Form high-degree and low-degree graphs.
\begin{align*}
U &\leftarrow \{v\in V(G) \;|\; \Color(v)=\: \bottom\}			& \mbox{uncolored nodes}\\
U^{\hi} &\leftarrow \{v\in U \;|\; \deg_{U}(v) > \hat{\Delta} \bydef c\ln n\}	 \hcm & \mbox{high-degree nodes}\\
G^{\hi} &\leftarrow \mbox{ the graph induced by $U^{\hi}$}			\\
G^{\lo} &\leftarrow \mbox{ the graph induced by $U \backslash U^{\hi}$} 
\end{align*}
\item Repeat $5\log_{4/3}\hat{\Delta}$ times:\\
	$\OneShotColoring(G^{\hi},\Color)$
\item Repeat $5\log_{4/3}\hat{\Delta}$ times:\\
	$\OneShotColoring(G^{\lo},\Color)$\\
\item[] {\bf Phase II:}
\item Color all remaining uncolored components of $G^{\hi}$ with size at most $\hat{\Delta}^3$.
\item Color all remaining uncolored components of $G^{\lo}$ with size at most $\hat{\Delta}^3$.
\end{enumerate}
\end{minipage}
}
\caption{\label{alg:degColoring}}
\end{figure}

\section{Ruling Sets}\label{sect:ruling-sets}

The $(2,\beta)$ ruling set algorithm of Bisht, Kothapalli, and Pemmaraju~\cite{BishtKP14}
works as follows.  Given a graph $G=(V,E)$ with maximum degree $\Delta$, the algorithm generates a 
series of
node sets $V(G) = R_0 \supseteq R_1 \supseteq \cdots \supseteq R_{\beta-1} \supseteq R_\beta$
with three properties, namely
\begin{enumerate}[leftmargin=1cm,label=\roman*.]
\item $R_i$ dominates $R_{i-1}$, that is, $\hat\Gamma(R_i) \supseteq R_{i-1}$, 
\item the maximum degree in the graph induced by $R_i$ is $\Delta_i \approx 2^{(\log\Delta)^{1-i\epsilon}}$, and
\item $R_\beta$ is an MIS in the graph induced by $R_{\beta-1}$.
\end{enumerate}
Property (i) implies that for all $v\in V(G)$, $\dist(v,R_\beta) \le \beta$.
Together with Property (iii) this implies that $R_\beta$ is a $(2,\beta)$-ruling set.

Using our MIS algorithm, the time to compute $R_\beta$ from $R_{\beta-1}$ 
is $O(\log^2 \Delta_{\beta-1} + \exp(O(\sqrt{\log\log n}))) = O(\log^{2(1-(\beta-1)\epsilon)} \Delta + \exp(O(\sqrt{\log\log n}))$, so we want to make $\epsilon$ as large as possible.
On the other hand, the time to compute $R_i$ from $R_{i-1}$ is 
$O(\log_{\Delta_i} \Delta_{i-1}) = O(\log^\epsilon \Delta)$.
Balancing these costs we get a time bound of $O\paren{\beta\log^{\frac{1}{\beta - 1/2}}\Delta + \exp(O(\sqrt{\log\log n}))}$
using messages with length $\poly(\Delta_{\beta-1})\log n$.  
The improvement over Bisht et al.'s~\cite{BishtKP14} time bound (namely, $O(\log^{\frac{1}{\beta-1}} \Delta + \exp(O(\sqrt{\log\log n})))$) comes solely from a better MIS algorithm.

The algorithm for computing $R_i$ from $R_{i-1}$ (which satisfies Properties (i) and (ii)) 
was first described by Kothapalli and Pemmaraju~\cite{KothapalliP12}.
For the sake of completeness we reproduce this sparsification algorithm and its analysis.
Refer to Figure~\ref{alg:Sparsify} for the pseudocode of $\Sparsify$ and $(2,\beta)$-$\RulingSet$.

\begin{lemma} (Kothapalli and Pemmaraju~\cite{KothapalliP12})
Given $G=(V,E)$ and a threshold $f$, a subset $U\subseteq V$ can be computed in 
$O(\log_f \Delta)$ time such that for every $v\in V(G)$, $\dist_G(v,U) \le 1$, and for every $v\in U$, $\deg_U(v) \le 2cf\ln n$,
with probability $n^{-c+2}$.
\end{lemma}

\begin{figure}
\centering
\begin{tabular}{c}
\framebox{\hcm[.1]
\begin{minipage}{4.8in}
$\Sparsify(\mbox{Graph $G$, Integer $f$})$
\begin{enumerate}
	\item Initialize $U \leftarrow \emptyset$.
	\item For $i$ from 1 to $\log_f \Delta$,
		\begin{enumerate}
			\item For each $v \in V(G) \backslash \hat\Gamma(U)$, independently, and in parallel:\\
			Set $U\leftarrow U\cup \{v\}$ with probability $(c\ln n) f^i/\Delta$.
		\end{enumerate}
	\item Return $U$.
\end{enumerate}
\end{minipage}
}
\\
\ \\
\framebox{\hcm[.1]
\begin{minipage}{4.8in}
$(2,\beta)$-$\RulingSet(\mbox{Graph $G$})$
\begin{enumerate}
	\item $R_0 \leftarrow V(G)$
	\item For $i$ from 1 to $\beta-1$\\
			$R_i \leftarrow \Sparsify(G_{i-1}, f_i)$, where $G_{i-1}$ is the graph induced by $R_{i-1}$.
	\item $R_\beta \leftarrow \MIS(G_{\beta-1})$
	\item Return$(R_\beta)$
\end{enumerate}
\end{minipage}
}
\end{tabular}
\caption{\label{alg:Sparsify}}
\end{figure}

\begin{proof}
Consider an execution of $\Sparsify(G,f)$.  Let $U_i$ be $U$ after the $i$th iteration of the loop and $V_i \bydef V \backslash \hat\Gamma(U_{i})$.
Assume, inductively, that just before the $i$th iteration 
the maximum degrees in the graphs induced by 
$V_{i-1}$  and $U_{i-1}$ 
are at most $\Delta/f^{i-1}$ and $f\cdot 2c\ln n$.
These bounds hold trivially when $i=1$.
Each $v\in V_{i-1}$ is included in $U_{i}$ independently with probability $c\ln n f^{i}/\Delta$, 
so the probability that a $v\in V_{i-1}$ with $\deg_{V_{i-1}}(v) > \Delta/f^i$ is not in $\hat\Gamma(U_{i})$ is less than
$(1-c\ln n f^i/\Delta)^{\Delta/f^i} < n^{-c}$.  Furthermore, if $v\in U_{i}$, 
\[
\E[\deg_{U_i}(v)] = \deg_{V_{i-1}}(v) \cdot c\ln n f^i/\Delta \le cf\ln n.  
\]
By Theorem~\ref{thm:stdChernoff}, the probability that $\deg_{U_i}(v) \ge 2cf\ln n$ is at most $\exp(-fc\ln n/3) < n^{-c}$.
Note that since $v$ and its neighborhood are permanently removed from consideration, 
it never acquires new neighbors in $U$, so $\deg_{U_{i}}(v) = \deg_{U}(v)$.
Thus, with high probability the induction hypothesis holds for the next iteration.
\end{proof}

\begin{theorem}\label{thm:rulingset}
A $(2,\beta)$-ruling set can be computed in $O(\beta \log^{\frac{1}{\beta-1/2}} \Delta + \exp(O(\sqrt{\log\log n})))$ time with high probability.
\end{theorem}

\begin{proof}
The algorithm simply consists of $\beta-1$ calls to $\Sparsify$ followed by a call to $\MIS$.
Every node in $R_{i-1}$ is in or adjacent to $R_i$, for $1\le i < \beta$, which implies that
$\dist(v,R_\beta) \le \beta$ for all $v\in V$.  Since $R_\beta$ is an independent set it is also a $(2,\beta)$-ruling set.
The time to compute $R_\beta$ is on the order of
\[
\frac{\log\Delta}{\log f_1}
+ \frac{\log(f_1\log n)}{\log f_2}
+\cdots
+ \frac{\log(f_{\beta-2}\log n)}{\log f_{\beta-1}}
+ \log^2(f_{\beta-1}\log n) + \exp(O(\sqrt{\log\log n})).
\]
Setting $\log f_i = (\log\Delta)^{1-i(\frac{2}{2\beta-1})}$, the time for each call to $\Sparsify$ 
is $O((\log \Delta)^{\frac{2}{2\beta-1}})$ and the time for the final $\MIS$ is 
$\exp(O(\sqrt{\log\log n}))$ plus 
\[
\log^2 f_{\beta-1} = (\log\Delta)^{2\paren{1-(\beta-1)\frac{2}{2\beta-1}}} = (\log\Delta)^{\frac{2}{2\beta-1}}.
\]
\end{proof}

Theorem~\ref{thm:rulingset} highlights an intriguing open problem.  Together with the KMW lower bound,
it shows that $(2,2)$-ruling sets
are provably easier to compute than $(2,1)$-ruling sets, the upper bound for the former being
$O(\log^{2/3}\Delta + \exp(O(\sqrt{\log\log n})))$ and the lower bound on the latter being $\Omega(\log\Delta)$. 
Is it possible to obtain any non-trivial lower bound on the complexity of computing
$(2,\beta)$-ruling sets for some $\beta > 1$?  In order to apply~\cite{KuhnMW10} one would need to 
invent a reduction from $O(1)$-approximate minimum vertex cover to $(2,\beta)$-ruling sets.

\section{Bounded Arboricity Graphs}\label{sect:arb}

Recall that a graph has arboricity $\lambda$ if its edge set is the union of $\lambda$ forests.
In the proofs of Lemma~\ref{lem:arb} and Theorem~\ref{thm:arb-to-bounded-degree},
$\deg_{E'}(u)$ is the number of edges incident to $u$ in $E'\subseteq E$
and $\deg_{V'}(u)$ is the number of neighbors of $u$ in $V'\subseteq V$.

\begin{lemma} \label{lem:arb}
Let $G$ be a graph of $m$ edges, $n$ nodes, and arboricity $\lambda$.
\begin{enumerate}
\item $m < \lambda n$.
\label{arb1}
\item The number of nodes with degree at least $t\ge \lambda+1$ is less than $\lambda n/(t-\lambda)$.
\label{arb2}
\item The number of edges whose endpoints both have degree at least $t\ge \lambda+1$ is less than $\lambda m/(t-\lambda)$.
\label{arb3}
\end{enumerate}
\end{lemma}

\begin{proof}
Part \ref{arb1} follows from the definition of arboricity.
For Parts \ref{arb2} and \ref{arb3}, let $U = \{v \;|\; \deg_G(v) \ge t\}$ be the set of high-degree nodes.
We have that 
\begin{align*}
\lambda n > m &\ge |\{(u,v)\in E(G) \;|\; \mbox{$u\in U$ or $v\in U$ or both}\}|\\
			&\ge \sum_{u \in U}(t - \deg_U(u)) + \frac{1}{2} \sum_{u \in U} \deg_U(u)\\
			&\geq t \cdot |U| - |E(U)| > (t - \lambda)\cdot |U|.
\end{align*}
Thus $|U| < \lambda n/(t-\lambda)$, proving Part~\ref{arb2}. 
Part~\ref{arb3} follows since the number of such edges is less than $\lambda |U| \le \lambda m/(t-\lambda)$.
\end{proof}

\begin{theorem}\label{thm:arb-to-bounded-degree}
Let $G$ be a graph of arboricity $\lambda$ and maximum degree $\Delta$, and let 
$t\ge \max\{(5\lambda)^8,(4(c+1)\ln n)^7\}$ be a parameter.  
In $O(\log_t \Delta)$ time we can find a matching $M\subseteq E(G)$
(or an independent set $I\subseteq V(G)$) such that with probability at least $1-n^{-c}$, 
the maximum degree in the graph induced by $V\backslash V(M)$ (or the graph induced by $V\backslash \hat\Gamma(I))$) 
is at most $t\lambda$.
\end{theorem}

\begin{proof}
In $O(\log_t n)$ rounds we commit edges to $M$ (or nodes to $I$) and remove all incident edges (or incident nodes).  Let $G$ be the graph still under consideration before some round and let 
$\High = \{v \in V \;|\; \deg_G(v) \ge t\lambda\}$ be the remaining high-degree nodes.  Our goal is to reduce the size of $\High$ by a roughly $t^{1/7}$ factor.
Let $\I = \{v \in \High \;|\; \deg_\High(v) \ge t\lambda/2\}$.  It follows that any node 
$v \in \High' \bydef \High\backslash \I$ has $\deg_{V\backslash \High}(v) \ge t\lambda/2$ since at most $t\lambda/2$ of its neighbors can be in $\High$.  
Let $\tilde{E}$ be any set of edges crossing the cut $(\High,V\backslash \High)$
such that for $v\in \High'$, $\deg_{\tilde{E}}(v) = t\lambda/2$.  In other words, discard all but $t\lambda/2$ edges incident to each $\High'$ node arbitrarily. 
Let $\Shell = \{u \;|\; v\in \High' \mbox{ and } (v,u) \in \tilde{E}\}$ be the neighborhood of 
$\High'$ with respect to $\tilde{E}$.  Note that $|\Shell| \le t\lambda |\High'|/2$.
See Figure~\ref{fig:arb}.
\begin{figure}[h]
\centering
\scalebox{.47}{\includegraphics{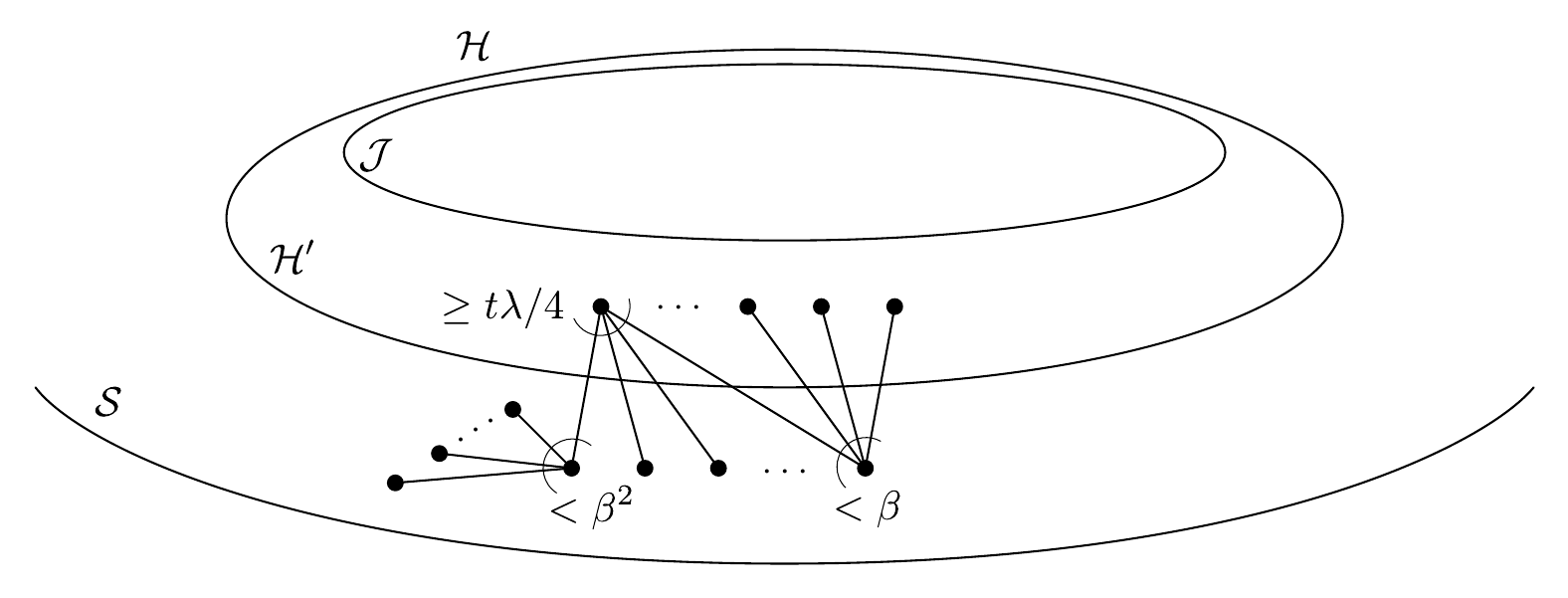}}
\caption{\label{fig:arb}
Good $\Shell$-nodes have fewer than $\beta$ neighbors in $\High'$ and fewer
than $\beta^2$ neighbors in $\Shell$.
Good $\High'$-nodes have at least $t\lambda/4$ good neighbors
in $\Shell$.}
\end{figure}
We define bad $\Shell$-nodes, bad $\tilde E$-edges, and bad $\High'$-nodes as follows, where $\beta = t^{1/7}$.
\begin{align*}
B_\Shell &= \left\{u \in\Shell \;|\; \mbox{ $\deg_{\tilde{E}}(u) \ge \beta$ or $\deg_{\Shell}(u) \ge \beta^2$ or both}\right\},\\
B_{\tilde E} &= \left\{\left.(u,v)\in\tilde E \;\right|\; u\in B_\Shell\right\},\\
\mbox{and } \; B_{\High'} &= \left\{v \in \High' \;\left|\; \deg_{B_{\tilde E}}(v) > \lambda t/4\right.\right\}.
\end{align*}
By Lemma~\ref{lem:arb}(\ref{arb3}) the number of edges $(u,v) \in B_{\tilde E}$ designated bad because $\deg_{\tilde E}(u) \ge \beta$ is less than $\lambda|\tilde{E}|/(\beta-\lambda)$. By Lemma~\ref{lem:arb}(\ref{arb2}) the number of {\em additional} edges $(u,v) \in B_{\tilde E}$ designated bad because 
$\deg_{\Shell}(u) \ge \beta^2$
is at most $(\beta-1) \lambda |\Shell|/(\beta^2 - \lambda)$ since there are less than $\lambda |\Shell|/(\beta^2-\lambda)$ such nodes and each contributes fewer than $\beta$ edges to $\tilde{E}$.
In total we have
\begin{align*}
|B_{\tilde{E}}| &< \frac{\lambda|\tilde{E}|}{\beta-\lambda} + \frac{(\beta-1) \lambda |\Shell|}{\beta^2 - \lambda} \\
		&\le \frac{\lambda(t\lambda|\High'|/2)}{\beta-\lambda} + \frac{(\beta-1)\lambda(t\lambda |\High'|/2)}{\beta^2-\lambda}			& \mbox{\{$|\Shell| \le |\tilde{E}| = t\lambda|\High'|/2$\}}\\
		&= \lambda^2t|\High'|\paren{\frac{1}{2(\beta-\lambda)} + \frac{\beta-1}{2(\beta^2-\lambda)}}\\
		&< \frac{\lambda^2t|\High'|}{\beta-\lambda}.
\end{align*}
A bad $v\in\High'$ is incident to more than $t\lambda/4$ edges in $B_{\tilde E}$, 
so
\begin{equation}\label{eqn:bad-high-degree}
|B_{\High'}| < \frac{|B_{\tilde E}|}{t\lambda / 4} < \frac{4\lambda |\High'|}{\beta-\lambda}.
\end{equation}

Our goal now is to select some nodes for the MIS (or edges for the maximal matching) that eliminate all good $\High'$ nodes, with high probability.
Each node $u\in \Shell \backslash B_{\Shell}$ selects a random real in $(0,1)$ and joins the MIS if it holds a local maximum.
The probability that $u$ joins the MIS is $1/(\deg_{\Shell \backslash B_{\Shell}}(u)+1) \ge 1/\beta^2$ and this event is clearly independent of all $u'\in\Shell\backslash B_{\Shell}$ at distance (in $\Shell\backslash B_{\Shell}$) at least 3.  Note that since the maximum degree in the graph induced by $\Shell\backslash B_{\Shell}$ is less than $\beta^2$, the neighborhood
of any good $v\in \High'\backslash B_{\High'}$ contains a subset of at least $(t\lambda/4) / \beta^4$ nodes, each pair of which is at distance at least 3 with respect to $\Shell \backslash B_{\Shell}$.  (Such a set could be selected greedily.)  Thus, the probability that no neighbor of $v\in\High'\backslash B_{\High'}$ joins the MIS is at most
\[
\paren{1-\frac{1}{\beta^2}}^{t\lambda/(4\beta^4)} < e^{-t\lambda / (4\beta^6)} = e^{-t^{1/7}\lambda / 4} \le 1/n^{c+1}.
\]
Thus, with high probability all good nodes $\High'\backslash B_{\High'}$ are eliminated.  Any remaining high degree nodes must be in either $\I$ or $B_{\High'}$.  By Lemma~\ref{lem:arb} and (\ref{eqn:bad-high-degree}), 
\begin{align*}
|\I| + |B_{\High'}| < \frac{\lambda |\High|}{t/2 - \lambda} + \frac{4\lambda |\High'|}{\beta-\lambda} < \frac{5\lambda|\High|}{\beta-\lambda}.
\end{align*}
Since $\beta = t^{1/7} \ge (5\lambda)^{8/7}$, the number of high-degree nodes is reduced by a $t^{\Omega(1)}$ factor.  Thus, after $O(\log_t \Delta)$ time all high-degree nodes have been eliminated with probability $1-1/n^c$.

In the case of maximal matching we want to select a matching that matches all $\High'$ nodes.
Each $u\in \Shell\backslash B_{\Shell}$ chooses an edge $(u,v) \in \tilde{E}\backslash B_{\tilde{E}}$ uniformly at random and proposes to $v$ that $(u,v)$ be included in the matching.  Any $v\in \High'\backslash B_{\High'}$ receiving a proposal accepts one arbitrarily and becomes matched.  
A good $v\in \High'\backslash B_{\High'}$ has at least $\deg_{\tilde{E}\backslash B_{\tilde{E}}}(v) \ge t\lambda/4$ neighbors $u\in \Shell\backslash B_{\Shell}$ with $\deg_{\tilde{E}\backslash B_{\tilde{E}}}(u) < \beta$, so the probability that $v$ receives no proposal (and remains unmatched) is less than 
$(1-1/\beta)^{t\lambda/4} < e^{-t^{6/7}\lambda/4} < o(1/n^{c+1})$.  As in the case of MIS, the number of high-degree nodes is reduced by a $t^{\Omega(1)}$ factor in $O(1)$ time.
(For the maximal matching problem our proof could be simplified somewhat since edges inside $\Shell$ play no part in the algorithm and need not be classified as good or bad.)
\end{proof}


\subsection{Consequences of Theorem~\ref{thm:arb-to-bounded-degree}}

Theorems~\ref{thm:arb-MM}, \ref{thm:arb-MIS}, \ref{thm:arb-coloring}, and \ref{thm:arb-rulingset} give new bounds
on the complexity of maximal matching, MIS, vertex coloring, and ruling sets in terms of $\lambda$.
Some results are a consequence of Theorem~\ref{thm:arb-to-bounded-degree}.
Others are obtained by combining the Phase I portion of our algorithms from Sections \ref{sect:MIS}--\ref{sect:Deltaplusone}
with one of the Barenboim-Elkin~\cite{BarenboimE10,BarenboimE11,BarenboimE13} algorithms for Phase II.

\begin{theorem}\label{thm:arb-MM}
In a graph with maximum degree $\Delta$ and arboricity $\lambda$, 
a maximal matching can be computed in time on the order of
\[
\min\left\{\;
\log \lambda + \sqrt{\log n}, \;\;
\log\Delta + \lambda + \log\log n
\; \right\}
\]
for all $\lambda$, and in time $O\paren{\log\Delta + \f{\log\log n}{\delta\log\log\log n}}$ when $\lambda = (\log\log n)^{1-\delta}$.
\end{theorem}

\begin{proof}
The first maximal matching 
bound is a consequence of Theorem~\ref{thm:arb-to-bounded-degree} and Theorem~\ref{thm:MM}.
We reduce the maximum degree to $\lambda t \bydef \lambda\cdot \max\{2^{\sqrt{\log n}}, (5\lambda)^8\}$
in $O(\log_t n) = O(\sqrt{\log n})$ time and find a maximal matching of the resulting graph in $O(\log(\lambda t) + \log^4\log n) = O(\log\lambda + \sqrt{\log n})$ time.
To obtain the second and third bounds 
we use the same algorithm from Theorem~\ref{thm:MM}, 
but rather than invoke~\cite{HanckowiakKP01} on each component of $s \le (c\ln n)^9$ nodes, 
we use the deterministic maximal matching algorithms of Barenboim and Elkin~\cite{BarenboimE10,BarenboimE13}.
Their algorithms run in $O(\f{\log s}{\delta\log\log s})$ time on graphs with size $s$ and arboricity $\lambda = \log^{1-\delta} s$ and in time $O(\lambda + \log s)$ in general.
\end{proof}

\begin{theorem}\label{thm:arb-MIS}
In a graph with maximum degree $\Delta$ and arboricity $\lambda$, 
a maximal independent set (MIS) can be computed in time on the order of
\[
\min\left\{\begin{array}{l@{\istrut[3]{0}}}
\log^2 \lambda + \log^{2/3} n,\\
\log^2\Delta + \lambda + \lambda^\epsilon\log\log n,\\
\log^2\Delta + \lambda + (\log\log n)^{1+\epsilon},\\
\log^2\Delta + \lambda^{1+\epsilon} + \log\lambda\log\log n\\
\end{array}
\right\}
\]
for all $\lambda$ and any constant $\epsilon>0$.
When $\lambda=(\log\log n)^{1/2-\delta}$,  an MIS can be computed in
$O\paren{\log^2\Delta + \f{\log\log n}{\delta\log\log\log n}}$ time.
\end{theorem}

\begin{proof}
The first bound is a consequence of Theorem~\ref{thm:arb-to-bounded-degree} and Theorem~\ref{thm:MIS}.
We can reduce the maximum degree to $\lambda t\bydef \lambda\cdot \max\{2^{(\log n)^{1/3}}, (5\lambda)^8\}$
in $O(\log_t n) = O(\log^{2/3} n)$ time, then find an MIS in the resulting graph in $O(\log^2(\lambda t) + \exp(O(\sqrt{\log\log n}))) = O(\log^2\lambda + \log^{2/3} n)$ time.

To obtain the remaining bounds we execute $\IndependentSet$ on the input graph, which, with high probability, 
returns an independent set $I$ such that the components induced by $B \bydef V(G) \backslash \hat\Gamma(I)$ 
have size at most $s = \Delta^4\log_\Delta n$. On each component we invoke one of the deterministic coloring algorithms of 
Barenboim and Elkin~\cite{BarenboimE10,BarenboimE11,BarenboimE13} for small arboricity graphs, 
then construct an MIS in time linear in the number of color classes.
For any fixed $\epsilon>0$, a $\lambda^{1+\epsilon}$-coloring can be computed in $O(\log\lambda\log s)$ time,
which gives an MIS algorithm running in time
\begin{align*}
\lefteqn{O\paren{\log^2\Delta + \lambda^{1+\epsilon} + \log\lambda\log(\Delta^4\log n)}}\\
&= O\paren{\log^2\Delta + \lambda^{1+\epsilon} + \log\lambda\log\log n},
\end{align*}
since $\lambda \le \Delta$.  Alternatively, we could use a slower $O(\lambda)$-coloring algorithm running in 
time $O(\min\{\lambda^\epsilon\log s,  \lambda^\epsilon + (\log s)^{1+\epsilon}\})$,\footnote{The leading constant in the palette size is exponential in $1/\epsilon$.} leading to an MIS algorithm running in time
\begin{align*}
\lefteqn{O\paren{\log^2\Delta + \lambda + \min\{\lambda^\epsilon\log(\Delta^4\log n), \; (\log(\Delta^4\log n))^{1+\epsilon}\}}}\\
&=
O\paren{\log^2\Delta + \lambda + \min\{\lambda^\epsilon\log\log n, \; (\log\log n)^{1+\epsilon}\}}.
\end{align*}
\end{proof}

\begin{theorem}\label{thm:arb-coloring}
Fix a constant $\epsilon>0$.
A graph of maximum degree $\Delta$ and arboricity $\lambda$ can, with high probability, be $(\Delta + \lambda^{1+\epsilon})$-colored
in $O(\log\Delta + \log\lambda\log\log n)$ time 
or $(\Delta+O(\lambda))$-colored in 
$O(\log\Delta + \min\{\lambda^\epsilon\log\log n,\; \lambda^\epsilon + (\log\log n)^{1+\epsilon}\})$ time.
Furthermore, a $(\deg+1)$-coloring can, with high probability, be computed in time on the order of
\[
\min\left\{
\begin{array}{l@{\istrut[3]{0}}}
\log\Delta + \lambda + \lambda^\epsilon \log\log n,\\
\log\Delta + \lambda + (\log\log n)^{1+\epsilon},\\
\log\Delta + \lambda^{1+\epsilon} + \log\lambda\log\log n
\end{array}
\right\}.
\]
\end{theorem}

\begin{proof}
Following the algorithm from Section~\ref{sect:Deltaplusone}, 
we first execute $O(\log\Delta)$ iterations of $\OneShotColoring$ then decompose
the problem into two subproblems on a graph with maximum degree $\hat{\Delta} \bydef \Theta(\log n)$.
On each subproblem we perform another $O(\log \hat{\Delta})$ iterations of $\OneShotColoring$,
after which the subgraph induced by uncolored nodes consists, with high probability, of
components with size at most $s = \hat{\Delta}^2\log_{\hat{\Delta}} n = o(\log^3 n)$.
At this point we apply one of the deterministic Barenboim-Elkin~\cite{BarenboimE11} coloring algorithms to each such component
using a fresh palette of $p$ previously unused colors, say $\{-1,\ldots,-p\}$.
We can find a $p$-coloring with $p = \lambda^{1+\epsilon}$ in $O(\log\lambda\log s) = O(\log\lambda\log\log n)$ time
or with $p=O(\lambda)$ in $O(\min\{\lambda^\epsilon\log s,\lambda^\epsilon + (\log s)^{1+\epsilon}\}) = O(\min\{\lambda^\epsilon\log\log n, \lambda^\epsilon + (\log\log n)^{1+\epsilon}\})$ time.
Every $v\in V(G)$ has been assigned a color
$\Color(v) \in \{1,\ldots,\deg(v)+1\} \cup \{-1,\ldots,-p\}$.  To obtain a $(\deg+1)$-coloring
we examine each color $\kappa \in \{-1,\ldots,-p\}$ in turn, letting every node $v$ with $\Color(v)=\kappa$
recolor itself using an available color from $\{1,\ldots,\deg(v)+1\}$.
\end{proof}

\begin{theorem}\label{thm:arb-rulingset} (\cite{BarenboimE10} + \cite{AwerbuchGLP89})
A $(2,O(\log \lambda + \sqrt{\log n}))$-ruling set can be computed deterministically
in $O(\log\lambda + \sqrt{\log n})$ time.
\end{theorem}

\begin{proof}
Begin by computing a decomposition of the edge set into $\lambda\cdot 2^{\sqrt{\log n}}$ oriented forests,
in $O(\sqrt{\log n})$ time~\cite[\S 3]{BarenboimE10}.  Given this decomposition, compute an $O(\lambda^2 \cdot 2^{2\sqrt{\log n}})$-coloring, in $O(\log^* n)$ time~\cite[\S 5.1.2]{BarenboimE10}.  
Finally, using this coloring,
compute a $(2,O(\log\lambda + \sqrt{\log n}))$-ruling set in $O(\log\lambda + \sqrt{\log n})$ time~\cite{AwerbuchGLP89}.
\end{proof}

\subsection{Maximal Matching in Trees}\label{sect:TreeMM}

Our maximal matching algorithm from Theorem~\ref{thm:arb-MM} runs in $O(\sqrt{\log n})$ time for every arboricity $\lambda$ 
from $1$ (trees) to $2^{O(\sqrt{\log n})}$.
We argue that this bound is optimal even for $\lambda=1$ by appealing to the KMW lower bound~\cite{KuhnMW04,KuhnMW10}.  In~\cite{KuhnMW10} 
it is shown that there exist constants $0 < 3c < c'$ such that any (possibly randomized) algorithm for 
computing an approximate minimum vertex cover (MVC) in graphs with girth
at least $c' \cdot \sqrt{\log n}$ 
either (i) runs in $c \sqrt{\log n}$ time, or (ii) has approximation ratio $\omega(1)$.
We review below a well known reduction from 2-approximate MVC to maximal matching,
which implies an $\Omega(\sqrt{\log n})$ lower bound for maximal matching algorithms that succeed with high probability.
The graphs used in the KMW bound have arboricity $2^{O(\sqrt{\log n})}$, so it does directly 
imply an $\Omega(\sqrt{\log n})$ lower bound on trees.

\begin{theorem} \label{thm:TreeMMlb}
For some absolute constant $c>0$, no algorithm can, with probability $1-n^{-2}$, 
compute a maximal matching on a tree 
in $c\sqrt{\log n}$ time, nor in
$c\log\Delta + o(\sqrt{\log n})$ time for every $\Delta$.
\end{theorem}

\begin{proof}
We first recount the lower bound for maximal matching on general graphs.
Suppose, for the purpose of obtaining a contradiction, that there exists a maximal matching algorithm running in time $c\sqrt{\log n}$ on the KWM graph that fails with probability at most $1/n$. 
To obtain an approximate MVC algorithm, run the maximal matching algorithm for $c\sqrt{\log n}$ time.  Any matched node joins the approximate MVC, as well as any node that detects a local violation,
namely a node incident to another unmatched node.
As the MVC is at least the size of any maximal matching, 
the expected approximation ratio of this algorithm is at most $2 \cdot \Pr(\mbox{no failure occurs}) + n \cdot \Pr(\mbox{some failure occurs}) \leq 2 + n \cdot \frac{1}{n} = 3$,
a contradiction.  Hence there is no algorithm that runs for $c \cdot \sqrt{\log n}$
time in graphs with girth at least $c' \cdot \sqrt{\log n}$ that computes a maximal matching 
with probability at least $1 - 1/n$.

We use an indistinguishability argument to show that the 
$\Omega(\sqrt{\log n})$ lower bound also holds for trees, and therefore any class of graphs that includes trees.
Observe that to show a lower bound for a randomized algorithm, it is enough to prove the same lower bound under the assumption that the identities of 
graph nodes were selected independently and uniformly at random, from, say, $[1,n^{10}]$.  
Suppose there is, in fact, an algorithm that given a tree with a random (in the above sense) assignment of 
identities, constructs a maximal matching within $c \cdot \sqrt{\log n}$ time with success probability at least $1 - n^{-2}$. Run this algorithm for $c \cdot \sqrt{\log n}$ 
time on the KMW graph $G$ with girth $c' \cdot \sqrt{\log n}$, assuming random assignment of identities in $G$. 
Due to the girth bound, the view of every node in $G$ is identical to its view in a tree, and thus from its perspective a correct maximal matching must be computed with probability at least $1 - n^{-2}$.  By a union bound, a correct maximal matching for the entire graph $G$ will be computed with probability at least $1 - n^{-1}$, a contradiction.

The KMW graph has maximum degree $\Delta = 2^{\Theta(\sqrt{\log n})}$ and girth $\Theta(\log \Delta)$.
All the KMW-based $\Omega(\sqrt{\log n})$ lower bounds can be scaled down to 
$\Omega(\log\Delta)$ lower bounds (for any $\Delta < 2^{O(\sqrt{\log n})}$) 
simply by applying the lower bound argument to the union of numerous identical KMW graphs.
\end{proof}

\begin{remark}
Theorem~\ref{thm:TreeMMlb} posited the existence of a maximal matching algorithm for trees whose {\em global} probability of failure is $n^{-2}$.  When we run this algorithm on the KMW graph we can no longer use $n^{-2}$ as the global failure probability.  It may be that, when run in an actual tree, nodes within distance $c\sqrt{\log n}$ of a leaf node fail with probability zero: all the failure probability is concentrated at the small set of nodes that cannot ``see'' the leaves.  In the KMW graph {\em all} nodes think they are in this small set. We must assume, pessimistically, that 
failure occurs at {\em every} node in the KMW graph with probability $n^{-2}$.
\end{remark}

\begin{remark}
Theorem~\ref{thm:TreeMMlb}, strangely, does not imply any lower bound for the MIS problem on trees, even though MIS appears to be just as hard as maximal matching on any class of graphs.  
The $\Omega(\sqrt{\log n})$ lower bound for MIS~\cite{KuhnMW04,KuhnMW10} is obtained by considering the 
{\em line graph} of the KWM graph, which has girth 3, not $\Theta(\sqrt{\log n})$.  Thus, our indistinguishability argument does not apply. 
\end{remark}

\section{MIS in Trees and High Girth Graphs}\label{sect:TreeMIS}

One of the MIS algorithms of Luby~\cite{Luby86} works as follows.
In each round
each remaining node $v$ chooses a random real $r(v) \in (0,1)$ and includes itself in the MIS if $r(v)$ is greater than $\max_{w\in \Gamma(v)} r(w)$, thereby {\em eliminating} $v$ and its neighborhood from further consideration.\footnote{In practice it suffices to generate only the $O(\log n)$ most significant bits.  That is, nodes choose an integer from, say, $\{1,\ldots,n^{10}\}$ uniformly at random.} 
Observe that the probability that $v$ joins the MIS
in a round is $1/(\deg(v)+1)$, irrespective of the degrees of its neighbors.

We would like to say that degrees decay geometrically, 
that is, after $O(k)$ iterations of Luby's 
algorithm the maximum degree is $\Delta/2^k$, with high probability.
Invariant~\ref{inv:Tree} is not quite this strong but 
just as useful, algorithmically.  It states that
after $O(k\log\log\Delta)$ iterations, no node has
$\Delta/2^{k+2}$ neighbors with degree at least $\Delta/2^k$,
provided that $\Delta/2^k$ is not too small.  

\begin{invariant}\label{inv:Tree}
At the end of scale $k$, for all $v\in \VIB$,
\[
\left|\left\{w\in \GammaIB(v) \;|\; \degIB(w) > \Delta/2^k \right\}\right| \;\le\; \max\{\Delta/2^{k+2}, 12\ln \Delta\}.
\]
\end{invariant}

Randomness plays no role in Invariant~\ref{inv:Tree}: it holds with probability 1. 
Any node that violates the invariant is marked {\em bad} (placed in $B$) and temporarily
excluded from consideration.  As we will soon prove, the 
probability a node is marked bad is $1/\poly(\Delta)$.
We will only make use of Invariant~\ref{inv:Tree} when $\Delta/2^{k+2}$ is, in fact, greater than $12\ln\Delta$,
so the $12\ln\Delta$ term will not be mentioned until we need to have a lower bound on $\Delta/2^{k+2}$.

\begin{figure}
\centering
\framebox{\hcm[.1]
\begin{minipage}{6in}
$\TreeIndependentSet(\mbox{Graph } G)$
\begin{enumerate}
	\item Initialize sets $I,B\subset V(G)$:\\
	$
	\begin{array}{rlll}
		I &\leftarrow \emptyset	&\hcm[.4]& \mbox{\{an independent set\}}\\	
		B &\leftarrow \emptyset	&& \mbox{\{a set of `bad' nodes\}}
	\end{array}
	$
	
	Throughout, let $\VIB \bydef V(G) \backslash (\hat\Gamma(I) \cup B)$ be the nodes still under consideration.
	Let $\GIB$ be the graph induced by $\VIB$ and let $\GammaIB$ and $\degIB$ be the neighborhood and degree functions w.r.t.~$\GIB$.
	\item For each scale $k$ from 1 to $\log\paren{\fr{\Delta}{48\ln\Delta}}$,
	
	\begin{enumerate}
		\item Execute $\log_{5/4}(33\ln\Delta)$ iterations of steps i and ii.
		\begin{enumerate}
			\item Each node $v\in \VIB$ chooses a priority $r(v)$.\istrut[2]{0}
			\item[]
				
	$r(v) \leftarrow \left\{\begin{array}{l@{\hcm[.2]}l} 
			0,							&  \mbox{\istrut[3]{0}if $\left|\{w\in \GammaIB(v) \;|\; \degIB(w) > \Delta/2^{k}\}\right|$}\\
										& \mbox{\hcm[2.4] $> \Delta(8\ln\Delta+1)/2^{k+1}$,}\\
			 \mbox{a random real in $(0,1)$, \hcm[.3]}	& \mbox{otherwise.}
			\end{array}\right.$\\

			\item $I \leftarrow I \cup \{v\in \VIB \;|\; r(v) > \max\{r(w) \;|\; w\in \GammaIB(v)\}$\\
				{\em (Add nodes to the independent set.)}
		\end{enumerate}
		\item $B\leftarrow B \cup \left\{v\in \VIB \;\;\big|\;\;
									|\{w\in \GammaIB(v) \;|\; \degIB(w) > \Delta/2^{k}\}|     >     \Delta/2^{k+2}\right\}$.\\
			{\em (Mark nodes that violate Invariant~\ref{inv:Tree} as bad.)}
	\end{enumerate}
	
	\item Return $(I, B)$.
\end{enumerate}
\end{minipage}
}
\caption{\label{alg:TreeIndependentSet}}
\end{figure}

\begin{lemma}\label{lem:Tree-elimination-prob}
In one iteration of scale $k$, a node $w$ with $\degIB(w) > \Delta/2^k$ is eliminated (appears in $\hat\Gamma(I)$) 
with probability at least $(1-o(1))(1-e^{-1/4}) > 0.22$.  Moreover, this probability holds even if we condition on
arbitrary behavior at a single neighbor of $w$.
\end{lemma}

\begin{proof}
By Invariant~\ref{inv:Tree}, $|\{x\in\GammaIB(w) \;|\; \degIB(x) > \Delta/2^{k-1}\}| \le \Delta/2^{k+1}$.
Let $M$ be the neighbors of $w$ with degree at most $\Delta/2^{k-1}$, so $|M| \ge \degIB(w) - \Delta/2^{k+1} > \Delta/2^{k+1}$.
Refer to the portion of Figure~\ref{fig:TreeMIS} depicting $w$ and its neighborhood.
The probability that $w$ is eliminated is minimized when $M$-nodes attain their maximum degree $\Delta/2^{k-1}$,
so in the calculations below we shall assume this is the case.
Let $x^\star \in M$ be the first neighbor for which $r(x^\star) > \max\{r(y) \;|\; y\in \GammaIB(x^\star)\backslash\{w\}\}$.  
The probability $x^\star$ exists is at least
\[
\Pr(x^\star \mbox{ exists}) = 1 - \prod_{x\in M} \paren{1-\frac{1}{\degIB(x)}} \ge 1 - \paren{1-\frac{1}{\Delta/2^{k-1}}}^{\Delta/2^{k+1}} > 1-e^{-1/4}.
\]
Since, in the most extreme case, $\degIB(x) = \Delta/2^{k-1}$,
$\Pr(\mbox{$x^\star$ joins $I$} \;|\; \mbox{$x^\star$ exists}) = \Pr(r(x^\star) > r(w) \;|\; \mbox{$x^\star$ exists}) \ge 1-\frac{1}{\Delta/2^{k-1} + 1}$.
The probability that $w$ is eliminated is therefore at least 
$(1-\frac{1}{\Delta/2^{k-1}+1})(1 - e^{-1/4}) > (1-\frac{1}{96\ln \Delta})(1-e^{1/4}) > 0.22 > 1/5$.
Moreover, this probability is perturbed by a negligible 
$(1-\Theta(1/\Delta/2^k)) = (1-o(1))$ factor if one conditions on arbitrary behavior by a single neighbor of $w$.
\end{proof}

\begin{figure}
\centering
\scalebox{.38}{\includegraphics{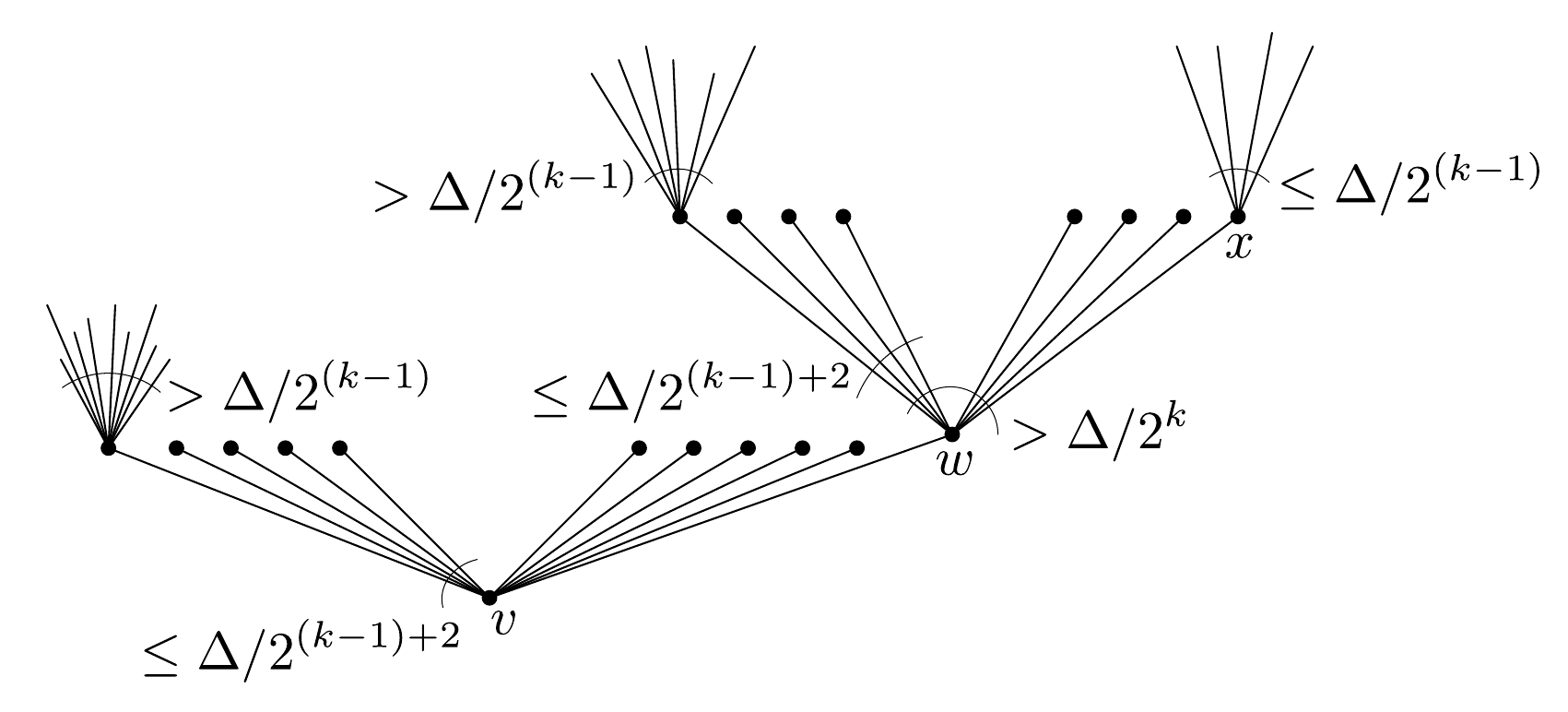}}
\caption{\label{fig:TreeMIS} The $k$th scale of $\TreeIndependentSet$, from the perspective of $v$.
Only $v$'s neighbors with degree greater than $\Delta/2^k$ are shown; $w$ is one such neighbor.
They are partitioned into those with degrees in $(\Delta/2^{k-1},\infty)$ and $(\Delta/2^k,\Delta/2^{k-1}]$.
The first category numbers at most $\Delta/2^{(k-1)+2}$; the second category is unbounded.  
At most $\Delta/2^{(k-1)+2}$ of $w$'s neighbors have degree more than $\Delta/2^{(k-1)}$, leaving at least half with degree at most $\Delta/2^{(k-1)}$.
If any neighbor $x$ joins the MIS, $w$ will be eliminated.
}
\end{figure}

\begin{lemma}\label{lem:TreeMIS-badnode}
In any scale, a node $v$ 
is included in $B$ with probability at most $1/\Delta^2$, independent of the behavior of any one neighbor of $v$.
\end{lemma}

\begin{proof}
Fix a node $v$ and let $N = \{w\in\GammaIB(v) \;|\; \degIB(w) > \Delta/2^k\}$ at the beginning of scale $k$.  
See Figure~\ref{fig:TreeMIS}.  In the figure, only $N$-node neighbors of $v$ are depicted.
If $|N| \le \Delta/2^{k+2}$ then the invariant is already satisfied at $v$, so assume otherwise.
There are two cases, depending on the size of $N$.
\paragraph{Case 1: $|N|$ is large}
We argue that if $|N|>\Delta(8\ln\Delta+1)/2^{k+1}$, then $v$ is eliminated with probability at 
least $1-\Delta^{-2}$ in a single iteration, and can therefore be bad with probability at most $\Delta^{-2}$.  
According to the algorithm, $r(v)=0$, so $v$ has no chance to hold a locally maximum $r$-value.
Since, by Invariant~\ref{inv:Tree}, $v$ has at least $|N| - \Delta/2^{k+1} > 8\Delta\ln\Delta/2^{k+1}$ neighbors with degree
at most $\Delta/2^{k-1}$, the probability that $v$ is not eliminated is at most the probability that no $N$-node joins $I$.
This occurs with probability at most
\[
\paren{1-\frac{1}{\Delta/2^{k-1}}}^{|N|-\Delta/2^{k+1}} \le \exp\paren{\f{8\Delta\ln\Delta}{2^{k+1}}\cdot \frac{2^{k-1}}{\Delta}} = \Delta^{-2}.
\]

\paragraph{Case 2: $|N|$ is small} 
In this case $|N|\le \Delta(8\ln\Delta+1)/2^{k+1}$, that is, $|N|$ is within a $O(\log\Delta)$ factor of satisfying Invariant~\ref{inv:Tree}.
By Lemma~\ref{lem:Tree-elimination-prob} each $N$-node, so long as it has degree at least $\Delta/2^k$, 
is eliminated with probability at least $1/5$.  Moreover, these events are independent, conditioned
upon some arbitrary behavior at $v$, the only common neighbor of $N$-nodes.
Thus, each node will survive $\log_{5/4}(4(8\ln \Delta + 1)) = O(\log\log\Delta)$
iterations with probability $1/[4(8\ln\Delta+1)]$. The expected number 
of surviving $N$-nodes is therefore less than $\Delta/2^{k+3}$.  
By a Chernoff bound (Theorem~\ref{thm:stdChernoff}), the probability that 
this quantity exceeds twice its expectation, thereby putting $v$ into $B$, is 
$\exp(-(\Delta/2^{k+3})/3)$, which is at most $\Delta^{-2}$ since 
$\Delta/2^{k} \ge 48\ln\Delta$.
\end{proof}

\begin{lemma}\label{lem:Tree-bad-component-size}
All connected components in the subgraph induced by $B$ have at most 
$t = c\log_\Delta n$ nodes with probability $1-n^{-c/2}$.
\end{lemma}

\begin{proof}
There are less than $4^t$ topologically distinct rooted $t$-node trees and at most $n\Delta^{t-1}$ ways to 
embed such a tree, say $T$, in the graph.  There are $(\log \Delta)^t$ schedules for when (in which scale) 
the $T$-nodes were added to $B$.  Since the probability that each $T$-node becomes bad in a scale 
is at most $\Delta^{-2}$, {\em independent of the behavior of its parent in $T$},
the probability that $B$ contains a $t$-node tree is at most
\begin{align*}
\lefteqn{4^t \cdot n\Delta^{t-1} \cdot (\log\Delta)^t \cdot \Delta^{-2 t} }\\
& < (4\log\Delta)^{c\log_\Delta n} \cdot n^{c+1} \cdot n^{-2c}\\
& < n^{-c/2}.		
\end{align*}
The last inequality holds when $\Delta$ is at least some sufficiently large constant.
\end{proof}

\subsection{The $\TreeMIS$ Algorithm}

Let us review the situation.  $\TreeIndependentSet(G)$ takes $O(\log\Delta\log\log \Delta)$ time
and returns a pair $(I,B)$ satisfying two properties, the second of which holds with 
probability $1-n^{-c/2}$.

\begin{itemize}
\item Although the degree of nodes in the graph induced by $\VIB = V(G)\backslash (\hat\Gamma(I) \cup B)$
is not bounded, no node has $12\ln\Delta$ neighbors with degree at least $48\ln\Delta$.
\item The graph induced by $B$ is composed of small connected components, 
each with size at most $t\le c\log_\Delta n$.
\end{itemize}

The $\TreeMIS$ algorithm (Figure~\ref{alg:TreeMIS}) starts by obtaining a pair $(I,B)$ satisfying these properties, 
then extends $I$ to a maximal independent set in three stages.  It partitions $\VIB$ into low and high degree sets $V_{\lo}$ and $V_{\hi}$.  By definition the graph induced by $V_{\lo}$ has maximum 
degree $48\ln\Delta$ and by the first property above
the graph induced by $V_{\hi}$ has maximum degree $12\ln\Delta$.  An MIS $I_{\lo}$ for $V_{\lo}$ can be computed
deterministically in $O(\log\Delta + \log^* n)$ time~\cite{BarenboimEK14}, that is, in time linear in the degree.\footnote{Since we are already spending $O(\log\Delta\log\log\Delta)$ time in $\TreeIndependentSet$, we can afford to 
use a simpler MIS algorithm~\cite{KuhnW06} running in $O(\log\Delta\log\log\Delta + \log^* n)$ time.}
An MIS $I_{\hi}$ for $V_{\hi}\backslash \hat\Gamma(I_{\lo})$ can then be computed, also in $O(\log\Delta + \log^* n)$ time.
At this point only $B$-nodes may not be adjacent to some node in $I\cup I_{\lo} \cup I_{\hi}$.  
For each component $C$ in the graph induced by $B\backslash \hat\Gamma(I\cup I_{\lo} \cup I_{\hi})$
we compute an MIS $I_C$ in $O(\log t/\log\log t) = O(\frac{\log\log n}{\log\log\log n})$ time using Barenboim-Elkin~\cite{BarenboimE10} algorithm.

In total the running time of $\TreeMIS$ is $O(\log\Delta\log\log\Delta + \frac{\log\log n}{\log\log\log n})$ and its failure probability is less than $n^{-c/2}$.

\begin{figure}
\centering
\framebox{\hcm[.1]
\begin{minipage}{4.6in}
$\TreeMIS(\mbox{Graph } G)$
\begin{enumerate}
	\item[] {\bf Phase I:}
	\item $(I,B) \leftarrow \TreeIndependentSet(G)$
	\item[] {\bf Phase II:}
	\item Partition $\VIB = V(G) \backslash (\hat\Gamma(I) \cup B)$ into low- and high-degree sets.
	\item[] $V_{\lo} \leftarrow \{v\in \VIB \;|\; \degIB(v) \le 48\ln \Delta\}$
	\item[] $V_{\hi} \leftarrow \{v\in \VIB \;|\; \degIB(v) > 48\ln \Delta\}$
	\item Compute maximal independent sets on $V_{\lo}$ and $V_{\hi}$. 
	\item[] $I_{\lo} \leftarrow $ an MIS of the graph induced by $V_{\lo}$.
	\item[] $I_{\hi} \leftarrow $ an MIS of the graph induced by $V_{\hi}\backslash \hat\Gamma(I_{\lo})$.\\
	\item[] Let $\mathscr{C}$ be the set of connected components with size at 
	most $c\ln n$ in the graph induced by $B \backslash \hat\Gamma(I \cup I_{\lo} \cup I_{\hi})$.
	\item For each $C\in\mathscr{C}$,	
	\item[] $I_C \leftarrow $ an MIS of $C$
	\item Return $\displaystyle I \cup I_{\lo} \cup I_{\hi} \cup \bigcup_{C\in\mathscr{C}} I_C$
\end{enumerate}
\end{minipage}
}
\caption{\label{alg:TreeMIS}}
\end{figure}

\begin{theorem}\label{thm:TreeMIS}
In an unoriented tree with maximum degree $\Delta$, a maximal independent set
can, with high probability, be computed in time on the order of 
\[
\min\left\{\; \log\Delta\log\log\Delta + \frac{\log\log n}{\log\log\log n}, \; \; \; \sqrt{\log n\log\log n}\; \right\}.
\]
\end{theorem}

\begin{proof}
The $O(\log\Delta\log\log\Delta + \frac{\log\log n}{\log\log\log n})$ bound was shown above.
If $\Delta > \hat{\Delta} \bydef 2^{\sqrt{\log n/\log\log n}}$, use Theorem~\ref{thm:arb-to-bounded-degree} to reduce the maximum degree 
to $\hat{\Delta}$ in $O(\log_{\hat{\Delta}} n) = O(\sqrt{\log n\log\log n})$ time,
then compute an MIS in $O(\log\hat{\Delta}\log\log\hat{\Delta} + \frac{\log\log n}{\log\log\log n}) = O(\sqrt{\log n\log\log n})$ time.
\end{proof}

\subsection{MIS on High Girth Graphs}

Our analysis of $\TreeIndependentSet$ and $\TreeMIS$ requires that certain events are independent and this independence is guaranteed if the radius-3 neighborhood around each node looks like a tree.  In other words, parts of the analysis do not distinguish between actual trees and graphs with girth greater than 6.\footnote{The analysis could probably be made to work for graphs with girth 6 or 5, but it does {\em not} work for graphs of girth 4.  If the graph is formed by grafting together a sequence of bipartite 
$\Delta/2 \times \Delta/2$ cliques, the probability a node becomes bad after one scale
of $\TreeIndependentSet$ is not $1/\poly(\Delta)$ but $\exp(-\Omega((\log\log\Delta)^2/\log\log\log\Delta))$.}

In order to make the analysis work on graphs with girth greater than 6 we need to make a number of small modifications to $\TreeIndependentSet$ and $\TreeMIS$.
\begin{itemize}
\item We substitute $\log n$ for $\log\Delta$ in Invariant~\ref{inv:Tree}.  It now states that at the end of scale $k$,
for all $v\in \VIB$,
\[
\left|\left\{w\in \GammaIB(v) \;|\; \degIB(w) > \Delta/2^k \right\}\right| \;\le\; \max\{\Delta/2^{k+2}, c\ln n\}
\]
for some sufficiently large $c$.  

\item We change the critical threshold in $\TreeIndependentSet$ from $\Delta(8\ln\Delta+1)/2^{k+1}$ to $\Delta(8\ln n + 1)/2^{k+1}$
and change the number of iterations per scale from $O(\log\log \Delta)$ to $O(\log\log n)$.

\item Lemmas~\ref{lem:TreeMIS-badnode} and \ref{lem:Tree-bad-component-size} now claim that after $\log(\Delta/(4c\ln n))$ scales, 
\begin{itemize}
\item In $\GIB$, each node has no more than $c\ln n$ neighbors with degree greater than $4c\ln n$.
\item With high probability, namely $1-n^{-\Omega(c)}$, all nodes satisfy Invariant~\ref{inv:Tree}.  That is, $B = \emptyset$.
\end{itemize}

\item Provided that $B=\emptyset$, in order to extend $I$ to an MIS we only need to find an MIS $I_{\lo}$ of $V_{\lo}$
and $I_{\hi}$ of $V_{\lo}\backslash \hat\Gamma(I_{\hi})$.  Since the graphs induced by $V_{\lo}$ and $V_{\hi}$ have maximum degree $4c\ln n$, this takes $\exp(O(\sqrt{\log\log n}))$ time 
using the $\MIS$ algorithm of Section~\ref{sect:MIS}.
\end{itemize}

\begin{theorem}\label{thm:MIS-girthsix}
In a graph of girth greater than 6 (in which no cycle has length at most 6), an MIS can be computed
in $O(\log\Delta\log\log n + \exp(O(\sqrt{\log\log n})))$ time with high probability.
\end{theorem}

\section{Conclusions}\label{sect:conclusion}

In this work we have advanced the state-of-the-art in randomized symmetry breaking
using a powerful new set of algorithmic tools.
Our MIS and maximal matching algorithms represent the first significant improvements (for general graphs) 
to the classic algorithms of the 1980s~\cite{Luby86,ABI86,II86}.  Our maximal matching algorithms
(for general graphs, trees, and low-arboricity graphs) are among a small group of 
{\em provably} optimal symmetry breaking algorithms for a wide range of parameters.
However, we feel the most important contribution of this work is the identification of
the {\em union bound barrier} and the development of several tools for confronting it.

All of our algorithms reduce an $n$-node instance of the problem to a disjoint set
of $\poly(\log n)$-node components\footnote{(or in the case of the MIS algorithm, $(\poly(\Delta)\log n)$-size components)}, which is the threshold beyond which known randomized symmetry breaking strategies
fail to achieve a $(\log n)^{o(1)}$ running time.  Even if the probability of failure on one component is small,
by the union bound the probability of failure on {\em some} component is nearly certain.  
Unless, of course, the probability of failure is {\em zero}, meaning we forswear random bits altogether
and opt to use the best available deterministic algorithm.  We conjecture that this is essentially 
the only way to confront the union bound barrier.  If true, this means that the randomized complexities of many 
symmetry breaking problems are tethered to their deterministic counterparts.
For example, we could not hope to get rid of the $2^{O(\sqrt{\log\log n})}$ terms in our MIS and coloring algorithms without 
{\em first} improving the $2^{O(\sqrt{\log n})}$-time Panconesi-Srinivasan~\cite{PanconesiS96} algorithms.  
We also could not hope to achieve an (optimal) 
$O(\min\{\log\Delta,\sqrt{\log n}\} + \log^* n)$-time algorithm for MIS or maximal matching
unless that algorithm were deterministic.  

We leave many problems open, some of which are accessible and some quite hard.
The most difficult problem is to find optimal $O(\min\{\log\Delta, \sqrt{\log n}\} + \log^* n)$-time algorithms for MIS and maximal matching,
or, as a first step, {\em any} $o(\log n)$ time algorithm.  An easier problem is to find an
$O(\min\{\log\Delta,\sqrt{\log n}\} + \log\log n)$-time MIS algorithm for bounded arboricity graphs, or even trees.
The complexity of the $(\Delta+1)$-coloring problem is the least understood.  Expressed in terms of $n$,
the best known upper bound is $O(\log n)$~\cite{Luby86,Johansson99} and
best known lower bound $\Omega(\log^* n - \log^*\Delta)$~\cite{Linial92}.  Is there an algorithm running in $o(\log n)$ time?  

\paragraph{Acknowledgement.} 
We would like to thank James Hegeman and Sriram Pemmaraju for pointing out a flaw in an earlier
proof of Lemma~\ref{lem:highdegreeneighbors}.

\ignore{
\bibliographystyle{ACM-Reference-Format-Journals}
\bibliography{../../references}
}



\clearpage
\appendix
\centerline{\LARGE\bf Appendix}

\section{Concentration Inequalities}\label{sec:concentration-bounds}

See Dubhashi and Panconesi~\cite{DubhashiPanconesi09} for proofs of these and related concentration bounds.

\begin{theorem}\label{thm:stdChernoff} {\bf (Chernoff)}
Let $X$ be the sum of $n$ independent, identically distributed 0/1 random variables.
For any $\delta \in (0,1)$,
\begin{align*}
			\Pr(X < (1-\delta)\E[X]) &< \exp\paren{-\delta^2\E[X]/2}\\
\mbox{ and } \; \Pr(X > (1+\delta)\E[X]) &< \exp\paren{-\delta^2\E[X]/3}.
\end{align*}
\end{theorem}

\begin{theorem}\label{thm:Chernoff} {\bf (Negative correlation)}
Let $X = X_1 +\cdots + X_n$ be the sum of $n$ random variables, 
where the $\{X_i\}$ are independent or negatively correlated.  Then for any $t>0$:
\[
\Pr(X \ge \E[X] + t), \Pr(X \le \E[X] - t) \; \le \; \exp\paren{- \f{2t^2}{\sum_{i}(a_i'-a_i)^2}},
\]
where $a_i \le X_i \le a_i'$.
\end{theorem}

\begin{theorem}\label{thm:Janson} {\bf (Janson)}
For $X = X_1+\cdots+X_n$ the sum of $n$ random variables and $t>0$,
\[
\Pr(X \ge \E[X] + t), \Pr(X \le \E[X] - t) \; \le \; \exp\paren{- \f{2t^2}{\chi\cdot \sum_{i}(a_i'-a_i)^2}},
\]
where $a_i \le X_i \le a_i'$ and $\chi$ is the fractional chromatic number of the dependency graph 
$\mathscr{G} = (\mathscr{V},\mathscr{E})$.  By definition $\mathscr{V} = \{X_1,\ldots,X_n\}$ 
and the edge set $\mathscr{E}$ satisfies the property that
$X_i$ is independent of $\mathscr{V} \backslash \Gamma(X_i)$, for all $i$.
\end{theorem}

\begin{theorem}\label{thm:Azuma-Hoeffding} {\bf (Azuma-Hoeffding)}
A sequence $Y_0,\ldots,Y_n$ is a martingale with respect to $X_0,\ldots,X_n$
if $Y_i$ is a function of $X_0,\ldots,X_i$ and $\E[Y_i \:|\: X_0,\ldots,X_{i-1}] = Y_{i-1}$.
For such a martingale with bounded differences $a_i \le Y_i - Y_{i-1} \le a_i'$,
\[
\Pr(Y_n > Y_0 + t), \, \Pr(Y_n < Y_0 - t) \le \exp\paren{-\f{t^2}{2\sum_i (a_i'-a_i)^2}}.
\]
\end{theorem}

\begin{corollary}\label{cor:Azuma-Hoeffding}
Let $Z = Z_1+\cdots+Z_n$ be the sum of $n$ random variables and $X_0,\ldots,X_n$ be a sequence,
where $Z_i$ is uniquely determined by $X_0,\ldots,X_i$, 
$\mu_i = \E[Z_i \;|\; X_0,\ldots,X_{i-1}]$, $\mu=\sum_i\mu_i$, and $a_i \le Z_i \le a_i'$.  Then
\[
\Pr(Z > \mu + t), \, \Pr(Z < \mu-t) \le \exp\paren{-\f{t^2}{2\sum_i (a_i'-a_i)^2}}.
\]
\end{corollary}
\begin{proof}
Define the martingale $Y_0,\ldots,Y_n$ w.r.t.~$X_0,\ldots,X_n$ by $Y_0 = 0$
and $Y_i = Y_{i-1} + Z_i - \mu_i$, then apply Theorem~\ref{thm:Azuma-Hoeffding}.
Note $Y_n - Y_0 = Z-\mu$ and the range of $Y_i-Y_{i-1}$ still has size $a_i'-a_i$.
\end{proof}

Note that Corollary~\ref{cor:Azuma-Hoeffding} says that one random variable, $Z$, 
is well concentrated around another random variable, namely $\mu$.

\end{document}